\documentclass[11pt,onecolumn]{article}

\usepackage{amsmath,amsthm,amssymb,amsfonts,epsfig,graphicx}
\usepackage[linesnumbered,ruled,vlined]{algorithm2e}
\usepackage{color}
\usepackage{dsfont}
\usepackage{enumerate}
\usepackage{url}
\usepackage{bm}
\usepackage{enumitem}
\usepackage[left=0.9in,top=0.9in,right=0.9in,bottom=0.5in,nohead,textheight=10in,footskip=0.3in]{geometry}

\let\oldnl\nl% Store \nl in \oldnl
\newcommand{\nonl}{\renewcommand{\nl}{\let\nl\oldnl}}% Remove line number for one line

\newtheorem*{rep@theorem}{\rep@title}
\newcommand{\newreptheorem}[2]{%
\newenvironment{rep#1}[1]{%
 \def\rep@title{#2 \ref{##1}}%
 \begin{rep@theorem}}%
 {\end{rep@theorem}}}
\makeatother

\usepackage{setspace}
\let\OLDthebibliography\thebibliography
\renewcommand\thebibliography[1]{
  \OLDthebibliography{#1}
  \setlength{\itemsep}{0pt}
}

\newcommand{\lft}{{\tt left}}
\newcommand{\rgt}{{\tt right}}
\newtheorem{theorem}{Theorem}
\newtheorem{corollary}[theorem]{Corollary}

\newtheorem{lemma}[theorem]{Lemma}

\newtheorem{proposition}[theorem]{Proposition}
\newtheorem{observation}[theorem]{Observation}

\newreptheorem{proposition}{Proposition}
\newreptheorem{theorem}{Theorem}

\DeclareMathOperator*{\argmax}{arg\,max}

\DeclareMathOperator*{\polylog}{polylog}
\DeclareMathOperator*{\poly}{poly}

\begin{document}
\renewcommand{\d}[1]{\ensuremath{\operatorname{d}\!{#1}}}
\def\myparagraph#1{\vspace{2pt}\noindent{ #1~~}}
\newcommand{\eqdef}{{\stackrel{\mbox{\tiny \tt ~def~}}{=}}}

%\pagestyle{headings}

%%%%%%%%%%%%%%%%%%%%%%%%%%%%%%%%%%%%%%%%%%%%%%%%%%%%%%%%%%%
%%%%%%%%%%%%%%%%%%%%%%%%%%%%%%%%%%%%%%%%%%%%%%%%%%%%%%%%%%%
%%%%%%%%%%%%%%%%%%%%%%%%%%%%%%%%%%%%%%%%%%%%%%%%%%%%%%%%%%%

\long\def\ignore#1{}
%\epsfverbosetrue
%\def\myps[#1]#2{}
\def\myps[#1]#2{\includegraphics[#1]{#2}}
\def\etal{{\em et al.}}
\def\Bar#1{{\bar #1}}
\def\br(#1,#2){{\langle #1,#2 \rangle}}
\def\setZ[#1,#2]{{[ #1 .. #2 ]}}
\def\Pr#1{{{\mathbb P}\!\bm{\left[} #1 \bm{\right]}}}

\def\maxsuperscript[#1]{\raisebox{#1}{\scriptsize $({\max})$}}

\def\bmin{{\beta_{\min}}}
\def\bmax{{\beta_{\max}}}
\def\betamin{\beta_{\min}}
\def\betamax{\beta_{\max}}

\def\bP{{\bf P}}

\def\E{{\mathbb E}}
\def\P{{\mathbb P}}
\def\X{{\mathbb X}}
\def\Y{{\mathbb Y}}
\def\Z{{\mathbb Z}}
\def\W{{\mathbb W}}
\def\U{{\mathbb U}}
\def\S{{\mathbb S}}
\def\R{{\mathbb R}}

\def\Vrel(#1){{{\mathbb S}[{#1}]}}

\def\Pcoef#1{{P^{#1}_{\tt count}}}
\def\Ptcoef#1{{P^{#1}_{\tt count-weak}}}
\def\Pratio#1{{P^{{#1}}_{\tt ratio-point}}}
\def\PratioAll#1{{P^{{#1}}_{\tt ratio-all}}}

\def\setof#1{{\left\{#1\right\}}}
\def\suchthat#1#2{\setof{\,#1\mid#2\,}} % so says Knuth, page 174
\def\event#1{\setof{#1}}
\def\q={\quad=\quad}
\def\qq={\qquad=\qquad}
\def\calA{{\cal A}}
\def\calB{{\cal B}}
\def\calC{{\cal C}}
\def\calD{{\cal D}}
\def\calE{{\cal E}}
\def\calF{{\cal F}}
\def\calG{{\cal G}}
\def\calI{{\cal I}}
\def\calJ{{\cal J}}
\def\calH{{\cal H}}
\def\calHfrac{{{\cal L}}}
\def\calL{{\cal L}}
\def\calM{{\cal M}}
\def\calN{{\cal N}}
\def\calK{{\cal K}}
\def\calP{{\cal P}}
\def\calQ{{\cal Q}}
\def\calR{{\cal R}}
\def\calS{{\cal S}}
\def\A{\mathfrak A}
\def\calT{{\cal T}}
\def\calU{{\cal U}}
\def\calV{{\cal V}}
\def\calO{{\cal O}}
\def\calX{{\cal X}}
\def\calY{{\cal Y}}
\def\calZ{{\cal Z}}
\def\psfile[#1]#2{}
\def\psfilehere[#1]#2{}

\newcommand{\eps}{\varepsilon}
\newcommand\myqed{{}}

%%%%%%%%%%%%%%%%%%%%%%%%%%%%%%%%%%%%%%%%%%%%%%%%%%%%%%%%%%%
%%%%%%%%%%%%%%%%%%%%%%%%%%%%%%%%%%%%%%%%%%%%%%%%%%%%%%%%%%%
%%%%%%%%%%%%%%%%%%%%%%%%%%%%%%%%%%%%%%%%%%%%%%%%%%%%%%%%%%%

\title{\Large\bf  \vspace{-30pt} Parameter estimation for Gibbs distributions}

\author{David G. Harris \\ \normalsize University of Maryland, Department of Computer Science\\ {\normalsize\tt davidgharris29@gmail.com}
 \and Vladimir Kolmogorov \\ \normalsize Institute of Science and Technology Austria \\ {\normalsize\tt vnk@ist.ac.at}}
\date{}
\maketitle

%\vspace{-6pt}
\begin{abstract}
A central problem in computational statistics is to convert a procedure for \emph{sampling} combinatorial  objects into a procedure for \emph{counting} those objects, and vice versa. We consider sampling problems coming from \emph{Gibbs distributions}, which are families of probability distributions over a discrete space $\Omega$ with probability mass function of the form $\mu^\Omega_\beta(\omega) \propto e^{\beta H(\omega)}$ for $\beta$ in an interval $[\betamin, \betamax]$ and $H( \omega ) \in \{0 \} \cup [1, n]$.   Two important parameters are the \emph{partition function}, which is the normalization factor $Z(\beta)=\sum_{\omega \in\Omega}e^{\beta H(\omega)}$, and the vector of preimage \emph{counts} $c_x = |H^{-1}(x)|$.

 We develop black-box sampling algorithms to estimate the counts using roughly $\tilde O( \frac{n^2}{\eps^2} )$ samples for integer-valued distributions and $\tilde O( \frac{q}{\eps^2})$ samples for general distributions, where $q = \log \frac{Z(\betamax)}{Z(\betamin)}$ (ignoring some second-order terms and parameters). We show this  is optimal up to logarithmic factors. We illustrate with improved algorithms for counting connected subgraphs, independent sets, and perfect matchings.

As a key subroutine, we estimate all values of the partition function using $\tilde O(\frac{n^2}{\varepsilon^2})$ samples for integer-valued distributions and $\tilde O(\frac{q}{\varepsilon^2})$ samples for general distributions. This improves over a prior algorithm of Huber (2015) which computes a single point estimate $Z(\betamax)$ and which uses a slightly larger amount of samples. We show matching lower bounds, demonstrating this complexity is optimal as a function of $n$ and $q$ up to logarithmic terms. 
\end{abstract}

This is an extended version, which includes work under the same name from ICALP 2023, as well as the earlier work \cite{Kolmogorov:COLT18} appearing in COLT 2018.

\section{Introduction}
A central problem in computational statistics is to convert a procedure for \emph{sampling} combinatorial objects into a procedure for \emph{counting} those objects, and vice versa. In this work, we consider sampling algorithms  for Gibbs distributions.  Formally, given a real-valued function $H(\cdot)$ over a finite set $\Omega$,
the {\em Gibbs distribution} is the family of distributions $\mu^{\Omega}_{\beta}$ over $\Omega$, 
parameterized by $\beta$ over an interval $[\betamin, \betamax]$, of the form 
$$
\mu^{\Omega}_\beta(\omega)=\frac{e^{\beta H(\omega)}}{Z(\beta)} 
$$

The normalizing constant $Z(\beta)=\sum_{\omega \in\Omega}e^{\beta H(\omega)}$ is called the {\em partition function}. These distributions appear in a number of sampling algorithms and are also common in physics, where the parameter $-\beta$ corresponds to the inverse temperature, and $H(\omega)$ is called the {\em Hamiltonian} of the system.

There are two parameters of these distributions which are particularly important. The first is the vector of counts $c_x = |H^{-1}(x)|$ for $x \in \mathbb R$. This is also known as the {\em (discrete) density of states}, that is, the number of preimages of each point under $H$. In statistical physics, for instance, it essentially gives full information about the system and physically relevant quantities
such as entropy, free energy, etc.  A second parameter, whose role is less intuitive, is the partition ratio function 
$$
Q({\beta}) = \frac{Z(\beta)}{Z(\beta_{\min})} \qquad \qquad \text{for $\beta \in [\betamin, \betamax]$}
$$

We also consider an associated probability distribution we call the \emph{gross Gibbs distribution} $\mu_{\beta}(x)$ over $\mathbb R$ given by
$$
\mu_{\beta}(x) = \frac{c_x e^{\beta x}}{Z(\beta)}, \qquad \qquad Z(\beta) = \sum_x c_x e^{\beta x}
$$

Our main aim is to develop \emph{black-box} algorithms to estimate the parameters using oracle access to the gross distribution $\mu_{\beta}$ for chosen query values $\beta \in [\betamin, \betamax]$. No information about the distribution $\mu_{\beta}^{\Omega}$ is used. This can simplify and unify many of the existing problem-specific estimation algorithms for Gibbs distributions across many settings.

\bigskip

We assume  (after rescaling if necessary)  that we are given known parameters $n, q$ with
$$
\log Q(\betamax) \leq q, \qquad  \qquad H(\Omega) \subseteq \mathcal F \eqdef \{0 \} \cup [1,n]
$$
In some cases, the domain is integer-valued, i.e. $$
H(\Omega) \subseteq \mathcal H \eqdef  \{0,1,\ldots,n\} \qquad \qquad \text{for integer $n$}
$$
 We call this the \emph{general integer setting}. A special case of the integer setting, which we call the \emph{log-concave setting}, is when the counts $c_0, c_1, c_2, \dots, c_{n-1}, c_n$ are non-zero and satisfy $c_k/c_{k-1} \geq c_{k+1} / c_k$ for $k = 1, \dots, n - 1$.   The general case, where $H(\omega)$ takes unrestricted values in $\mathcal F$, is called the \emph{continuous setting}.\footnote{The log-concave algorithms still work if some of the counts $c_i$ are equal to zero; in this case, the non-zero counts must form a discrete interval $\{i_0, i_0 + 1, \dots, i_1 - 1, i_1 \}$ and the required bound must hold for $k = i_0+1, \dots, i_1 - 1$.}

When $c_0 > 0$, we allow the value $\beta = -\infty$; here $Z(-\infty) = c_0$ and $\mu_{-\infty}$ is simply the distribution which is equal to $0$ with probability one.  In particular, for $\betamin = -\infty$ we have $Q(\beta) = Z(\beta) / c_0$.

We let $\gamma$ denote the target failure probability and $\eps$ the target accuracy of our algorithms, i.e. with probability at least $1 - \gamma$, the algorithms should return estimates within a factor of $[e^{-\eps}, e^{\eps}]$ of the correct value. Throughout, ``sample complexity'' refers to the number of calls to the oracle; for brevity, we also define the \emph{cost} of an algorithm to be its  \emph{expected sample complexity}.  

To avoid degenerate cases, we always assume $n, q \geq 2$ and $\eps, \gamma \in (0, \tfrac{1}{2})$ unless stated otherwise. If upper bounds $n$ and/or $q$ are not available directly, they can often be estimated by simple algorithms (up to constant factors), or can be guessed by exponential back-off strategies; for simplicity we do not discuss such issues here.

\subsection{Our contributions} 
Our main contribution is a black-box algorithm to estimate counts. To explain this, we must address two technical issues. First, note that in the gross Gibbs distributions, the counts can only be recovered up to scaling. So some (arbitrary) normalization must be chosen; we use a convenient normalization $\pi(x)$ defined as:
$$
\pi(x) \eqdef \frac{c_x}{Z(\betamin)}.
$$

The second issue is that if a count $c_x$ is relatively small, it is inherently hard to estimate  accurately. For example, if $\mu_{\beta}(x) < \delta$ for all $\beta \in [\betamin, \betamax]$ then  $\Omega( 1/\delta)$ samples are needed to distinguish between $c_x = 0$ and $c_x > 0$; with fewer samples, we will never draw $x$ from the oracle. So the complexity of estimating $c_x$ depends on the \emph{largest} $\mu_{\beta}(x)$, over all allowed values of $\beta$.  This gives rise to the parameter $\Delta(x)$ defined as
$$
\Delta(x) \eqdef \max_{\beta \in [\betamin, \betamax]} \mu_{\beta}(x)
$$

With these two provisos, we can define $\Pcoef{\delta,\eps}$ for parameters $\delta, \eps \in (0,\tfrac{1}{2})$ to be the problem of finding a vector $\hat \pi \in [0,1]^{\mathcal F}$ with the following two properties:\begin{itemize}
\item for all $x \in \mathcal F$ with $c_x \neq 0$, there holds $|\hat \pi(x) - \pi(x)| \leq \eps  \pi(x) (1 + \frac{\delta}{\Delta(x)})$.
\item for all $x \in \mathcal F$ with $c_x = 0$, there holds $\hat \pi(x) = 0$.
\end{itemize}

In particular, if $\Delta(x) \geq \delta$, then $\Pcoef{\delta, \eps}$ provides a $1 \pm 3 \eps$ relative approximation to $\pi(x)$.  We emphasize here that when the algorithm succeeds (which should have probability at least $1 - \gamma$), this guarantee holds \emph{simultaneously} for all values $x$.  We develop three main algorithmic results: 
\begin{theorem}\label{th:main:three1}
$\Pcoef{\delta, \eps}$ can be solved with the following complexities:
\begin{itemize}
\item In the continuous setting, with cost $O \bigl(  \frac{ \min \{q^2 + q/\delta,  q \log n + \sqrt{q \log n}/\delta \} }{\eps^2}   \log \tfrac{q}{ \delta \gamma}  \bigr)$.
\item In the general integer setting, with cost $O \bigl( \frac{ n^2  + n/\delta}{\eps^2}    \log^2 \tfrac{n q}{\gamma} \bigr)$.
\item In the log-concave setting,  with  cost  $O \bigl(  \frac{  \min \{  (q+n) \log n, n^2 \}  + 1/\delta }{\eps^2} \log \tfrac{n q}{\gamma}  \bigr)$.
\end{itemize}
where recall that cost refers to the expected number of queries to the oracle.
\end{theorem}
Our full results  are somewhat more precise, see Theorems~\ref{pcoefalg1thm},~\ref{th:Problem2:general1} and \ref{th:Problem2:concave1} for more details.  Note that the continuous-setting algorithm can still be applied to integer-valued problems, and is more efficient in many parameter regimes. 

We also show lower bounds for black-box algorithms to solve $\Pcoef{\delta, \eps}$ by sampling the gross Gibbs distribution. We summarize these results  as follows:
\begin{theorem}
\label{main-lb-thmx}
Let $n > n_0, q > q_0, \eps < \eps_{0}, \delta < \delta_{0}$ for absolute constants $n_0, q_0, \eps_{0}, \delta_{0}$. There are problem instances $\mu$ which satisfy the given bounds $n$ and $q$  such that:
\begin{itemize}
\item $\mu$ is continuous and $\Pcoef{\delta, \eps}$  requires cost $\Omega \bigl( \frac{ (q + \sqrt{q}/\delta) }{\eps^2} \log \frac{1}{\gamma} \bigr)$.

\item $\mu$ is integer-valued and $ \Pcoef{\delta, \eps}$  requires cost 
$\Omega \bigl(  \frac{ \min \{q + \sqrt{q}/\delta, n^2 + n/\delta \} }{\eps^2}  \log \frac{1}{\gamma} \bigr)$. 

\item   $\mu$ is log-concave and $ \Pcoef{\delta, \eps}$  requires cost $\Omega \bigl( \frac{ 1/\delta +  \min \{q, n^2 \} }{\eps^2} \log \frac{1}{\gamma} \bigr)$.
\end{itemize}
\end{theorem}

The first two results match Theorem~\ref{th:main:three1} up to logarithmic factors in $n$ and $q$. There is a small gap between the upper and lower bounds in the log-concave setting. See Theorem~\ref{main-lb-thm} for a more precise and general statement of these bounds. 

Some specialized algorithms have been developed and analyzed for count estimation in combinatorial problems, see e.g.  \cite{jerrum1989approximating,davies_et_al}. Alternatively, the \emph{Wang-Landau (WL)} algorithm \cite{WangLandau:PRL86:2001} is a popular counting heuristic in physical applications; this uses a completely different methodology from our algorithm, based on a random walk on $\mathcal F$ with a running count estimate $\hat c$.  As discussed in~\cite{Schur:JP1252:2019}, there are more than 1500 papers on the WL algorithm and variants such as the $1/t$-WL algorithm~\cite{BelardinelliPereyra:JCP127:2007}. Some variants are guaranteed to converge asymptotically~\cite{Fort:15}, but bounds on convergence rate or accuracy seem to be lacking.  For a representative example, see  \cite{ojeda}, which describes a Gibbs distribution to model protein folding, and uses the WL algorithm to determine relevant properties.

\paragraph{Estimating the partition ratio.} As a key subroutine,  we develop new algorithms to estimate partition ratios. Formally, we define $\PratioAll{\eps}$ to be the problem of computing a static data structure $\mathcal D$ with an associated \emph{deterministic} function $\hat Q( \alpha |\calD)$ satisfying the property
$$
 |\log \hat Q(\alpha|\calD) - \log Q(\alpha)| \leq \eps \qquad \text{for all $\alpha \in [\betamin, \betamax]$}
 $$
 
 We say in this case that $\mathcal D$ is \emph{$\eps$-close}.  We show the following:
\begin{theorem}
\label{th:main:one}
$\PratioAll{\eps}$ can be solved with the following complexities:
\begin{itemize}
\item  In the continuous setting,  with cost $O \bigl( \frac{ \min \{q^2, q \log n \}}{\eps^2}  \log \tfrac{1}{\gamma} \bigr)$.

\item In the general integer setting,  with  cost $O \bigl( \frac{n^2 \log n}{\eps^2} \log \tfrac{1}{\gamma} + n \log q \bigr)$.

\item In the log-concave integer setting,  with  cost $O \bigl( \frac{n^2}{\eps^2} \log \tfrac{1}{\gamma} + n \log q \bigr)$.
	\end{itemize}
\end{theorem}

We again emphasize that the continuous-setting algorithm can be used for integer distributions, and is often more efficient depending on the parameters.  The problem $\PratioAll{\eps}$ was first considered in \cite{TPA:journal}, under the name \emph{omnithermal approximation}, along with an algorithm of cost $O( \frac{q^2}{\eps^2} \log \tfrac{1}{\gamma} )$. A number of algorithms have been developed for \emph{pointwise} estimation of $Q(\betamax)$, which we denote by $\Pratio{\eps}$, with steadily improving sample complexities \cite{Bezakova08, Stefankovic:JACM09,Huber:Gibbs}; the best prior algorithm \cite{Huber:Gibbs} had cost $O( (q \log n) (\log q + \log \log n + \eps^{-2}))$. No specialized algorithms were known  for the integer setting. 

We also show lower bounds for $\Pratio{\eps}$ (and, a fortiori, for $\PratioAll{\eps}$) for black-box sampling algorithms, which match Theorem~\ref{th:main:one} up to $\polylog(n,q)$ factors:
\begin{theorem}
\label{main-lb-thmx2}
Let $n > n_0, q > q_0, \eps < \eps_{0}, \delta < \delta_{0}$ for absolute constants $n_0, q_0, \eps_{0}, \delta_{0}$.  There are problem instances $\mu$ which satisfy the given bounds $n$ and $q$  such that:
\begin{itemize}
\item $\mu$ is continuous and $\Pratio{\eps}$ requires cost  $\Omega \bigl( \frac{ q }{\eps^2} \log \frac{1}{\gamma} \bigr)$.

\item $\mu$ is log-concave and $\Pratio{\eps}$ requires cost  $\Omega \bigl( \frac{ \min\{q, n^2 \} }{\eps^2} \log \frac{1}{\gamma} \bigr).
$
\end{itemize}
\end{theorem}

\paragraph{Explaining the problem definitions.} The definitions of the two problems $\Pcoef{\delta, \eps}$ and $\PratioAll{\eps}$, and especially the error term for $\Pcoef{\delta, \eps}$, may seem obscure. For motivation, we note a few ways to combine their error bounds for estimates of the entire family of Gibbs distributions. (See Appendix~\ref{app:approx-mualpha} for proofs.)
\begin{theorem}
\label{approx-mualpha}
Given solutions for  $\Pcoef{\delta, \eps}$ and $\PratioAll{\eps}$, we can estimate $\mu_{\alpha}$ for any $\alpha \in [\betamin, \betamax]$ and $x \in \mathcal F$ via:
$$
\hat \mu_{\alpha}(x) = \frac{\hat \pi(x) e^{\alpha x}}{\hat Q(\alpha \mid \mathcal D)} \qquad \qquad \text{with  $| \hat \mu_{\alpha}(x) - \mu_{\alpha}(x) | \leq 3 \eps ( \mu_{\alpha}(x) + \delta )$}
$$
\end{theorem}

  \begin{theorem}
    \label{th:main:one-restate2}
Given a solution for $\Pcoef{1/n, \eps/3}$ in the integer setting, we can solve $\PratioAll{\eps}$ with data structure $\mathcal D = \hat \pi$ and setting
$$
\hat Q(\alpha \mid \mathcal D) = \sum_{x \in \mathcal H} \hat \pi(x) e^{\alpha x}.
$$
  \end{theorem}

Indeed, via Theorem~\ref{th:main:one-restate2}, our algorithms for $\Pcoef{\delta, \eps}$ in the general integer setting and log-concave setting immediately give us the corresponding algorithms in Theorem~\ref{th:main:one}.

\subsection{Algorithmic sampling-to-counting}
As a concrete application of our algorithms,  consider the following scenario:  we have a combinatorial system with weighted items, along with an algorithm to (approximately) sample from the corresponding Gibbs distribution. The runtime (e.g. mixing time of a Markov chain) may depend on $\beta$. As a few prominent examples:
\begin{enumerate}
\item Connected subgraphs of a given graph; the weight of a subgraph is the number of edges \cite{GuoJerrum:ICALP18, GuoHe:18, chen}.
\item Matchings of a given graph; the weight of a matching, again, is the number of edges \cite{jerrum1996markov,jain2022approximate}.
\item  Independent sets in bounded-degree graphs; the weight is the size of the independent set \cite{davies_et_al,jain2022approximate}.
\item Vertex cuts in the ferromagnetic Ising model; the weight is the number of crossing edges \cite{carlson2022computational}.
\end{enumerate}

Our algorithms for $\Pcoef{}$ can be combined with any prior sampling algorithms to yield improved and unified counting algorithms, essentially for free. As some examples, we will show the following:
\begin{theorem}
\label{th:count:subgraph}
Let $G = (V,E)$ be a connected graph, and for $i = 0, \dots, |E| - |V| + 1$ let $c_i$ be the number of connected subgraphs with $|E| - i$ edges. There is a fully-polynomial randomized approximation scheme (FPRAS) to simultaneously estimate all values $c_i$ to relative error $\eps$ in time $\tilde O( |E|^3 / \eps^2 )$. (Throughout, we use the notation $\tilde O(x) = x \polylog(x)$.)
\end{theorem}

\begin{theorem}
\label{th:count:matchings}
Let $G = (V,E)$ be a graph and let $v = |V|/2$. For each $i = 0, \dots, v$ let $c_i$ be the number of matchings in $G$ with $i$ edges.  Suppose $c_{v} > 0$ and $c_{v - 1}/c_{v} \leq f$ for a known parameter $f$. There is an FPRAS to simultaneously estimate all counts $c_i$ to relative error $\eps$ in time $\tilde O( |E| |V|^3 f / \eps^2 )$. In particular, if $G$ has minimum degree at least $|V|/2$, the time complexity is $\tilde O( |V|^7 / \eps^2 )$.
\end{theorem}

\begin{theorem}
\label{th:count:ind}
Let $D, \xi > 0$ be arbitrary constants and let $G = (V,E)$ be a graph of maximum degree $D$. For $i = 0, \dots, |V|$ let $c_i$ be the number of independent sets of size $i$.  There is an FPRAS  to simultaneously estimate all values $c_0, \dots, c_t$ to relative error $\eps$ in time $\tilde O( |V|^2/\eps^2 )$ , where $t = (\alpha_c - \xi) |V|$ and $\alpha_c = \alpha_c(D)$ is the computational hardness threshold shown in \cite{davies_et_al}.
\end{theorem}

Theorem~\ref{th:count:subgraph} improves by a factor of $|E|$ over the algorithm in \cite{GuoJerrum:ICALP18}. Similarly,  Theorem~\ref{th:count:matchings} improves by a factor of $|V|$ over the algorithm in \cite{jerrum1996markov}. Theorem~\ref{th:count:ind} matches the runtime of an FPRAS for a \emph{single} count $c_i$ given in \cite{jain2022approximate}.  Full details of these constructions appear in Section~\ref{app-sec}.

There are two minor technical issues to clarify here. First,  to get a randomized estimation algorithm, we must bound the computational complexity of our procedures in addition to the number of oracle calls. In all the algorithms we develop, the computational complexity is a small logarithmic factor times the query complexity. The computational complexity of the oracle is itself typically much larger than this overhead. Thus, our sampling procedures translate directly into efficient randomized algorithms, whose runtime is the expected sample complexity multiplied by the oracle's computational complexity. We will not comment on computational issues henceforth.

Second, we may only have access to an approximate oracle $\tilde\mu_\beta$ instead of the exact oracle $\mu_{\beta}$.  For example, given access to an MCMC sampler with stationary distribution $\mu_{\beta}$ and mixing time $\tau$, we can estimate $\mu_{\beta}$ to total variation distance $\rho$ by running for $\Theta( \tau \log \tfrac{1}{\rho} )$ iterations. By a standard coupling argument, our results remain valid if exact oracles are replaced with sufficiently close approximate oracles (see e.g.\ \cite[Remark 5.9]{Stefankovic:JACM09}); see Appendix~\ref{sec:approx} for a formal statement.
\subsection{Overview}
We will develop two, quite distinct, types of algorithms: the first uses ``cooling schedules'' similar to \cite{Huber:Gibbs}, and the second is based on a new type of ``covering schedule'' for the integer setting.   In Section~\ref{lb-sec}, we show the lower bounds for the problems $\Pratio{}$ and $\Pcoef{}$.

Let us provide a high-level roadmap here.  For simplicity, we assume that tasks need to be solved with constant success probability. 
 
\paragraph{The continuous setting.} 
For the problem $\PratioAll{\eps}$, we extend a statistical method of \cite{Huber:Gibbs} called the \emph{Paired Product Estimator} (see Section~\ref{sec:pratio1}). This constructs a {\em cooling schedule} $\beta_0, \beta_1, \ldots,\beta_{t-1}, \beta_t$, where the log partition function is near-linear between successive values $\beta_i$. As shown in \cite{Huber:Gibbs}, this property gives an unbiased estimator with low variance for the ratios $Q(\beta_i)/Q(\beta_{i-1})$. In addition, we show how to use this property to fill in values $Q(\alpha)$ for $\alpha \in (\beta_i, \beta_{i-1})$ via linear interpolation.

For the problem $\Pcoef{\delta, \eps}$, we start with the following identity:
\begin{equation}
\label{pix-eqn}
\pi(y)=e^{-\beta y}\cdot\mu_\beta(y)\cdot Q(\beta) \qquad \text{for all $y \in \mathcal F, \beta \in [\betamin, \betamax]$}.
\end{equation}
We can estimate  $Q(\beta)$ using our algorithm for $\PratioAll{}$, and we can estimate $\mu_\beta(y)$ by sampling. We then make use of the following important result: if $\mu_\beta([0,y])$ and $\mu_\beta([y,n])$ are both bounded below by constants, then $\mu_\beta(y)\ge \Omega( \Delta(y) )$. For this value $\beta$, we can estimate $\mu_{\beta}(y)$ with $O(\frac{1}{\delta \eps^2})$ samples, and from standard concentration bounds, the resulting estimate $\hat \pi$ satisfies the conditions of $\Pcoef{\delta,\eps}$. For example, if $\Delta(y) \geq \delta$, then $\mu_{\beta}(y)$, and hence $\pi(y)$, is within relative error $O(\eps)$. 

To estimate the full vector of counts,  we find cut-points $x_1 \geq x_2 \geq \dots \geq x_T$, where each interval $[x_{i+1}, x_{i}]$ has a corresponding value $\beta_i$ with  $\mu_{\beta_i}([x_i, n]) \geq \Omega(1)$ and $\mu_{\beta_i}([0,x_{i+1}]) \geq \Omega(1)$. Then every $y \in [x_{i+1}, x_i]$ has $\mu_{\beta_i}(y) \geq \Omega(  \Delta(y))$. In this way, a single value $\beta_i$ provides estimates for all the counts $c_y$ in the interval. We show that only $T = O(\sqrt{q \log n})$ distinct intervals are needed, leading to a cost of $O( \frac{\sqrt{q \log n}}{\delta \eps^2} )$ plus the cost of solving $\PratioAll{}$. The formal analysis appears in Section~\ref{pcoef-contin-sec}.

\paragraph{The integer setting.} To solve  $\Pcoef{\delta, \eps}$, we develop a new data structure we call a \emph{covering schedule}. This consists of a sequence $\betamin = \beta_0,\beta_1,\ldots,\beta_t = \betamax$ and corresponding values $k_1,  \dots, k_t$ so that $\mu_{\beta_{i}}(k_{i})$ and $\mu_{\beta_i}(k_{i+1})$ are large for all $i$. (The definition is adjusted slightly  for the endpoints $i = 0$ and $i = t$). Define $w_i = \min\{\mu_{\beta_{i}}(k_{i}),\mu_{\beta_i}(k_{i+1})\}$ (``weight'' of $i$). If we take $\Omega( \frac{1}{w_i \eps^2})$ samples from $\mu_{\beta_i}$, we can accurately estimate the quantities $\mu_{\beta_i}(k_{i}), \mu_{\beta_i}(k_{i+1})$, in turn allowing us to estimate
$$
\frac{Q(\beta_i)}{Q(\beta_{i-1})} = e^{ (\beta_{i} - \beta_{i-1}) k_i } \frac{ \mu_{\beta_{i-1}}(k_i)}{  \mu_{\beta_{i}}(k_i)}
$$

By telescoping products, this in turn allows us to estimate every value $Q(\beta_i)$. 

Next, for each index $x \in \mathcal H$,  we use binary search to find  $\alpha$ with $\mu_\alpha([0,x])\approx \mu_\alpha([x,n])$ and then estimate $\mu_{\alpha}(x), \mu_{\alpha}(k_i)$ by taking $O(\frac{1}{\delta \eps^2})$ samples. If $\alpha$ lies in interval $[\beta_i, \beta_{i+1}]$ of the covering schedule, we can use the estimates for $Q(\beta_{i})$ and $Q(\beta_{i+1})$ to estimate $Q(\alpha)$.  (A slightly different procedure is used to handle exceptional cases where $\mu_{\alpha}(k_i)$ is small.)   Since we do this for each $j \in \mathcal H$,  the overall cost of this second phase is roughly $O(\frac{n}{\delta \eps^2})$. 

In the log-concave setting, there is a more efficient  algorithm for $\Pcoef{\delta, \eps}$. Here, for any index $i$ of the covering schedule and for $j \in [k_i, k_{i+1}]$ we have $\mu_{\beta_i}(j)\ge \min\{ \mu_{\beta_i}(k_i),\mu_{\beta_i}(k_{i+1}) \} = w_i$; thus, $\beta_i$ ``covers'' the interval $[k_i,k_{i+1}]$.  We can solve $\Pcoef{\delta, \eps}$ with $O(  \frac{1}{\delta \eps^2} + \sum_i \frac{1}{w_i \eps^2} )$
samples, by drawing $\Theta(\frac{1}{w_i\eps^2})$ samples at~$\beta_i$
and $\Theta(\frac{1}{\delta \eps^2})$ samples at $\betamin$ and $\betamax$. 

In either case, after solving $\Pcoef{\delta, \eps}$, we use Theorem~\ref{th:main:one-restate2} to solve $\PratioAll{\eps}$ for free.

\paragraph{Obtaining a covering schedule.} 
The general $\Pcoef{}$ algorithm  above uses $O(  \frac{n}{\delta \eps^2} + \sum_i \frac{1}{w_i \eps^2} )$ samples, and similarly the log-concave algorithm uses $O( \frac{1}{\delta \eps^2} + \sum_i \frac{1}{w_i\eps^2})$ samples.  We thus refer to the quantity $\sum_i \frac{1}{w_i}$ as the \emph{inverse weight} of the schedule.  In the most technically involved part of the paper, we produce a covering schedule with inverse weight $O(n \log n)$ (or $O(n)$ in the log-concave setting). Here we just sketch some key ideas.

First, we construct a ``preschedule'' where each interval can choose two different indices $\sigma^-_i, \sigma^+_{i}$ instead of a single index $k_i$, with the indices interleaving as $\sigma_i^- \leq \sigma_{i+1}^- \leq \sigma_i^+ \leq \sigma_{i+1}^+$. The algorithm repeatedly fill gaps: if some half-integer $\ell + \tfrac{1}{2}$ is not currently covered,  we can select a value $\beta$
with $\mu_\beta([0,\ell])\approx\mu_\beta([\ell+1,n])$, along with corresponding values  $\sigma^-\in[0,\ell], \sigma^+ \in [\ell+1, n]$ with $\mu_\beta(\sigma^+)\cdot(\sigma^+-\ell)\ge \Omega(\frac{1}{\log n})$ and $\mu_\beta(\sigma^-)\cdot(\ell-\sigma^-+1)\ge \Omega(\frac{1}{\log n})$. The interval $[\sigma^-,\sigma^+]$  then fills the gap and has weight $w \geq \Omega( \frac{1}{(\sigma^+ - \sigma^-) \log n})$.

At the end of the process, we throw away redundant intervals so that each $x$ is covered by at most two intervals, and ``uncross'' them into a schedule  with $k_i \in \{ \sigma_{i-1}^+, \sigma_i^-\}$.    Since $\frac{1}{w_i}\le O( (\sigma^+_i - \sigma^-_i) \log n)$ for each $i$, this gives an  $O(n \log n)$ bound of the inverse weight of the schedule.

\section{Preliminaries}
\label{prelim-sec}
   
\noindent  Define $z(\beta) = \log Z(\beta)$ and $z(\beta_1, \beta_2) = \log \frac{Z(\beta_2)}{Z(\beta_1)} = \log \frac{ Q(\beta_2)}{Q(\beta_1)}$; note that $z(\betamin, \betamax) \leq q$ by definition. We write $z'$ and $z''$ for the derivatives of function $z$; observe (see \cite{Huber:Gibbs}) that ${\mathbb E}_{X\sim \mu_{\beta}}[X]=z'(\beta)$. As a consequence, $z', z''$ are nonnegative everywhere and the function $z$ is increasing concave-up.

 \
  
\noindent Define the Chernoff separation functions $F_{+}(u, t) = \bigl( \frac{e^{\delta}}{(1+\delta)^{1+\delta}} \bigr)^{u}$ and $F_{-}(u,t) =  \bigl( \frac{e^{-\delta}}{(1-\delta)^{1-\delta}} \bigr)^{u}$, where $\delta = t/u$. These are upper bounds on the probability that a binomial random variable with mean $u$ is larger than $u+t$ or smaller than $u - t$, respectively. We also write $F(u,t) = F_{+}(u,t) + F_{-}(u,t)$. We use some well-known bounds on these functions, e.g. $F_+(u,t) \leq e^{-t^2/(3 u)}$ for $t \leq u$ and $F_-(u,t) \leq e^{-t^2/(2 u)}$ for all $t$ and
$F(u,t) \leq  2 (1 + t/u)^{-t/2}$ for all $t$.

 \ 
 
 \noindent For a random variable $X$, we write $\Vrel(X) = \frac{\E[X^2]}{(\E[X])^2} - 1 = \frac{\text{Variance}(X)}{(\E[X])^2}$ for the relative variance of $X$.

 \ 
 
\noindent  We write $\mu_{\beta}(x,y), \mu_{\beta}[x,y), \mu_{\beta}[x,y]$ instead of $\mu_{\beta}((x,y)), \mu_{\beta}([x,y)), \mu_{\beta}([x,y])$ for readability. For any fixed value $\chi$, the map $\beta \mapsto \mu_{\beta}[0, \chi)$ is a non-increasing function of $\chi$, since $\frac{ \mu_{\beta}[\chi,n] }{1 - \mu_{\beta}[\chi,n] } = \frac{\sum_{x \geq \chi} c_x e^{\beta x}}{\sum_{x < \chi} c_x e^{\beta x}}$;  when $\beta$ increases by $\delta$, the numerator is multiplied by at least $e^{\delta \chi}$ and the denominator is multiplied by at most $e^{\delta \chi}$.  At several places, we also make use of the following inequality:
 \begin{equation}
\tag{Uncrossing Inequality} \label{prop:basic}
 \text{for values $\alpha \leq \beta, x \leq y$:} \qquad
{\mu_{\alpha}(x)  \mu_{\beta}(y)} \geq \mu_{\alpha}(y) \mu_{\beta}(x) 
\end{equation}

\subsection{The {\tt Balance} subroutine} Given a target $\chi$, we need to find a value $\beta$ with $\mu_{\beta}[0,\chi] \approx 1/2 \approx \mu_{\beta}[\chi, n]$. That is, $\chi$ is approximately the median of the distribution $\mu_{\beta}$.   We can find the target value $\beta$ via noisy binary search.  Formally, for $\tau \in (0,\tfrac12), \chi \in [0,n]$, define ${\tt BalancingVals}(\chi, \tau)$ to be the set of values $\beta \in [\betamin, \betamax]$ satisfying the following two properties:
\begin{itemize}
\item Either $\beta = \betamin$ or $\mu_\beta[0,\chi)\ge \tau$ 
\item Either $\beta = \betamax$ or $\mu_\beta[\chi,n]\ge \tau$
\end{itemize}

Note the slight asymmetry with respect to the left and right endpoints. In Appendix~\ref{sec:Balance}, we show the following result:
\begin{theorem}\label{th:subroutines2x}
For $\tau \in (0,\tfrac12)$ an arbitrary constant, the subroutine ${\tt Balance}(\chi, \tau, \gamma)$ has cost $O(\log \frac{n q}{\gamma})$. With probability at least $1 - \gamma$, it returns a value $\beta \in {\tt BalancingVals}(\chi, \tau)$ (we say in this case that the call is \emph{good}).
\end{theorem}

The algorithm can be accelerated if we can solve $\PratioAll{}$ as a preprocessing step:

\begin{theorem}\label{th:subroutines2xa}
Let $\tau \in (0, \tfrac12)$ be an arbitrary constant and let $\eps = (1/2 - \tau)/30$. If we are given  a data structure $\mathcal D$ for $\PratioAll{\eps}$, then subroutine ${\tt BalancePreprocessed}(\chi, \tau,\mathcal D, \gamma)$ has cost $O(\log \frac{q}{\gamma})$. If $\mathcal D$ is $\eps$-close, then with probability at least $1 - \gamma$, it returns a value $\beta \in {\tt BalancingVals}(\chi, \tau)$.
\end{theorem}

The following observation explains the motivation for the definition.
\begin{proposition}
\label{binarysearch-delta-thm}
If $\beta \in {\tt BalancingVals}(x, \tau)$, then  $\mu_{\beta}(x) \geq \tau \Delta(x)$.
\end{proposition}
\begin{proof}
Consider $\alpha \in [\betamin, \betamax]$ with  $\mu_{\alpha}(x) = \Delta(x)$.  The result is clear if $\alpha = \beta$. Suppose that $\alpha < \beta$; the case $\alpha > \beta$ is completely analogous. So $\beta > \betamin$, and  since $\beta \in {\tt BalancingVals}(x, \tau)$, this implies $\mu_{\beta}[0,x) \geq \tau$. We then have:
\[
\mu_{\alpha}(x) = \frac{c_x e^{\alpha x}}{\sum_y c_y e^{\alpha y}} \leq \frac{c_x e^{\alpha x}}{\sum_{y \leq x}  c_y e^{\alpha y}} \leq  \frac{c_x e^{\beta x}}{\sum_{y \leq x} c_y e^{\beta y}} = \frac{ \mu_{\beta}(x)}{\mu_{\beta}[0,  x]} \leq \frac{\mu_{\beta}(x)}{\tau}. \qedhere
\]
\end{proof}

\subsection{Statistical sampling}
We discuss a few procedures for statistical sampling in our algorithms; see Appendix~\ref{ppe-app} for proofs.

First, we can obtain an unbiased estimator of the probability vector $\mu_\beta$  by computing empirical frequencies $\hat \mu_{\beta}$ from $N$ independent samples from $\mu_\beta$; we denote this process as $\hat\mu_\beta\leftarrow{\tt Sample}(\beta, N)$.  We record the following standard concentration bound, which we will use repeatedly:
\begin{lemma}
\label{binom:succ:lem}
For $\eps, \gamma \in (0,\tfrac12),\bar p \in (0,1]$, suppose we draw random variable $\hat p \sim \frac{1}{N} \text{Binom}(N,p)$ where $N \geq\frac{3 e^{ \eps} \log(4/\gamma)}{(1-e^{-\eps})^2 \bar p}$.
Then, with probability at least $1 - \gamma$, the following two bounds both hold:
\begin{align}\label{eq:hatp-bounds}
|\hat p - p| \leq \eps (p + \bar p), \qquad  \text{and} \qquad
\hat p \in \begin{cases}
[e^{-\eps}p,e^{\eps}p] & \mbox{if }p \ge  e^{-\eps} \bar p \\
[0, \bar p) & \mbox{if }p <  e^{-\eps} \bar p \\
\end{cases}
\end{align}
In particular, if Eq.~(\ref{eq:hatp-bounds}) holds and $\max\{p, \hat p \}\ge \bar p$, then 
$|\log \hat p - \log p| \leq \eps$. 
\end{lemma}

Many of our algorithms are based on calling $\hat\mu_\beta\leftarrow{\tt Sample}(\beta, N)$  and using certain values $\hat \mu_{\beta}(I)$ for sets $I \subseteq \mathcal F$; they succeed when these estimates $\hat \mu_{\beta}(I)$ are close to $\mu_{\beta}(I)$.  We say the execution of {\tt Sample}   \emph{well-estimates} $I$ if Eq.~\eqref{eq:hatp-bounds} hold for $p=\mu_\beta(I)$ and $\hat p=\hat\mu_\beta(I)$; otherwise it \emph{mis-estimates} $I$.  Likewise we say {\tt Sample} well-estimates $k$ if it well-estimates the singleton set $I = \{ k \}$.  Since this comes up so frequently, we write 
$$
\hat\mu_\beta\leftarrow{\tt Sample}(\beta;\eps, \bar p, \gamma) \qquad \text{as shorthand for:} \quad  \hat\mu_\beta\leftarrow{\tt Sample} \bigl( \beta, \big \lceil \tfrac{3 e^{ \eps} \log(4/\gamma)}{(1-e^{-\eps})^2 \bar p} \big \rceil \bigr)
$$
This has cost $O( \frac{ \log(1/\gamma)}{\eps^2 \bar p} )$, and each set $I$ is well-estimated with probability at least $1 - \gamma$. Here $\eps$ should be interpreted as the desired relative error,  $\bar p$ as the minimum ``detection threshold'' for $\mu_{\beta}$, and $\gamma$ as the desired failure probability.

In a similar way, our algorithms for $\Pcoef{\delta, \eps}$ estimate $\pi(x)$ by sampling some well-chosen value $\mu_{\alpha}(x)$. We record the following straightforward error bound:
\begin{lemma}
\label{general-pcoef-lemma}
Suppose that for $x \in \mathcal F$, we have $\alpha \in [\betamin, \betamax]$ and empirical estimates $\hat Q(\alpha), \hat \mu_{\alpha}(x)$ satisfying the following bounds:
 
\noindent (A1) $|\log \hat Q(\alpha) - \log Q(\alpha)| \leq \eps/3$.  \\
 (A2) $| \hat \mu_{\alpha}(x) - \mu_{\alpha}(x)| \leq \eps/3 \cdot  \mu_{\alpha}(x) (1 + \delta/\Delta(x) )$. 
 
 Then the estimated value $ \hat \pi(x) = \hat Q(\alpha) e^{- \alpha x}\hat \mu_{\alpha}(x)$ satisfies the criteria for the problem $\Pcoef{\delta, \eps}$.
\end{lemma}

Finally, in a few places, we need to estimate certain telescoping products. Direct Monte Carlo sampling does not give strong tail bounds, so we use a standard method based on median amplification.
\begin{theorem} 
\label{ppe-est-thm}
Suppose we can draw independent samples of non-negative random variables $X_1, \dots, X_N$. Given input $\gamma \in (0,\tfrac12)$ and $\tau \geq 1$, the procedure  ${\tt EstimateProducts}(X, \tau, \gamma)$ uses $O( N \tau \log \tfrac{1}{\gamma})$  total samples of the $X$ variables and returns  estimates $( \hat X^{\tt prod}_1, \dots, \hat X^{\tt prod}_N )$. If $\tau \geq \sum_{i=1}^N \Vrel(X_i)$,
 then with probability at least $1 - \gamma$, it holds that
$$
\text{for all $i = 1, \dots, N$: } \frac{  \hat X_i^{\tt prod}}{ \prod_{j=1}^i \E[X_j]  } \in [e^{-\eps}, e^{\eps}] \qquad \qquad \text{for $\eps = \sqrt{ \sum_{i=1}^N \Vrel(X_i) / \tau }$}
$$
In this context, we also define $\hat X^{\tt prod}_0 = 1$.
\end{theorem}

\section{Solving $\PratioAll{\eps}$ in the continuous setting}
\label{sec:pratio1}

In this section, we will develop our algorithm {\tt PratioContinuous} to solve $\PratioAll{\eps}$ based on a stochastic process called TPA developed in \cite{Huber:Gibbs}. The TPA process generates a sequence $\beta_i$ of values in the range $[\betamin, \betamax]$  as follows:

\begin{algorithm}[H]
\DontPrintSemicolon
set $\beta_0=\betamax$ \\
\For {$i=1$ {\bf to} $+\infty$}
{
   draw $X$ from $\mu_{\beta_{i-1}}$  and draw $\eta$ from the Exponential distribution of rate one. \\
   set $\beta_i = \beta_{i-1} - \eta / X$ \\
  \textbf{if $X = 0$ or $\beta_i < \betamin$ then return}  set $B = \{ \beta_1, \dots, \beta_{i-1} \}$ \\
}  
\caption{ {The TPA process. {\tt Output:} a multiset of values in the interval $[\betamin,\betamax]$}}

\label{alg:TPAone0}
\end{algorithm}

$ {\tt TPA}(k)$ is defined to be the union of $k$ independent runs of Algorithm~\ref{alg:TPAone0}. The critical property shown in \cite{TPA, TPA:journal}  is  that $z({\tt TPA}(k))$ is a rate-$k$ Poisson Point Process on $[z(\betamin), z(\betamax)]$. In other words, if $\{ \beta_1, \dots, \beta_{\ell} \}$ is the output of ${\tt TPA}(k)$, then the random variables $z_i = z(\beta_i)$  are generated by the following process:

\begin{algorithm}[H]
\DontPrintSemicolon
set $z_0=z(\betamax)$ \\
\For {$i=1$ {\bf to} $+\infty$}
{
  draw $\eta$ from the Exponential distribution of rate $k$\\  
   set $z_i=z_{i-1}-\eta$ \\
  \textbf{if $z_i < z(\betamin)$ then return} set $\{ z_1, \dots, z_{i-1} \}$ \\
}  
\caption{Equivalent process for generating $z({\tt TPA}(k))$. }\label{alg:TPAone}
\end{algorithm}

As noted in \cite{TPA:journal}, the TPA process can be used for a simple algorithm for $\PratioAll{\eps}$, which already gives us half of our result.

\begin{algorithm}[H]
compute $ B \leftarrow {\tt TPA}(k)$ and set $\mathcal D = B$.
\BlankLine
\BlankLine
\nonl \ \ \  \textbf{Query algorithm $Q(\alpha \mid \calD )$:} \\
\BlankLine
\vspace{-0.03in}
\textbf{return} $\hat Q(\alpha|\calD)=\exp \bigl(  | B \cap [\betamin, \alpha)| / k  \bigr)$
\caption{\label{alg:PPE-wrapper0} Algorithm ${\tt TPA-Pratio}(k)$:}
\end{algorithm}

\begin{theorem}[\cite{TPA:journal}]
\label{tparatiothm}
Algorithm {\tt TPA-Pratio} has cost $O(k (1 + \log Q(\betamax)))$. If $k \geq \frac{20 \log Q(\betamax)}{\eps^2} \log \tfrac{3}{\gamma}$, then the data structure $\mathcal D$ is $\eps$-close with probability at least $1 - \gamma$. 

In particular,with $k = \lceil \frac{20 q}{\eps^2} \log \tfrac{3}{\gamma} \rceil $, then the algorithm solves {\tt PratioAll} with cost $O( \frac{q^2}{\eps^2} \log \tfrac{1}{\gamma})$.
\end{theorem}
\begin{proof}
The runtime follows since each run of TPA has cost $O( 1 + \log Q(\betamax))$.  

For any values  $\beta_1 < \beta_2$, let $B(\beta_1, \beta_2) = |B \cap (\beta_1, \beta_2)|$. This is a Poisson random variable with mean $k \cdot z(\beta_1, \beta_2)$. So $\Pr{ |B(\beta_1, \beta_2)/k - z(\beta_1, \beta_2)| \geq  \eps/2} \leq F( q k, \eps k / 2) \leq 2 e^{-\frac{\eps^2 k}{12 q}} \leq \gamma/2$ for fixed $\beta_1, \beta_2$. 

Let $\mathcal E$ denote the event that $\mathcal D$ fails to be $\eps$-close. Consider the random process where $B \cap [\betamin, \alpha]$ is revealed while $\alpha$ gradually increases; suppose  $\mathcal E$ occurs and let $\alpha$ be the first value with $| B(\betamin, \alpha)/k - z(\betamin, \alpha)| > \eps$. Since $z(B)$ is a Poisson process, $B \cap (\alpha, \betamax)$ retains its original, unconditioned probability distribution.  Accordingly, we have $| B(\alpha, \betamax)/k - z(\alpha, \betamax) | \leq \eps/2$ with probability at least $1 - \gamma/2$ , in which case we get
\begin{align*}
& | B(\betamin, \betamax)/k - z(\betamin, \betamax) | =  | (B(\betamin, \alpha)/k  - z(\betamin,\alpha)) + (B(\alpha, \betamax)/k - z(\alpha, \betamax))   | \\
&\qquad \geq |B(\betamin, \alpha)/k  - z(\betamin,\alpha)| -  | B(\alpha, \betamax)/k - z(\alpha, \betamax)|   \geq \eps - \eps /2 \geq \eps/2.
\end{align*}

So  $\Pr{ | B(\betamin, \betamax)/k - z(\betamin, \betamax)| \geq \eps/2 \mid \mathcal E }  \geq 1 - \gamma/2$. On the other hand, we have also shown that $\Pr{ | B(\betamin, \betamax)/k - z(\betamin, \betamax)| \geq \eps/2} < \gamma/2$. Hence $\Pr{\mathcal E} \leq \frac{\gamma/2}{1 - \gamma/2} \leq \gamma$. 
\end{proof}

\subsection{Paired Product Estimator}
To improve the sample complexity to $O( \frac{q \log n}{\eps^2} \log \tfrac{1}{\gamma})$, we need a more advanced statistical technique called the Paired Product Estimator (PPE) developed in \cite{Huber:Gibbs}. Here, we describe a simple version of the algorithm.\footnote{See also \cite{Huber:Gibbs,Kolmogorov:COLT18} for a stand-alone version of the PPE algorithm to directly solve  $\Pratio{\eps}$ with high probability.}

\begin{algorithm}[H]
\DontPrintSemicolon
compute $B \leftarrow {\tt TPA}(k)$ sorted as $B = \{ \beta_1, \dots, \beta_{t-1} \}$ \\
define $\beta_0 = \betamin$ and $\beta_t = \betamax$. \\
\For{$i=1,\ldots,t$} {
define random variable $W_i = \exp \bigl( \frac{\beta_{i}-\beta_{i-1}}{2} \cdot X \bigr)$ where $X$ is drawn from  $\mu_{\beta_{i-1}}$ \\
define random variable $V_i = \exp \bigl( -\frac{\beta_{i}-\beta_{i-1}}{2}\cdot Y \bigr)$ where $Y$ is drawn from  $ \mu_{\beta_i}$ 
}
set $\hat W^{\tt prod} = {\tt EstimateProducts}(W_i,32,1/10)$ \\
set $\hat V^{\tt prod} = {\tt EstimateProducts}(V_i, 32,1/10)$ \\
\textbf{return} $\calD=((\beta_0,\ldots,\beta_t),(\hat Q(\beta_0),\ldots,\hat Q(\beta_t))$ where $\hat Q(\beta_i) = \hat W^{\tt prod}_i / \hat V^{\tt prod}_i$
\BlankLine
\BlankLine
\nonl \ \ \  \textbf{Query algorithm $Q(\alpha \mid \calD )$:} \\
\vspace{-0.03in}
\BlankLine
find  index $i$ with $\alpha \in [\beta_i, \beta_{i+1}]$ \\
let $\alpha = (1-v) \beta_{i-1} + v \beta_i$ for $v \in [0,1]$ \\
\textbf{return} $\hat Q(\alpha|\calD)=\exp \bigl(  (1-v) \log \hat Q(\beta_{i-1})  + v \log \hat Q(\beta_{i}) \bigr)$
\caption{\label{alg:PPE-wrapper} Algorithm ${\tt PPE - Pratio}(k)$}
\end{algorithm}

In this context, the sequence $(\beta_0, \dots, \beta_t)$ constructed at Line 2 is called a \emph{cooling schedule}. We show the following main result for this algorithm:

\begin{lemma}
\label{ppe-summary}
The algorithm {\tt PPE-Pratio} has cost $O(k q)$. If $k \geq \frac{10 \log (z'(\betamax)/z'(\betamin))}{\eps^2}$ for $\eps \in (0,1)$, then the resulting data structure $\mathcal D$ is $\eps$-close with probability at least $0.7$.
\end{lemma}

For the complexity bound,  Line 1 has cost $O( k q)$. The two applications of {\tt EstimateProducts} have cost $O(t)$ given the cooling schedule $\beta_i$, and the expected size of $t$ is $O(k q)$.

For the correctness bound, we consider curvature parameters $\kappa_i$ and $\kappa$ of the schedule defined as:
$$
\kappa_i = z(\beta_{i-1})-2z \bigl( \tfrac{\beta_{i-1}+\beta_{i}}{2} \bigr) + z(\beta_{i}), \qquad \qquad \kappa = \sum_{i=1}^{t} \kappa_i
$$

Our main technical result here is the following bound:
\begin{lemma}
\label{kapp-bnd1}
There holds $\E[ \kappa ] \leq \frac{\log(z'(\betamax)/z'(\betamin))}{k}.$
\end{lemma}
\begin{proof}
As shown in \cite[Lemma 3.2]{Huber:Gibbs}, we have $\frac{z'(\beta_i)}{z'(\beta_{i-1})} \geq \exp(  \frac{2 \kappa_i}{z(\beta_{i-1}, \beta_i)})$ for each $i$ (see also \cite{Stefankovic:JACM09}).  Equivalently, we have $\kappa_i \leq \frac{ z(\beta_{i-1},\beta_i)}{2} \cdot \log \bigl( \frac{z'(\beta_i)}{z'(\beta_{i-1})} \bigr)$.

Observe that $z(\beta)$ and $z'(\beta)$ are both strictly increasing functions of $\beta$.   Thus, for any fixed value $x \in [ z'(\betamin),  z'(\betamax)]$, there is a unique value $\beta \in [\betamin, \betamax]$ with $z'(\beta) =  x$, and a unique index $i$ with $z(\beta) \in [z(\beta_{i-1}), z(\beta_{i}))$. Define $A_x = z(\beta_{i-1}, \beta_{i})$, that is, the distance between the nearest points of the cooling schedule.   We can rewrite the bound for $\kappa_i$ as $$
\kappa_i \leq \frac{ z(\beta_{i-1},\beta_i)}{2} \cdot \log \Bigl( \frac{z'(\beta_i)}{z'(\beta_{i-1})} \Bigr) = \frac{1}{2} \int_{z'(\beta_{i-1})}^{z'(\beta_i)} \frac{A_x}{x} \ \d x
$$
 so that:
\begin{equation}
\label{kappaeqn3}
\kappa \leq \frac{1}{2} \sum_i \int_{z'(\beta_{i-1})}^{z'(\beta_{i})} \frac{A_x}{x} \ \d x = \frac 12 \int_{z'(\betamin)}^{z'(\betamax)} \frac{A_x}{x} \ \d x. 
\end{equation}

We claim that, for any fixed $x$,  we have $\E[A_x] \leq 2/k$. For, suppose $z'(\beta^{\ast}) = x$ and let $z^\ast=z(\beta^{\ast})$.  Consider a Poisson process $X_1,X_2,\ldots$ on $[z^\ast,+\infty)$ and a Poisson process $X_{-1},X_{-2},\ldots$  on $(-\infty,z^\ast]$, both with rate $k$.  By the superposition theorem for Poisson processes~\cite[page 16]{Kingman:PPP}, the bidirectional sequence ${\bf X}=\ldots,X_{-2},X_{-1},X_1,X_2,\ldots$
is also a rate-$k$ Poisson process on $(-\infty,+\infty)$.  As discussed earlier,  the sequence $z(\beta_i)$ has the same distribution as the subsequence ${\bf X}\cap[z(\betamin),z(\betamax)]$.  So $$
\E[A_x] = \E[ \min \{ z(\betamax), X_1 \} - \max \{ z(\betamin), X_{-1} \} ] \leq \E[X_1 -  z^{\ast}]  + \E[z^{\ast} - X_{-1}] = 2/k
$$
where the last equality uses the fact that  $X_1 -  z^{\ast}$  and $z^{\ast} - X_{-1}$ are rate-$k$ Exponential random variables.

Substituting into Eq.~(\ref{kappaeqn3}) gives:
\[
\E[\kappa] \leq \frac 12\int_{z'(\betamin)}^{z'(\betamax)} \frac{\E[A_x]}{x} \ \d x \leq \frac 12\int_{z'(\betamin)}^{z'(\betamax)} \frac{2}{k x} \ \d x = \frac{\log(z'(\betamax)) - \log(z'(\betamin))}{k}.  \qedhere
\]
\end{proof}

\begin{proposition}
\label{dclosethm}
If $k \geq \frac{10 \log (z'(\betamax)/z'(\betamin))}{\eps^2}$ for $\eps \in (0,1)$, then with probability at least $0.7$, the following two bounds hold: (a) $\kappa \leq \eps^2$ and (b) $| \log (\hat W_i^{\tt prod}/ \hat V_i^{\tt prod}) - \log Q(\beta _i)| \leq \eps / 2$ for all $i = 0, \dots, t$.
\end{proposition}
\begin{proof}
By Markov's inequality, we have $\kappa \leq \eps^2$ with probability at least $0.9$. Suppose this holds, and  condition on the fixed schedule $B$. Denote $\bar\beta_{i-1,i}=\frac{\beta_{i-1}+\beta_{i}}2$ for each $i$. A calculation shows (see~\cite{Huber:Gibbs}):
\begin{align*}
&\E[W_i]=\frac{Z(\bar\beta_{i-1,i})}  {Z(\beta_{i-1})}\;\;\quad\quad
\E[V_i]=\frac{Z(\bar\beta_{i-1,i})}{Z(\beta_{i})}\;\;\quad
&\Vrel(W_i)=\Vrel(V_i)=\frac{Z(\beta_{i-1})Z(\beta_{i})}{Z(\bar\beta_{i-1,i})^2} - 1=e^{\kappa_i} - 1
\end{align*}

So $\sum_i \Vrel(W_i) \leq \sum_i (e^{\kappa_i} - 1) \leq e^{\sum_i \kappa_i} - 1 = e^{\kappa} - 1$; since $\kappa \leq \eps^2 \leq 1$, this is at most $2 \eps^2$. Likewise $\sum_i \Vrel(V_i) \leq 2 \eps^2$. So by Theorem~\ref{ppe-est-thm},  with probability at least $0.8$ the values $\hat W^{\tt prod}, \hat V^{\tt prod}$ returned by {\tt EstimateProducts} satisfy
$$
\frac{\hat V^{\tt prod}_i}{\prod_{j=1}^i \E[ V_j ]} \in [e^{-\eps/4, \eps/4}], \qquad  \frac{\hat W^{\tt prod}_i}{\prod_{j=1}^i \E[ W_j ]}  \in [e^{-\eps/4, \eps/4}] \qquad \text{for all $i = 0, \dots, t$}
$$
In this case, each value $\hat Q(\beta_i) = \hat W^{\tt prod}_i / \hat V^{\tt prod}_i$ is within $[e^{-\eps/2}, e^{\eps/2}]$ of the product $\prod_{j=1}^i \frac{\E[W_j]}{\E[V_j]}$; by telescoping products, this is precisely $Z(\beta_i)/Z(\beta_0) = Q(\beta_i)$.  
\end{proof}

\begin{proposition}
If the bounds of Proposition~\ref{dclosethm} hold,  then the data structure $\mathcal D$ is $\eps$-close. 
\end{proposition}
\begin{proof}
Consider $\alpha = (1-v) \beta_{i-1} + v \beta_i$ for $v \in [0,1]$.   Suppose that $v \leq \tfrac12$;  the case where $v \geq  \tfrac12$ is completely symmetric. We need to show:
\begin{equation}
\label{gag15}
 | (1-v) \hat z(\beta_{i-1})  + v  \hat z(\beta_{i}) - z( \alpha ) | \leq \eps
 \end{equation}
 where we write $\hat z(\beta) = \log \hat Q(\alpha \mid \mathcal D)$ for brevity. We can rewrite the definition of $\kappa_i$ as:
\begin{equation}
\label{keqn1}
\kappa_i = 2 z(\bar \beta, \beta_i) - z(\beta_{i-1}, \beta_i) \qquad \qquad \text{for $\bar \beta = \tfrac{\beta_{i-1}+\beta_{i}}2$}
\end{equation}
Since $z$ is an increasing concave-up function, $z(\beta_{i-1}, \alpha) \leq v \cdot z(\beta_{i-1}, \beta_i)$ and $z(\beta_{i-1}, \alpha) \geq z(\beta_{i-1}, \beta_i) - 2 ( 1- v) z(\bar \beta, \beta_i)$. By Eq.~(\ref{keqn1}), this implies that $z(\beta_{i-1}, \alpha) \geq v \cdot z(\beta_{i-1},\beta_i) - \kappa_i$. So overall
$$
|z (\beta_{i-1}, \alpha) - v \cdot z (\beta_{i-1}, \beta_i)  | = |(1-v) z(\beta_{i-1}) + v \cdot z(\beta_i) - z (\alpha)| \leq \kappa_i;
$$
 by Proposition~\ref{dclosethm}(a) this at most $\eps^2 \leq \eps / 2$.  By Proposition~\ref{dclosethm}(b), we have $| \hat z(\beta_{i-1}) - z(\beta_{i-1})| < \eps/2$ and $| \hat z(\beta_i) - z(\beta_i) | < \eps/2$, so this implies Eq.~(\ref{gag15}).
\end{proof}

 This concludes the proof for Lemma~\ref{ppe-summary}.

\subsection{A hybrid algorithm}
\label{sec-est3}
Our hybrid algorithm for $\PratioAll{\eps}$ combines the two simpler algorithms in their appropriate parameter regimes.

\begin{algorithm}[H]
\If{$q > \log n$}{
set $\beta_{\circ} \leftarrow \betamax$ and $t = 0$ \\
call $\mathcal D_0 = {\tt TPA - Pratio}(k_1)$ for $k_1 = \lceil \frac{20 q}{\eps^2} \log \tfrac{3}{\gamma} \rceil $
} \ElseIf {$q \leq \log n$} {
set $\beta_{\circ} \leftarrow {\tt Balance}(1/2, 1/4, \gamma/4)$ \\
call $\mathcal D_0 = {\tt TPA - Pratio}(k_1)$ with $k_1 = \lceil \frac{40 \log 4}{\eps^2} \log \tfrac{12}{\gamma} \rceil$ for the interval $[\betamin, \beta_{\circ}]$. \\
\For{$i = 1, \dots, t = 40 \log \tfrac{2}{\gamma}$} {
Set $\mathcal D_i \leftarrow {\tt PPE - Pratio}( k_2 )$ with $k_2 = \lceil \frac{10 \log(4 n)}{\eps^2} \rceil$ for the interval $[\beta_{\circ}, \betamax]$ \\
}
}
output tuple $\mathcal D = (\beta_{\circ}, \mathcal D_0, \mathcal D_1, \dots, \mathcal D_t)$

\BlankLine
\BlankLine
\nonl \ \ \  \textbf{Query algorithm $Q(\alpha \mid \calD )$:} \\
\BlankLine
\vspace{-0.03in}
\If{$\alpha \leq \beta_{\circ}$} {

\textbf{return} $\hat Q (\alpha \mid \mathcal D_0)$.
} \Else {
\textbf{return} $\hat Q( \beta_{\circ} \mid \mathcal D_0) \cdot \text{median} \{ \hat Q( \alpha \mid \mathcal D_1), \dots, \hat Q(\alpha \mid \mathcal D_t) \}$
}

\caption{Algorithm ${\tt PratioContinuous}(\eps, \gamma)$  \label{alg:hybrid}}
\end{algorithm}
\begin{theorem}
\label{raaw1}
Algorithm {\tt PratioContinuous} solves $\PratioAll{\eps}$ with cost 
 $$
O \Bigl( \frac{\min\{q^2, q \log n \} }{\eps^2} \log \tfrac{1}{\gamma} \Bigr).
$$
\end{theorem}
\begin{proof}
If $q > \log n$, then this simply reduces to calling {\tt TPA-Pratio} and the result follows from Theorem~\ref{tparatiothm}. So suppose $q \leq \log n$.  In this case, Line 5 has cost $O( \log \frac{n q}{\gamma} )$ and each iteration of the loop at Line 7 has cost $O( k_2 q) = O(\frac{q \log n}{\eps^2} )$. Line 6 has cost $O( k_1 (1+\log Q(\beta_{\circ}))) \leq O( \frac{q}{\eps^2} \log \tfrac{1}{\gamma})$. Overall, in this case the algorithm has cost $O( \frac{q \log n }{\eps^2} \log \tfrac{1}{\gamma})$.  

 For the algorithm correctness, suppose the call to {\tt Balance} at Line 5 is good, which occurs with probability at least $1 - \gamma/4$. In this case, we claim that the following two bounds hold:
\begin{equation}
\label{rat-eqns}
Q(\beta_{\circ}) \leq 4, \qquad \qquad z'(\betamax) / z'(\beta_{\circ}) \leq 4 n
\end{equation}

The first bound is immediate if $\beta_{\circ} = \betamin$. Otherwise, we have $\mu_{\beta_{\circ}}(0) = \mu_{\beta_{\circ}}[0, 1/2] \geq 1/4$. We can estimate $Z(\betamin) \geq c_0$ and $Z(\beta_{\circ}) = \frac{c_0}{\mu_{\beta_{\circ}}(0)}  \leq 4 c_0$. So $Q(\beta_{\circ}) \leq 4$.

 Likewise, the second bound is immediate if $\beta_{\circ} = \betamax$. Otherwise, we have $\mu_{\beta_{\circ}}[1, n] \geq 1/4$. Thus $z'(\beta_{\circ}) = {\mathbb E}_{X\sim \mu_{\beta_{\circ}}} [X]  \geq  \mu_{\beta_{\circ}}[1,n] \geq 1/4$.   So  $z'(\betamax)/z'(\beta_{\circ}) \leq n/(1/4) = 4 n$. 

By Eq.~(\ref{rat-eqns}), Theorem~\ref{tparatiothm} and our choice of $k_1$, the data structure $\mathcal D_0$ is $\eps/2$-close for the range $[\betamin, \beta_{\circ}]$ with probability at least $1 - \gamma/4$.  In particular, $\hat Q(\beta_{\circ} \mid \mathcal D_0)$ is within $e^{\pm \eps/2}$ of  $Q(\beta_{\circ})$.  Also, by Lemma~\ref{ppe-summary}, each $\mathcal D_i$ is $\eps/2$-close for the range $[\beta_{\circ}, \betamax]$ with probability at least $0.7$. So with probability at least $1 - F_-(0.7 t, 0.2 t) \geq 1 - \gamma/2$, at least half the structures $\mathcal D_i$ are $\eps/2$-close; when this holds, then $\text{median} (\hat Q( \alpha \mid \mathcal D_1), \dots, \hat Q(\alpha \mid \mathcal D_t))$ is within $e^{\pm \eps/2}$ of $\frac{Q(\alpha)}{Q(\beta_{\circ})}$ for all values $\alpha \in [\beta_{\circ}, \betamax]$. Then the final data structure $\mathcal D$ is $\eps$-close for the entire range $[\betamin, \betamax]$.
\end{proof}

We have established the continuous-setting algorithm of Theorem~\ref{th:main:one}.

\section{Solving $\Pcoef{\delta, \eps}$ in the continuous setting}
\label{pcoef-contin-sec}
In this section, we develop our algorithm for $\Pcoef{\delta, \eps}$, using the algorithm for $\PratioAll{}$ as a subroutine.

\begin{algorithm}[H]
set $\mathcal D \leftarrow {\tt PratioContinuous}( \min\{1/120, \eps/3 \}, \gamma/4)$ \\
initialize $x_0 \leftarrow n$ \\
\For{$t = 1$ \bf{to} $T = 2 \lceil \min \{q/2, \sqrt{q \log n} \} \rceil + 4$ }
{
  set $\beta_t \leftarrow {\tt BalancePreprocessed}(x_{t-1}, 1/4, \mathcal D, \frac{\gamma}{4 T})$ \\
  
  set $\hat \mu_{\beta_t} \leftarrow {\tt Sample}( \beta_t,  \lceil \frac{10^6 \log{ \frac{10 T}{\delta \gamma} }}{\delta \eps^2} \rceil)$ \\

\If{$\beta_t > \betamin$} {

set $x_t$ to be the minimum value with $\hat \mu_{\beta_t} [0, x_t]  \geq 1/90$ \\

\textbf{foreach $y \in (x_t, x_{t-1}]$ do} set $\hat \pi(y) = \hat Q(\beta_t \mid \mathcal D) e^{-\beta_t} \hat \mu_{\beta}(y)$
}

  \ElseIf{$\beta_t = \betamin$} {
\textbf{foreach $y \in [0, x_{t-1}]$ do} set $\hat \pi(y) = e^{-\betamin} \hat \mu_{\betamin}(y)$ \\
\textbf{return} all estimated values $\hat \pi$
  }
  }
  
  output ERROR and \textbf{terminate}
\caption{Solving $\Pcoef{\delta, \eps}$ for error parameter $\gamma$.  \label{pcoef-alg1} }
\end{algorithm}

\begin{theorem}
\label{pcoefalg1thm}
Algorithm~\ref{pcoef-alg1}  solves $\Pcoef{\delta, \eps}$ with cost
$$
O\Bigl(  \min \{  \frac{ q^2 \log \tfrac{1}{\gamma} + (q/\delta) \log \tfrac{q}{\delta \gamma}}{\eps^2}, 
 \frac{q \log n \log \tfrac{1}{\gamma} + (\sqrt{q \log n}/ \delta ) \log \tfrac{q}{\delta \gamma}  }{\eps^2 }  \} \Bigr)
$$
\end{theorem}
The complexity bound follows immediately from specification of subroutines.  We next analyze the success probability. Throughout, we suppose that that the call at Line 1 produces a data structure $\mathcal D$ which is  $\min\{1/120, \eps/3 \}$-close. There are a number of intermediate calculations.
\begin{proposition}
\label{rlistprop}
With probability at least $1 - \gamma/2$, the following conditions hold for all iterations $t$: \\
\noindent (a)  If $\beta_t < \betamax$, then $\mu_{\beta_t}[x_{t-1}, n] \geq 1/4$; if $\beta_t > \betamin$ then $\mu_{\beta_t}[0, x_{t-1}) \geq 1/4$. \\
\noindent (b) If $\beta_t > \betamin$, then $x_t < x_{t-1}$ and $\mu_{\beta_t}[0,x_t) \leq 1/80$ and $\mu_{\beta_t} [0,x_t] \geq 1/100$. 
\end{proposition}
\begin{proof}
With probability at least $1 - \frac{\gamma}{4 T}$, the call to {\tt BalancePreprocessed} at each iteration $t$ is good. Suppose this holds. By definition, this shows (a). For (b), suppose $\beta_t > \betamin$, and let $v,w$ be the minimum  values with $\mu_{\beta_t} [0,v]  \geq 1/80$ and  $\mu_{\beta_t}[0,w]  \geq 1/100$ respectively. (These both exist since the function $x \mapsto \mu_{\beta}[0, x]$ is right-continuous.) Then $\mu_{\beta_t}[0,v) \leq 1/80, \mu_{\beta_t}[0,w)   \leq 1/100$. Since $\mu_{\beta_t}[0,x_{t-1}) \geq 1/4$, we have $w \leq v < x_{t-1}$.

By Lemma~\ref{binom:succ:lem}, there is a probability of at least $1 - \frac{\gamma}{4 T}$ that Line 5 well-estimates both intervals $[0,v], [0,w)$ with respect to parameters $\eps/10$ and $\bar p = 1/100$. Suppose this also holds.  Then $\hat \mu_{\beta_t}[0, v] \geq e^{-\eps/10} \mu_{\beta_t}[0,v] \geq e^{-1/20} \cdot 1/80 > 1/90$. So $x_t \leq v < x_{t-1}$ and $\mu_{\beta_t}[0,x_t) \leq \mu_{\beta_t}[0,v) \leq 1/80$.  Likewise, we have $\hat \mu_{\beta_t}[0, w)  \leq e^{1/20}  \cdot 1/100  <  1/90$. Thus $x_t \geq w$ and $\mu_{\beta_t}[0,x_t] \geq \mu_{\beta_t}[0,w] \geq 1/100$.
\end{proof}

\begin{proposition}
\label{ggprop}
If the bounds of Proposition~\ref{rlistprop} hold, the algorithm does not terminate with ERROR. 
\end{proposition}
\begin{proof}
Suppose the algorithm reaches Line 12, so $\beta_1, \dots, \beta_{T} > \betamin$ and we have a decreasing sequence $x_1 > x_2 > \dots > x_{T-1} > x_T$. Note that $x_i \in \mathcal F$ for each $i$, and so $x_{T-1} \geq 1$. 

For each iteration $t < T$, Proposition~\ref{rlistprop} gives $\mu_{\beta_{t+1}}[x_{t+1}, x_{t})  = \mu_{\beta_{t+1}} [0,x_{t}) - \mu_{\beta_{t+1}}[0,x_{t+1}) \geq 1/4 - 1/80 \geq 1/5$ and $\mu_{\beta_{t}}[x_{t+1}, x_{t}) \leq \mu_{\beta_{t}} [0,x_t) \leq 1/80$ and $\mu_{\beta_{t+1}}[0,x_{t}) \geq 1/4$.   Since the map $\beta \mapsto \mu_{\beta}[0,\chi)$ is non-increasing and $\mu_{\beta_{t+1}}[0,x_t) \geq 1/4 > 1/80 \geq \mu_{\beta_t}[0,x_t)$, we have $\beta_{t+1} < \beta_t$ strictly. We can estimate the partition ratio between $\beta_{t+1}$ and $\beta_t$ as:
\begin{equation}
\label{fvp1}
\frac{Z(\beta_t)}{Z(\beta_{t+1})} =  \frac{ \mu_{\beta_{t+1}}[x_{t+1}, x_{t})}{ \mu_{\beta_t}[x_{t+1}, x_{t}) } \cdot \frac{\sum_{x \in [x_{t+1}, x_{t})} c_x e^{\beta_{t} x}}{ \sum_{x \in [x_{t+1}, x_{t})} c_x e^{\beta_{t+1} x}} \geq \frac{ 1/5 }{ 1/80 } e^{(\beta_{t} - \beta_{t+1}) x_{t+1}} 
\end{equation}
where the last inequality holds since  $x_{t+1}$ is the smallest element in $[x_{t+1}, x_t)$. Alternatively, we can estimate:
\begin{equation}
\label{fvp2}
\frac{Z(\beta_{t})}{Z(\beta_{t+1})} = \frac{ \mu_{\beta_{t+1}}[0, x_{t-1}) }{ \mu_{\beta_{t}}[0, x_{t-1}) } \cdot \frac{\sum_{x < x_{t-1}} c_x e^{\beta_{t} x}}{ \sum_{x < x_{t-1}} c_x e^{\beta_{t+1} x}} \leq \frac{ 1 }{ 1/4  } e^{(\beta_{t} - \beta_{t+1}) x_{t-1}} 
\end{equation}
where again the last inequality holds since every element in $[0,x_{t-1})$ is smaller than $x_{t-1}$. 

Inequalities (\ref{fvp1}) and (\ref{fvp2}) together show $16 e^{(\beta_{t} - \beta_{t+1}) x_{t+1}} \leq 4 e^{(\beta_{t} - \beta_{t+1}) x_{t-1}}$, i.e. $\beta_{t} - \beta_{t+1} \geq \frac{\log 4}{x_{t-1} - x_{t+1}}$. Substituting into~(\ref{fvp1}) and take logarithm gives the bound
\begin{equation}
\label{fvp3}
z(\beta_{t+1}, \beta_{t}) \geq \log 16 + \frac{x_{t+1} \log 4}{x_{t-1} - x_{t+1}}\qquad \text{for all $t = 1, \dots, T - 1$} 
\end{equation}

Now let $g =T/2 - 2= \lceil \min \{q/2, \sqrt{q \log n} \} \rceil $, and consider the sequence $b_i = x_{T - 2i}$ for $i = 1, \dots g+1$. From Eq.~(\ref{fvp3}), and only including the subset of terms, we have
\begin{equation}
\label{fvp4}
\sum_{t=1}^{T-1} z(\beta_{t+1}, \beta_t) \geq \sum_{i=1}^g \bigl( \log 16 + \frac{b_i \log 4}{b_{i+1} - b_i} \bigr)
\end{equation}

As we show in Lemma~\ref{lemma:incrseq} in Appendix~\ref{sec:FindSegment}, since $b_1, \dots, b_{g+1}$ is a strictly-increasing positive sequence,  we have $\sum_{i=1}^g \frac{b_i}{b_{i+1} - b_i} \geq \frac{g^2}{\log(b_{g+1}/b_1)} - g/2$. Here $b_{g+1} = x_2 \leq n$, and $b_1 = x_{T-2} \geq 1$. So overall we have shown that $$
\sum_{t=1}^{T-1} z(\beta_{t+1}, \beta_t) \geq g \log 16 + \frac{g^2 \log 4}{\log n}- \frac{g \log 4}{2} > \frac{g^2}{\log n} + 2 g.
$$
 Since $\sum_{t=1}^{T-1} z(\beta_{t+1}, \beta_t) \leq z(\betamin, \betamax) \leq q$, this implies that $g < \sqrt{q \log n}$ and $g < q/2$, which is a contradiction to the definition of $g$.
\end{proof}

\begin{proposition}
\label{all-good-prop}
Suppose the bounds of Proposition~\ref{rlistprop} hold. Then with probability at least $1 - \gamma/4$,  the preconditions of Lemma~\ref{general-pcoef-lemma} for estimating $\pi$ hold for all $y \in \mathcal F$.
\end{proposition}
\begin{proof}
Our assumption that the call at Line 1 is good immediately gives property (A1) for all $y$. 

Consider some iteration $t$ and $y \in (x_t, x_{t-1}]$, where we denote $x_t = -\infty$ for $\beta_t = \betamin$.   For brevity write  $a_y = \mu_{\beta_t}(y), \hat a_y = \hat \mu_{\beta_t}(y)$. By Proposition~\ref{rlistprop}, if $\beta_t > \betamin$ we have $\mu_{\beta_t}[0, y) \geq \mu_{\beta_t}[0,x_t) \geq 1/100$. Likewise, if $\beta_t < \betamax$, we have $\mu_{\beta_t}[y,n] \geq \mu_{\beta_t}[x_{t-1},n] \geq 1/4$. So  $\beta_t \in {\tt BalancingVals}(y, 1/100)$. By Proposition~\ref{binarysearch-delta-thm} this shows $a_y \geq \Delta(y)/100$. Thus, condition (A2) will hold as long as 
\begin{equation}
\label{ayeqn}
|\hat a_y - a_y | \leq \eps/3 \cdot (a_y + \delta/100)
\end{equation}

For any $y \in \mathcal F$, let $\mathcal E_y$ denote the bad event that (\ref{ayeqn}) fails for $y$; note here that we are not necessarily requiring $y \in (x_t, x_{t-1}]$. Conditioned on $\beta_t$, we have
$$
\Pr{\mathcal E_y} \leq F(N a_y,  N \eps/3 \cdot (a_y + \delta/100)) \qquad  \text{for $N = \lceil 10^6 \log(\tfrac{10 T}{\delta \gamma}) / (\delta \eps^2) \rceil$}
$$
Now for $a_y \geq \delta$, we have $\Pr{ \mathcal E_y } \leq F(N a_y,  N a_y \eps / 3) \leq 2 e^{-N \delta \eps^2/27} \leq \frac{\delta \gamma}{4 T} \leq \frac{a_y \gamma}{4 T}$. Otherwise, for $a_y < \delta$, we use the bound $F(u, \tau) \leq  2 (1 + \tau/u)^{-\tau/2}$. Letting $v = \eps \delta / 300$, this gives
    \begin{equation}
    \label{eeqn1}
    \Pr{\mathcal E_y} \leq      F(N a_y,  N v)  \leq  2  (1 + v/a_y)^{-N v/2}
    \end{equation}
    
The map $x \mapsto (1 + v/x)^{-N v/2}$ has second derivative $\frac{N v^2 (v (N v-2)-4 x) \left(\frac{v+x}{x}\right)^{-N v/2}}{4 x^2 (v+x)^2}$    which is nonnegative for $x \leq (N v - 2) v / 4 \leq \delta ( \frac{25}{9} \log(\frac{10 T}{\delta \gamma}) - \frac{\eps}{600} )$. In particular,  it is concave-up for $x < \delta$. So for $a_y \leq \delta$, we can upper-bound the RHS of (\ref{eeqn1}) by its secant line from $0$ to $\delta$, i.e.
   $$
   F(N a_y, N v) \leq \frac{2 a_y}{\delta} (1 + v/\delta)^{-N v/2} \leq \frac{2 a_y}{\delta}  (1 + \eps/300)^{-\frac{5000}{3 \eps} \log \frac{10 T}{\delta \gamma}}
   $$
   Simple calculations show that $(1 + \eps/300)^{-\frac{5000}{3 \eps}} \leq 0.004$, so this is at most $\frac{2 a_y}{\delta} \cdot 0.004^{ \log \frac{10 T}{\delta \gamma}} \leq \frac{a_y  \gamma}{4 T}$. 
   
   Accordingly, we have shown in both cases that $\Pr{\mathcal E_y} \leq \frac{a_y \gamma}{4 T}$. By a union bound over all such $y$, the total probability that (A2) fails in iteration $t$ is at most
   \[
   \sum_y \Pr{ \mathcal E_y } \leq \sum_y \frac{a_y \gamma}{4 T} = \frac{\gamma}{4 T} \sum_y \mu_{\beta_t}(y) = \frac{\gamma}{4 T}.
   \]
The result follows by a union bound over the $T$ iterations.
 \end{proof}

Overall, the total failure probability is at most $\gamma/4$ (from Line 1) plus  $\gamma/2$ (from Proposition~\ref{rlistprop}) plus $\gamma/4$ (from Proposition~\ref{all-good-prop}).  This concludes the proof of Theorem~\ref{pcoefalg1thm} and the first part of Theorem~\ref{th:main:three1}.

\section{Solving $\Pcoef{}$ for integer-valued Gibbs distributions}
\label{sec:pratio2}

The integer-setting algorithms hinge on a data structure called the \emph{covering schedule}. Formally, we define a covering schedule to be a sequence of the form 
$$
\mathcal I = (\beta_0, w_0, k_1, \beta_1, w_1, k_2, \dots, \beta_{t-1}, w_{t-1}, k_t, \beta_t, w_t)
$$
which satisfies the following  constraints: \\
 (a)~$\betamin=\beta_0 < \ldots < \beta_t=\betamax$; \\
 (b)~$k_1 < k_2 < \dots < k_t$ with $k_1, \dots, k_t \in \mathcal H$;\\
(c) $w_i \in [0,1]$ for $i = 0, \dots, t$.

Note that $t \le n+1$.  We say that $\mathcal I$ is \emph{proper} if for all $i = 1, \dots, t$ it satisfies 
$$
\mu_{\beta_{i-1}}(k_i) \geq w_{i-1} \text{ and } \mu_{\beta_i}(k_i) \geq w_i.
$$

We define the inverse weight of $\mathcal I$, denoted ${\tt InvWt}(\mathcal I)$ by $$
{\tt InvWt}(\mathcal I) =\sum_{i=0}^t  \frac{1}{w_i}.
$$

It is quite involved to actually generate such a covering schedule, so we defer the technical details to Section~\ref{sec:buildschedule} where we show the following:
\begin{theorem}\label{th:schedule-existenceg}
Procedure ${\tt FindCoveringSchedule}(\gamma)$ produces a covering schedule $\mathcal I$, which is proper with probability at least $1 - \gamma$.

\noindent In general integer setting, it has cost  $O(n (\log^3 n +  \log n \log \tfrac{1}{\gamma} + \log q))$ and gives ${\tt InvWt}(\calI)\le O( n \log n )$.

\noindent  In the log-concave setting, it has cost $O(n(\log^2 n + \log \tfrac{1}{\gamma} + \log q))$ and gives ${\tt InvWt}(\calI)\le O( n  )$.
\end{theorem}

For the moment, let us suppose we are given such a covering schedule. There are then two main steps to solve $\Pcoef{\delta, \eps}$. First,  we use a telescoping-product calculation similar to the  Paired Product Estimator to estimate the values $Q(\beta_i)$ for $i = 0, \dots, t$. Next, we use these estimates $\hat Q(\beta_i)$ to interpolate the counts $c_i$. 

Via Theorem~\ref{th:main:one-restate2},  one may use the solution for $\Pcoef{\delta, \eps}$ with $\delta = 1/n$ to also solve  $\PratioAll{}$.

  We now describe these steps in detail. First, we develop the algorithm {\tt PratioCoveringSchedule} to estimate the values $Q(\beta_i)$ in the covering schedule.

\begin{algorithm}[H]
\DontPrintSemicolon
\If{$n \cdot {\tt InvWt(\mathcal I)} < q \log n$}{
\textbf{for $i = 1, \dots t$} form random vars $X_i \sim \text{Bernoulli}(\mu_{\beta_{i-1}}(k_i))$ and $Y_i \sim \text{Bernoulli}(\mu_{\beta_i}(k_i))$ \\
set $\hat X^{\tt prod} \leftarrow {\tt EstimateProducts}(X, {\tt InvWt(\mathcal I)} \cdot 4/\eps^2, \gamma/2)$ \\
set $\hat Y^{\tt prod} \leftarrow {\tt EstimateProducts}(Y, {\tt InvWt(\mathcal I)} \cdot 4/\eps^2, \gamma/2)$ \\
\textbf{for $i=0, \dots, t$ do} set $\hat Q(\beta_i) = \exp \bigl( \sum_{j=1}^i (\beta_j - \beta_{j-1}) k_j \bigr) \cdot \hat X^{\tt prod}_i / \hat Y^{\tt prod}_i$ \\
} \Else {
call $\mathcal D \leftarrow  {\tt PratioContinuous}(\eps, \gamma)$ \\
\textbf{for $i = 0, \dots, t$ do} set $\hat Q(\beta_i) = \hat Q(\beta_i \mid \mathcal D)$
}
\caption{Algorithm ${\tt PratioCoveringSchedule}(\mathcal I, \eps, \gamma)$ for a covering schedule $\mathcal I$. \label{alg:pratiodiscrete}}
\end{algorithm}

\begin{theorem}
The algorithm ${\tt PratioCoveringSchedule}(\mathcal I, \eps, \gamma)$ has cost $O\bigl( \frac{ \min \{ n W, q \log n, q^2 \} }{\eps^2}   \log \tfrac{1}{\gamma} \bigr)$ for $W = {\tt InvWt}(\mathcal I)$.  If $\mathcal I$ is proper, then with probability at least $1 - \gamma$ it satisfies $$
| \log \hat Q(\beta_i) - \log Q(\beta_i)| \leq \eps \qquad \qquad \text{for all $i = 0, \dots, t$.}
$$
 (When this holds, we say that the call to {\tt PratioCoveringSchedule} is \emph{good}).
\end{theorem}
\begin{proof}
The cost and correctness are clear where $n W \geq q \log n$. So suppose $n W < q \log n$.

Since $t \leq n+1$, Lines $3$ and $4$ have cost $O( \frac{n W}{\eps^2} \log \tfrac{1}{\gamma})$.
Assuming $\mathcal I$ is proper, we have  $\Vrel(X_i) = \frac{1}{\mu_{\beta_{i-1}}(k_i)} - 1 \leq \frac{1}{w_{i-1}}$ for all $i$, so $\sum_i \Vrel(X_i) \leq W$. Likewise $\sum_i \Vrel(Y_i) \leq W$. So with probability at least $1 - \gamma$, all estimates $\hat X^{\tt prod}_i, \hat Y^{\tt prod}_i$ are within $e^{\pm \eps/2}$ of $\prod_{j=1}^i \E[X_j], \prod_{j=1}^i \E[Y_j]$ respectively.    Observe that
$$
\frac{ \E[ \prod_{j=1}^{i} X_j ] }{ \E[ \prod_{j=1}^{i} Y_j ]}= \prod_{j=1}^{i-1} \frac{  \mu_{\beta_{j-1}}(k_j) }{ \mu_{\beta_{j}}(k_j) } = \prod_j e^{(\beta_{j-1}  - \beta_j) k_j } \frac{Z(\beta_j)}{Z(\beta_{j-1})}  = Q(\beta_i) \exp \Bigl( \sum_{j=1}^i (\beta_{j-1}  - \beta_j) k_j \Bigr)
$$
so in that case, the values $\hat Q(\beta_i)$ are also within $e^{\pm \eps}$ of $Q(\beta_i)$ as required.
\end{proof}

We next turn to using these estimates to finish solving the problem  $\Pcoef{\delta, \eps}$. We develop two separate algorithms: one for the general integer setting, and a second optimized for the log-concave setting. 

\subsection{Solving $\Pcoef{\delta, \eps}$ in the general integer setting}
\label{sec:pratio2:2}

\begin{algorithm}[H]
\DontPrintSemicolon
set $\mathcal I = ( \beta_0, w_0, k_1, \dots, k_t, \beta_t, w_t) \leftarrow {\tt FindCoveringSchedule}(\gamma/4)$  \\
set $(\hat Q(\beta_0), \dots, \hat Q(\beta_t)) \leftarrow {\tt PratioCoveringSchedule}(\mathcal I, \eps/9, \gamma/4)$  \\
{\bf for} $i = 0, \dots, t$ {\bf do} let $\hat\mu_{\beta_i}\leftarrow {\tt Sample}(\beta_i; \eps/9,w_i, \frac{\gamma}{10 (n+1)^2}   )$ \\
\For{$j = 0, \dots, n$} {
  set $\alpha \leftarrow  {\tt Balance}(j,  1/4, \frac{\gamma}{4 (n+1)^2})$ \\
  find index $i \in \{0, \dots, t-1 \}$ with $\alpha \in [\beta_i, \beta_{i+1}]$  \\  
  let $\hat \mu_{\alpha} \leftarrow {\tt Sample}(\alpha; \eps/9, \delta/4, \frac{\gamma}{4(n+1)^2})$  \\
\textbf{if $\hat \mu_{\alpha}( k_{i+1} ) \geq \delta/4$ then} set $\hat \pi(j) = \frac{ \hat \mu_{\beta_i}(k_{i+1}) }{\hat\mu_{\alpha}(k_{i+1})} e^{(\alpha - \beta_i) k_{i+1}} \hat Q(\beta_i) \cdot e^{-\alpha j} \hat \mu_{\alpha}(j)$ \\
\textbf{else if $j < k_{i+1}$ \ \  then} set $\hat \pi(j) = \hat Q(\beta_i) e^{-\beta_i j} \hat \mu_{\beta_i}(j)$ \\
\textbf{else if $j \geq k_{i+1}$ \ \ then} set $\hat \pi(j) = \hat Q(\beta_{i+1}) e^{-\beta_{i+1} j} \hat \mu_{\beta_{i+1}}(j)$ 
}
\caption{Solving problem $\Pcoef{\delta,\eps}$ for the general integer setting.  \label{alg:basedonn1}}
\end{algorithm}

\begin{theorem}\label{th:Problem2:general1}
In the integer setting, Algorithm~\ref{alg:basedonn1} solves $\Pcoef{\delta, \eps}$ with cost
$$
O \Bigl(  \frac{ (n/\delta) \log \tfrac{n}{\gamma} + n^2 \log n \log \tfrac{1}{\gamma} }{\eps^2} + n \log q  \Bigr)
$$
\end{theorem}
\begin{proof}
 Lines  1 and 2 have cost $O( \frac{n^2 \log n \log(1/\gamma)}{\eps^2}+ n \log q)$.  The sampling for $\beta_i$ in Line 3 has cost $O( \frac{\log(n/\gamma)}{w_i \eps^2} )$; summing over $i$ gives $ O( \frac{{\tt InvWt}(\mathcal I)}{\eps^2} \log \tfrac{n}{\gamma}) \leq O( \frac{ n \log  n}{\eps^2}  \log \tfrac{n}{\gamma})$. Line 5 has cost $O(n \log \tfrac{n q}{\gamma})$. Line 7 has cost $O(\frac{n }{\delta \eps^2}  \log \tfrac{n}{\gamma} )$. This shows the bound on complexity. 
  
  For correctness,  assume $\mathcal I$ is proper and  the calls to \texttt{Balance} and {\tt PratioCoveringSchedule} are good.  Also, assume that Lines 3 and 7 well-estimate every value in $\mathcal H$. By specification of subroutines, these events all hold with probability at least $1 - \gamma$.   In this case, we claim that the preconditions of Lemma~\ref{general-pcoef-lemma} hold for all $j \in \mathcal H$. To show this, consider $j \in \mathcal H$, and let $\alpha, i$ denote the corresponding parameters at Lines 5 and 6 respectively, and define $k = k_{i+1}, \beta = \beta_i$ for brevity.  By Proposition~\ref{binarysearch-delta-thm}, we have $\mu_{\alpha}(j) \geq \Delta(j)/4$.  There are two cases:
\begin{itemize}
\item Suppose $\hat \mu_{\alpha}(j) \geq \delta/4$. For (A2), Line 7 well-estimates $j$, so $$
| \hat \mu_{\alpha}(j) - \mu_{\alpha}(j) | \leq \eps/9 \cdot( \mu_{\alpha}(j) + \delta / 4) \leq  \eps/9 \cdot \mu_{\alpha}(j) (1 + \delta / \Delta(j)).$$

   For (A1), observe that $Q(\alpha) = \frac{  \mu_{\beta}(k) }{\mu_{\alpha}(k)} e^{(\alpha - \beta) k}  Q(\beta)$. Line 3 well-estimates $j$ and $\mu_{\beta}(j) \geq w_i$, so $| \log \hat \mu_{\beta}(j) - \log \mu_{\beta}(j) | \leq  \eps/9$.  Line 7 well-estimates $k$ and $\hat \mu_{\alpha}(k) \geq \delta/4$, so $|\log \hat \mu_{\alpha}(k) - \log \mu_{\alpha}(k) | \leq \eps/9$.  Also, $| \log \hat Q(\beta) - \log Q(\beta)| \leq \eps/9$. Overall $| \log \hat Q(\alpha) - \log Q(\alpha) | \leq \eps/3$ where $\hat Q(\alpha) = \frac{ \hat \mu_{\beta}(k_{i+1}) }{\hat\mu_{\alpha}(k_{i+1})} e^{(\alpha - \beta) k_{i+1}} \hat Q(\beta)$.

\item Suppose $\hat \mu_{\alpha}(j) < \delta/4$. Let us assume $j < k$; the case with $j \geq k$ is completely analogous.  (A1) holds since $|\log \hat Q(\beta) - \log Q(\beta)| \leq \eps/9$.  Since Line 7 well-estimates $k$ and $\hat \mu_{\alpha}(k) < \delta/4$ we have $\mu_{\alpha}(k) < e^{\eps/9} \delta/4 \leq \delta / 2$.  We have already seen that $\mu_{\alpha}(j) \geq \Delta(j)/4$ and since $\mathcal I$ is proper, we have $\mu_{\beta}(k) \geq w_{i}$.   So from the Uncrossing Inequality we get $
\mu_{\beta}(j) \geq \frac{ \mu_{\beta}(k)  \mu_{\alpha}(j)  }{\mu_{\alpha}(k) } \geq  \frac{ w_{i} \Delta(j)/4 }{ \delta/2} = \frac{ w_{i} \Delta(j) }{ 2  \delta}.$ Now Line 7 well-estimates $j$  so  \[
|\hat \mu_{\beta}(j) - \mu_{\beta}(j)| \leq \eps/9 \cdot (\mu_{\beta}(j) + w_{i}) \leq  \eps/9 \cdot (\mu_{\beta}(j) + \frac{2 \delta \mu_{\beta}(j)}{\Delta(j)}) \leq 2 \eps/9 \cdot \mu_{\beta}(j) (1 + \frac{\delta}{\Delta(j)}).   \qedhere 
\]
\end{itemize}
\end{proof}

With some simplification of parameters, this gives the second part of Theorem~\ref{th:main:three1}; via Theorem~\ref{th:main:one-restate2}, it also gives the second part of Theorem~\ref{th:main:one}.

\subsection{Solving $\Pcoef{\delta, \eps}$ in the log-concave setting}

\label{sec:pratio2:3}

\begin{algorithm}[H]
\DontPrintSemicolon
set $\mathcal I = ( \beta_0, w_0, k_1, \dots, k_t, \beta_t, w_t) \leftarrow {\tt FindCoveringSchedule}(\gamma/3)$  \\
set $(\hat Q(\beta_0), \dots, \hat Q(\beta_t)) \leftarrow {\tt PratioCoveringSchedule}(\mathcal I, \eps/3, \gamma/3)$  \\

set $w_0 \leftarrow \min\{   w_0, \delta, 1/n \}$ and set $w_t \leftarrow  \min\{ w_t, \delta, 1/n \}$ \\
{\bf for} $i = 0, \dots, t$ {\bf do} let $\hat\mu_{\beta_i}\leftarrow {\tt Sample}(\beta_i; \eps/9,w_i,\frac{\gamma}{3 (n+1)^2} )$ \\

\For{$i=1, \dots t-1$} {
\textbf{for} $j = k_i+1, \dots, k_{i+1}$ \textbf{do} set $\hat \pi(j) = \hat Q(\beta_i) e^{-\beta_i j} \hat \mu_{\beta_i}(j)$ 
}
\textbf{for} $j = 0, \dots, k_1$   \textbf{do} \ \ \ \ \ \ \  \ \ \ \  \ \ \ \negthinspace set $\hat \pi(j) = \hat Q(\betamin) e^{-\betamin j} \hat \mu_{\betamin}(j)$ \\
 \textbf{for} $j = k_t+1, \dots, n$ \textbf{do} \ \ \ \ \ \ \  \ \ \negthinspace set $\hat \pi(j) = \hat Q(\betamax) e^{-\betamax j} \hat \mu_{\betamax}(j)$

\caption{\sloppy Solving $\Pcoef{\delta, \eps}$ in the log-concave  setting.  \label{alg:Problem2:concave}} 
\end{algorithm}

\begin{theorem}
\label{th:Problem2:concave1}
In the log-concave setting, Algorithm~\ref{alg:Problem2:concave} solves $\Pcoef{\delta, \eps}$ with cost
$$
O \Bigl( n \log^2 n+n \log q+  \frac{ \min\{n^2, q \log n\} \log \frac{1}{\gamma} + (n + 1/\delta) \log \frac{n}{\gamma} }{\eps^2} \Bigr)
$$
\end{theorem}
\begin{proof}
Note that ${\tt InvWt}(\mathcal I) \leq O(n)$. So Line 1 has cost $O(n(\log^2 n+\log q+ \log\frac{1}{\gamma}))$, and Line 2 has cost $O( \frac{ \min\{n^2, q \log n\} }{\eps^2} \log \frac{1}{\gamma}  )$. Because of the modification step at Line 3, the total cost of Line 4 is $O \bigl( ( \frac{1}{\delta} + n + \sum_{i=0}^t \frac{1}{w_i}) \eps^{-2} \log \tfrac{n}{\gamma} \bigr) = O( \frac{ (n + 1/\delta) }{\eps^2} \log \frac{n}{\gamma})$. These add to the stated complexity.

For correctness, suppose $\mathcal I$ is proper and the call to {\tt PratioCoveringSchedule} is good and every value in $\mathcal H$ is well-estimated by every iteration of Line 4. By specification of the subroutines, these events hold with probability at least $1 - \gamma$.  We claim that the preconditions of Lemma~\ref{general-pcoef-lemma} then hold for any $j \in \mathcal H$.  There are three cases to consider.

\begin{itemize}
\item Suppose $j \in (k_i, k_{i+1}]$ is estimated at Line 6. Let $a = \mu_{\beta_i}(j), \hat a = \hat \mu_{\beta_i}(j)$.      Condition (A1) holds since $|\log \hat Q(\beta_i)  - \log Q(\beta_i)| \leq \eps/3$.   Note that $a \ge \min\{\mu_{\beta_i}(k_i), \mu_{\beta_i}(k_{i+1}) \} \geq w_i$, where the first inequality follows from log-concavity of counts and the second inequality holds from properness of $\mathcal I$. Since Line 4 well-estimates $j$, this implies $|a - \hat a| \leq \eps/6 \cdot (a + w_i) \leq \eps/3 \cdot a$, giving (A2). 

\item Suppose $j \leq k_1$ is estimated at Line 7,  and let $a = \mu_{\betamin}(j), \hat a = \hat \mu_{\betamin}(j)$. Again, condition (A1) is immediate.  Since Line 4 well-estimates $j$, we have $| \hat a - a | \leq  \eps/6 \cdot ( a + w_0' )$ where $w_0' = \min\{w_0, \delta, 1/n \}$ is the modified value and $w_0$ is the original value from $\mathcal I$. To show (A2), it thus suffices to show that $a \geq w_0'$.

 If $j$ is on the ``decreasing'' slope of the log-concave distribution $\mu_{\betamin}$, then properness of $\mathcal I$ implies that $a \geq \mu_{\betamin}(k_1) \geq w_0$ which immediately gives (A2).  So suppose $j$ is on the increasing slope of $\mu_{\betamin}$, and hence $\mu_{\betamin}(\ell)\le \mu_{\betamin}(j)=a$ for all $ \ell \leq j$. So $\mu_{\betamin}[0, j-1]\leq j a$.   
 
 Let $\Delta(j) = \mu_{\alpha}(j)$ for $\alpha \in [\betamin, \betamax]$.  For $\ell \geq j$, the Uncrossing Inequality gives  $\mu_{\betamin}(\ell) \le \frac{\mu_{\alpha}(\ell)  \mu_{\betamin}(j)}{\mu_{\alpha}(j)} =  \frac{ \mu_{\alpha}(\ell) a }{\Delta(j)}$,  so $\mu_{\betamin}[j, n]\le \frac{a}{\Delta(j)} \cdot\mu_{\alpha}[j,n]\le \frac{a}{\Delta(j)}$.   So $
 1=\mu_{\betamin}[0,j-1]+\mu_{\betamin}[j,n] \leq j a + \frac{a}{\Delta(j)}$,
 implying $a \geq \frac{\Delta(j)}{1 + j \Delta(j)}$. 

So $a (1 + \delta/\Delta(j)) \geq  \frac{\delta + \Delta(j)}{ 1 + j \Delta(j)} \geq \min \{ \delta, 1/j \} \geq w_0'$. This establishes (A2). 
 
\item Suppose $j > k_{t+1}$ is estimated at Line 8. This is completely symmetric to the previous case. \qedhere
\end{itemize} 
\end{proof}

Again, with some simplification of parameters, this gives the third part of Theorem~\ref{th:main:three1}; via Theorem~\ref{th:main:one-restate2}, it also gives the third part of Theorem~\ref{th:main:one}.

\section{Constructing a covering schedule}\label{sec:alg:schedule}
\label{sec:buildschedule}
We will show the following more precise bound on the weight of the schedule. 
\begin{theorem}\label{th:schedule-existence}
Let $a > 4$ be an arbitrary constant and define the following parameter\footnote{As motivation for the definition of $\rho$, see Lemma~\ref{lemma:logconcave:harmonic}.}
$$
\rho \eqdef \begin{cases}
 1+\log (n+1) & \mbox{in the general integer setting} \\
 e & \mbox{in the log-concave  setting} 
\end{cases}
$$

 In either the general integer or log-concave setting, the procedure ${\tt FindCoveringSchedule}(\gamma)$ produces a covering schedule $\mathcal I$ with
 ${\tt InvWt}(\calI)\le a (n+1) \rho$, which is proper with probability at least $1 - \gamma$. It has cost $O(n\rho(\log^2 n+ \log\frac{1}{\gamma})+n\log q)$.
\end{theorem}

This will immediately imply Theorem~\ref{th:schedule-existenceg}. The construction has two parts. First, in Section~\ref{sec:construct-pre}, we build an object with relaxed constraints called a \emph{preschedule}. Then in Section~\ref{sec:convert}, we convert this into a schedule.

\subsection{Constructing a preschedule} 
\label{sec:construct-pre}
We introduce basic terminology and definitions.
 
 An $\mathcal H$-interval is a discrete set of points $\{ \sigma^-, \sigma^-+1, \dots, \sigma^+-1, \sigma^+ \}$, for integers $0 \leq \sigma^- \leq \sigma^+ \leq n$. We also write this more compactly as $\sigma = [\sigma^-, \sigma^+]$; note that the set $\sigma$ has cardinality  $|\sigma| =  \sigma^+ - \sigma^- + 1$.
  
A {\em segment} is a tuple $\theta = (\beta,\sigma)$ where $\beta\in[\betamin,\betamax]$, and $\sigma$ is an $\mathcal H$-interval.  We say $\theta$ is {\em $\phi$-proper} for parameter $\phi$ (or just proper if $\phi$ is understood) if it satisfies the following two properties:
\begin{itemize}
\item Either $\beta = \betamin$ or $\mu_{\beta}(\sigma^-) \geq \phi/|\sigma|$
\item Either $\beta = \betamax$ or $\mu_{\beta}(\sigma^+) \geq \phi/|\sigma|$
\end{itemize}

A \emph{preschedule} is a sequence of segments $\calJ=((\beta_0,\sigma_0),\ldots,(\beta_t,\sigma_t))$ with  the following properties:
\begin{itemize}
\item[{({\tt I0})}] $\sigma_{i+1}^- \leq \sigma_i^+$ for $i = 0, \dots, t-1$.
\item[{({\tt I1})}] $\betamin = \beta_0\le\ldots\le\beta_t = \betamax$. 
\item[{({\tt I2})}] $0 = \sigma^-_0\le \ldots \le\sigma^-_t\le n$ and $0\le\sigma^+_0\le \ldots \le\sigma^+_t = n$
\end{itemize}
We say that $\cal J$ is {\em $\phi$-proper} if all segments $\theta_i$ are $\phi$-proper. 

\medskip

Let us fix constants $\tau\in(0,\frac{1}{2})$, $\lambda > 1$, and define parameter $$
\phi= \frac{\tau}{\lambda^3 \rho}.$$
 Thus, $\phi=\Theta(\frac 1{\log n})$ in the general setting and $\phi=\Theta(1)$ in the log-concave setting.  The main idea of the algorithm is to maintain a sequence of proper segments satisfying properties ({\tt I1}) and ({\tt I2}),  and grow it until it satisfies ({\tt I0}). 
Details are provided below:

\begin{algorithm}[H]
\DontPrintSemicolon
call $\theta_{\min}\,\leftarrow {\tt FindSegment}(\betamin,\,[0,0],[0,n])$ and  $\theta_{\max}\leftarrow {\tt FindSegment}(\betamax,[0,n],[n,n])$ \\
initialize $\mathcal J$ to contain the two segments $\theta_{\min}, \theta_{\max}$ \\
\While{$\mathcal J$ does not satisfy {\tt(I0)}} {
pick arbitrary consecutive segments $\theta_{\lft} = (\beta_{\lft},\sigma_{\lft})$ and $\theta_{\rgt} = (\beta_{\rgt},\sigma_{\rgt})$ in $\mathcal J$ with $\sigma_{\lft}^+ < \sigma_{\rgt}^-$. \\
let $M = \big \lfloor \frac{\sigma_{\lft}^+ + \sigma_{\rgt}^-}{2} \big \rfloor + \frac{1}{2}$. \\
	call $\beta\leftarrow{\tt Balance}(M, \tau, \frac{1}{4 n})$ \\
	call
	$
		\theta \leftarrow\begin{cases}		
			{\tt FindSegment}(\beta,  \ \ \ \ \ \thinspace [\sigma^-_{\lft},M - \tfrac12], \thinspace [M + \tfrac12,\sigma^+_{\rgt}]) & \mbox{if } \beta_{\lft}<\beta<\beta_{\rgt} \\
			{\tt FindSegment}(\beta_{\lft}, \ [\sigma^-_{\lft}, \sigma^-_{\lft}], \ \ \thinspace  [M + \tfrac12,\sigma^+_{\rgt}] )& \mbox{if }\beta \leq \beta_{\lft} \\
			{\tt FindSegment}(\beta_{\rgt},[\sigma^-_{\lft},M - \tfrac12], [\sigma^+_{\rgt}, \sigma^+_{\rgt}] ) & \mbox{if }\beta \geq \beta_{\rgt} 
		\end{cases}
	$ \\
	insert $\theta$ into $\mathcal J$ between $\theta_{\lft}$ and $\theta_{\rgt}$
}
\textbf{return} $\mathcal J$

\BlankLine

\BlankLine
\nonl \ \ \  \textbf{Subroutine ${\tt FindSegment}(\beta,[h_{\lft}, a_{\lft}],[a_{\rgt}, h_{\rgt}])$: } \\
\BlankLine

\DontPrintSemicolon
let $\hat\mu_\beta\leftarrow{\tt Sample}(\beta; \tfrac{1}{2} \log \lambda, \frac{\phi}{h_{\rgt} - h_{\lft} + 1}, \frac{1}{4(n+2)^2})$. \\
\vspace{0.05in}
\textbf{for $i = h_{\lft}$ to $h_{\rgt}$} set $\Phi(i)=\begin{cases}
\hat \mu_\beta(i) & \mbox{if }i \notin \{h_{\lft},h_{\rgt}\} \\
\lambda \cdot \hat\mu_\beta(i) & \mbox{if }i\in \{h_{\lft},h_{\rgt}\} \\
\end{cases}
\vspace{0.05in}
$ \\
set $k^- = \negthickspace \displaystyle \argmax_{i \in [h_{\lft}, a_{\lft}]}  (a_{\lft} - i + 1)\Phi(i)$  and $k^+ = \negthickspace \displaystyle \argmax_{i \in [a_{\rgt}, h_{\rgt}]} (i -  a_{\rgt} + 1) \Phi(i)$ \\
\textbf{return} segment $\theta = (\beta, \sigma)$ where $\sigma = [k^-, k^+]$

\caption{Computing an initial preschedule. \label{alg:schedule}}
\end{algorithm}

At each iteration, we say that $\mathcal J$ has a \emph{gap} $(\sigma_i^+, \sigma_{i+1}^-)$ whenever $\sigma_i^+ < \sigma_{i+1}^-$ strictly. The interval $G =  (\sigma^+_{{\lft}} , \sigma^-_{{\rgt}} )$ is one of these gaps. From the new segment produced by {\tt FindSegment}, this gap $G$ gets replaced at the next iteration by two new gaps $G', G''$. Overall, the set of gaps over all the iterations forms a laminar family of intervals. Between gaps, the ``filled'' regions have strictly increasing values of $\beta$.  Thus at Line 4 we have $\beta_{\tt left} < \beta_{\tt right}$ strictly.

The loop in Lines 3 -- 8 is executed at most $n$ times, since each time it covers a new half-integer value $M$. So there are at most $n+2$ calls to ${\tt FindSegment}$ and at most $n$ calls to ${\tt Balance}$.

\begin{proposition} \label{prop:schedule:post-analysis}
Algorithm~\ref{alg:schedule} has cost $O(n \log q + n \rho \log^2 n)$.
\end{proposition}
\begin{proof}
 By Theorem~\ref{th:subroutines2x}, the ${\tt Balance}$ subroutines have cost 
$O(n \log (n q))$. The two calls to {\tt FindSegment} at Line 1 have cost $O( n \rho \log n)$. Let $\calJ_i$ be the sequence and $M_i,\sigma_{{\lft}, i}, \sigma_{{\rgt},i}$ be the variables at 
the $i^{\text{th}}$ iteration and let $A_i=[\sigma_{{\lft}, i}^-,  \sigma_{{\rgt},i}^+]$ and $G_i = (\sigma_{\lft,i}^+, \sigma_{\rgt,i}^-)$. Observe that the call to ${\tt FindSegment}$ at the $i^{\text{th}}$ iteration
has cost $O(|A_i| \rho \log n  )$. We now show that $\sum_{i} | A_i  | =O(n\log n)$, which will yield the claim about the complexity.

For each  $k\in\calH$ define $I^-_k=\{i \::\:k\in A_i\;\wedge\;M_i<k\}$,  and consider $i,j\in I^-_k$ with $i<j$.   We claim that intervals $G_i, G_j$ overlap. For, suppose $G_j$ comes strictly after $G_i$. Since $G_j$ is always selected from a gap at iteration $j$, it cannot intersect with $\sigma_{\rgt, i}$, and in particular $\sigma_{\lft, j}^+ \geq \sigma_{\rgt, i}^+$. In this case, by definition of $M_j$ and $I_k$, we would have $k > M_j  > \sigma_{\lft, j}^+$ and $k  \leq \sigma_{\rgt, i}^+$, a contradiction.  

On the other hand, suppose $G_j$ comes strictly before $G_i$. Again, since $G_j$ is selected from a gap, it comes strictly before $\sigma_{\lft, i}$ and so $\sigma_{\rgt, j}^- \leq \sigma_{\lft, i}^-$. By Property ({\tt I2}) this implies $\sigma_{\rgt, j}$ comes before $\sigma_{\lft, i}$ and so also $\sigma_{\rgt, j}^+ \leq \sigma_{\lft, i}^+$. Then by definition of $M_i$ and $I_k^-$ we would have $\sigma_{\lft, i}^+ < M_i < k$ and $k \leq \sigma_{\rgt, j}^+$, again a contradiction.

So $G_i$ and $G_j$ overlap. Since $M_i$ is chosen near the median of $G_i$, we must have $|G_j| \leq \tfrac{1}{2} |G_i|$.  As this holds for all pairs $i,j \in I^-_k$, we  conclude that $|I^-_k|\le 1 + \log_2 n$. Similarly,  $|I^+_k| \leq 1 + \log_2 n$
where $I^+_k=\{i\::\:k\in A_i\;\wedge\;M_i>k\}$. It remains to observe that $\sum_{i} |A_i| =\sum_{k\in\calH}|I^-_k\cup I^+_k|$.
\end{proof}

The final result in order to analyze Algorithm~\ref{alg:schedule}, which is technically involved, is the following:
\begin{lemma}
\label{prop:sch1}
With probability at least $1/2$, the final preschedule $\mathcal J$ is proper.
\end{lemma}

In order to show Lemma~\ref{prop:sch1}, we will suppose for the remainder of the section that all calls to {\tt Balance} are good, and that Line 10 in every call to {\tt FindSegment} well-estimates every value $j \in \mathcal H$. By specification of subroutines, these properties hold with property at least $1/2$. 
 
We will show by induction on iteration count that the following two additional invariants are preserved in this case:
 
 \begin{itemize}
\item[{({\tt I3})}] Each segment $\theta$ of $\mathcal J$ is $\phi$-proper.
\item[{({\tt I4})}] Each segment $\theta$ of $\mathcal J$ satisfies two ``extremality'' conditions:
\begin{subequations}
\begin{eqnarray}
\mu_\beta(k)&\le& \lambda \cdot \frac{ |\sigma|}{|\sigma|+(\sigma^--k)}\cdot\mu_\beta(\sigma^-) \qquad\quad \forall k\in\{0,\ldots,\sigma^--1\} \label{eq:extremal:a} \\
\mu_\beta(k)&\le& \lambda \cdot \frac{ |\sigma|}{|\sigma|+(k-\sigma^+)}\cdot\mu_\beta(\sigma^+) \qquad\quad \forall k\in\{\sigma^+\!+\!1,\ldots,n\} \label{eq:extremal:b} 
\end{eqnarray}
\end{subequations}
\end{itemize}

Note that {\tt FindSegment} gives a slight bias to the endpoints $h_{\lft}$ or $h_{\rgt}$, in the definition of $\Phi$ at Line 11. This is needed to preserve the slack factor $\lambda > 1$ in the definition of extremality~\eqref{eq:extremal:a},\eqref{eq:extremal:b}.  Without this bias, the factor would grow uncontrollably as the algorithm progresses. 

\bigskip

Now consider some call to {\tt FindSegment}, either at Line 1 or Line 7. By induction, properties ({\tt I3}) and ({\tt I4}) hold for all segments in $\mathcal J$ up to this point. We must show that the segment $\theta$ returned by {\tt FindSegment} also satisfies these properties. The cases when {\tt FindSegment} is called in Line 1, or in Line 7 when $\beta \in \{ \beta_{\lft}, \beta_{\rgt} \}$, are handled very differently from the main case, which is Line 7 with $\beta \in (\beta_{\lft}, \beta_{\rgt})$ strictly. In the former cases, there is no ``free choice'' for the left-margin $k^-$ or right-margin $k^+$ respectively; for instance, when $\beta = \beta_{\lft}$, then $h_{\lft} = a_{\lft} = \sigma_{\lft}^-$ and our only choice is to set $k^- = \sigma_{\lft}^-$. 

We say the call to {\tt FindSegment} at Line 1 with $\beta = \betamin$, or the call at Line 7 with $\beta = \beta_{\lft}$, is \emph{left-forced}; the call at Line 1 with $\beta = \betamax$, or at Line 7 with $\beta = \beta_{\rgt}$ is \emph{right-forced}. Otherwise it is \emph{left-free} and \emph{right-free} respectively.

\begin{lemma} \label{lemma:extremal}
 \ \\
\noindent (a) If the call is left-free, then  $   (a_{\lft} - i + 1) \mu_{\beta}(i)  \le \lambda (a_{\lft} - h_{\lft} + 1) \mu_{\beta}(h_{\lft})$ for all $i < h_{\lft}$.

\noindent (b) If the call is right-free, then $
 (i - a_{\rgt} + 1) \mu_{\beta}(i) \le  \lambda (h_{\rgt} - a_{\rgt} + 1) \mu_{\beta}(h_{\rgt})$ for all $ i > h_{\rgt}$.
\end{lemma}
\begin{proof}
We only show (a); the proof of (b) is analogous.  If {\tt FindSegment} is called at Line 1  then $h_{\lft} = 0$ and the claim is vacuous. So assume  {\tt FindSegment} is called at Line 7, and consider $i < h_{\lft}$. Since segment $\theta_{\lft} = (\beta_{\lft}, \sigma_{\lft})$ satisfies Eq.~(\ref{eq:extremal:a}) and $h_{\lft}=\sigma_{\lft}^-$,  we have
\begin{equation}
\label{gbbg1}
\mu_{\beta_{\lft}}(i)
\le \lambda \cdot \frac{ |\sigma_{\lft}|}{|\sigma_{\lft}|+(h_{\lft}-i)}\cdot\mu_{\beta_{\lft}}(h_{\lft}).
\end{equation}
Since $\beta > \beta_{\lft}$, the Uncrossing Inequality gives
$
\mu_{\beta_{\lft}}(i) \mu_{\beta}(h_{\lft}) \ge \mu_{\beta_{\lft}}(h_{\lft}) \mu_{\beta}(i)
$.
Combined with Eq.~\eqref{gbbg1}, this yields
\begin{equation*}
\mu_{\beta}(i)
\le \lambda \cdot \frac{ |\sigma_{\lft}|}{|\sigma_{\lft}|+(h_{\lft}-i)}\cdot\mu_{\beta}(h_{\lft}).
\end{equation*}

Finally, since the call is left-free at Line 7, we have $a_{\lft} \ge\sigma_{\lft}^+$, so $|\sigma_{\lft}| \le a_{\lft} + 1 -h_{\lft}$ and
\[
\frac{ |\sigma_{\lft}|}{|\sigma_{\lft}|+(h_{\lft}-i)}
\le
\frac{ a_{\lft} + 1 -h_{\lft}  }{ (a_{\lft} + 1 -h_{\lft})+(h_{\lft}-i)}
=
\frac{ a_{\lft} - h_{\lft}  + 1 }{ a_{\lft} -i + 1}.
\qedhere
\]
\end{proof}

\begin{lemma}\label{lemma:harmonic} 
 \ \\
 (a) If the call is left-free, there is $k \in [h_{\lft}, a_{\lft}]$ with $(a_{\lft} - k + 1)\cdot\mu_\beta(k)\ge \phi  \lambda^{2}$. \\
(b) If the call is right-free, there is $k \in [a_{\rgt}, h_{\rgt}]$ with $(k-a_{\rgt} + 1)\cdot\mu_\beta(k)\ge \phi \lambda^{2}$.
\end{lemma}
\begin{proof}
The two parts are completely analogous, so we only prove (a).  We first observe that
\begin{equation}
\label{mube}
\mu_{\beta} [0, a_{\lft}]  \geq \tau.
\end{equation}

This is trivial if {\tt FindSegment} is called at Line 1 with $a_{\lft} = n$ in which case $\mu_{\beta}[0, a_{\lft}]  = 1$. Otherwise, if {\tt FindSegment} is called at Line 7 and the call is left-free,  we have $\beta \in {\tt BalancingVals}(M, \tau)$ where $M = a_{\lft} + 1/2$.

Now  assume (a) is false i.e.\ $(a_{\lft} + 1-k)\cdot\mu_\beta(k)<\phi \lambda^{2} = \frac{\tau}{\lambda\rho}$ for all $k\in [h_{\lft}, a_{\lft}]$. We will use this to derive a contradiction to Eq.~(\ref{mube}). Under the assumption that (a) is false, we claim the following:
\begin{equation}\label{eq:IGAJSASKASHFKA}
(a_{\lft} + 1 - k) \mu_\beta(k) <  \tau / \rho \qquad\qquad \text{for all } k \leq a_{\lft}
\end{equation}
Indeed, we have already assumed the stronger inequality $(a_{\lft} + 1 - k) \mu_\beta(k)<\frac{\tau}{\lambda \rho}$ for $k\in [h_{\lft}, a_{\lft}]$.
In particular, we know 
\begin{equation}
\label{eq:hhgtl1}
(a_{\lft} + 1 - h_{\lft}) \mu_\beta(h_{\lft})< \frac{\tau}{\lambda \rho}
\end{equation}
For $k < h_{\lft}$, Lemma~\ref{lemma:extremal}(a) with $i = k$ gives $ (a_{\lft} + 1  - k)  \mu_{\beta}(k) \le \lambda (a_{\lft} + 1 - h_{\lft}) \mu_{\beta}(h_{\lft})$.  Combined with Eq.~(\ref{eq:hhgtl1}), this gives the bound of Eq.~(\ref{eq:IGAJSASKASHFKA}): $( a_{\lft} + 1 - k  ) \mu_{\beta}(k) <  \lambda\cdot \frac{ \tau }{\lambda\rho} = \tau/\rho$.

Now let $\ell = a_{\lft}+1$, and consider the sequence $$
b_i =\mu_{\beta}( a_{\lft} + 1 - i) \cdot \rho/\tau \qquad \qquad \text{for $i = 1, \dots, \ell$}
$$
By Eq.~(\ref{eq:IGAJSASKASHFKA}), we have $b_i \leq 1/i$ for all $i$.  We claim that $\sum_{i=1}^{\ell} b_i < \rho$. There are two cases.
\begin{itemize}
\item {\bf Log-concave setting with $\pmb{\rho=e}$}. Since the counts  are log-concave, so is the sequence $b_i$
(since $\mu_{\beta}(k)\propto c_k e^{\beta k}$). As we show in Lemma~\ref{lemma:logconcave:harmonic} in Appendix~\ref{sec:FindSegment}, this implies that $\sum_{i} b_i <  e$.
\item {\bf General setting with $\pmb{\rho=1+\log (n+1)}$}. We have $\sum_{i=1}^{\ell}b_i <  1+\log  \ell \leq 1 + \log(n+1)$ by the well-known inequality for the harmonic series.
\end{itemize}

Now observe that $\mu_{\beta}[0, a_{\lft}] = \sum_{i} \frac{\tau}{\rho} \cdot b_i < \tau$. This indeed contradicts Eq.~(\ref{mube}). 
\end{proof}

\begin{proposition}
\label{hb1prop}
 \ \\
(a) If the call is left-free, then $\mu_{\beta}(k^-)  \ge \hat \mu_{\beta}(k^-) / \sqrt{\lambda} \geq \frac{ \phi}{a_{\lft} -k^- + 1}$. \\
(b) If the call is right-free, then $\mu_{\beta}(k^+) \ge \hat \mu_{\beta}(k^+) / \sqrt{\lambda} \geq \frac{\phi}{k^+-a_{\rgt} + 1}$. 
\end{proposition}
\begin{proof}
We only prove (a); the case (b) is completely analogous.

Let $\bar p = \frac{\phi}{h_{\rgt} - h_{\lft} + 1}$.
By Lemma~\ref{lemma:harmonic}, there is $k \in [h_{\lft},a_{\lft}]$ with $\mu_\beta(k) \ge \frac{\phi \lambda^2}{a_{\lft} - k + 1} \ge \bar p$. Since Line 10 well-estimates $k$, we have $\hat \mu_{\beta}(k) \ge \mu_{\beta}(k)/\sqrt{\lambda} \ge \frac{\phi \lambda^{3/2}}{a_{\lft} + 1 - k}$. Since $k^-$ is chosen as the argmax, $(a_{\lft} + 1 - k^-)\Phi(k^-)  \ge (a_{\lft} + 1 - k)\Phi(k)  \ge \phi \lambda^{3/2}$. So $\hat \mu_{\beta}(k^-) \ge \Phi(k^-)/\lambda \ge \frac{\phi \sqrt{\lambda}}{a_{\lft} + 1 - k^-} \ge \bar p$. Since Line 10 well-estimates $k^-$, this implies $\mu_{\beta}(k^-) \ge \hat \mu_{\beta}(k^-) / \sqrt{\lambda}$. 
\end{proof}

\begin{proposition}
The segment $\theta = (\beta, \sigma)$ produced at the given iteration is $\phi$-proper.
\end{proposition}
\begin{proof}
We need to show that if $\beta > \betamin$ then $\mu_{\beta}(k^-) \geq  \frac{\phi}{|\sigma|}$ and likewise if $\beta < \betamax$ then $\mu_{\beta}(k^+) \geq \frac{\phi}{|\sigma|}$. We show the former; the latter is completely analogous.  

If the call is left-forced, and $\beta \neq \betamin$, then necessarily $\beta = \beta_{\lft}$ and $k^- = \sigma_{\lft}^-$. Since $\theta_{\lft}$ satisfies ({\tt I3}), we have $\mu_{\beta}(k^-) \geq \frac{\phi}{|\sigma_{\lft}|} \geq \frac{\phi}{|\sigma|}$.

If the call is left-free, then Proposition~\ref{hb1prop} gives $\mu_{\beta}(k^-) \geq  \frac{\phi}{a_{\lft}  - k^- + 1} \geq \frac{\phi}{|\sigma|}$.
\end{proof}

\begin{proposition}
The segment $\theta = (\beta, \sigma)$ produced at the given iteration satisfies property {\tt (I4)}.
\end{proposition}
\begin{proof}
We only verify that $\theta$ satisfies Eq.~(\ref{eq:extremal:a}); the case of Eq.~(\ref{eq:extremal:b}) is completely analogous.   

First, suppose the call is left-forced. If $h_{\lft} = 0$ there is nothing to show. If {\tt FindSegment} is called at Line 7 with $\beta = \beta_{\lft}$, then $k^- = \sigma_{\lft}^-, k^+ \geq a_{\rgt} \geq \sigma_{\lft}^+$; so $|\sigma| \geq |\sigma_{\lft}|$.  For $i < k^-$, we use Eq.~(\ref{eq:extremal:a}) for segment $\theta_{\lft}$ to calculate:
$$
\mu_{\beta}(i) \leq \lambda \cdot \frac{ |\sigma_{\lft}|}{|\sigma_{\lft}|+(\sigma_{\lft}^--i)}\cdot\mu_\beta(\sigma_{\lft}^-) \leq \lambda \cdot \frac{ |\sigma|}{|\sigma|+(\sigma^--i)}\cdot\mu_\beta(k^-)  
$$

Otherwise, suppose the call is left-free. Since $|\sigma| \geq a_{\lft} - k^- + 1$, it suffices to show that
\begin{equation}
\label{ttx1}
(a_{\lft} - i + 1) \mu_\beta(i) \le \lambda (a_{\lft} - k^- + 1) \mu_\beta(k^-) \qquad\quad \text{for } i < k^-
\end{equation}

If $k^- = h_{\lft}$, this is precisely Lemma~\ref{lemma:extremal}(a). So suppose that $k^- > h_{\lft}$ and so $\Phi(k^-) = \hat \mu_{\beta}(k^-)$.  Define $\kappa_j = 1$ for $j  \notin \{h_{\lft}, h_{\rgt} \}$ and $\kappa_j = \lambda$ for $j \in \{ h_{\lft}, h_{\rgt} \}$, i.e. $\Phi(j) = \kappa_j \cdot \hat \mu_{\beta}(j)$. Now let $i \in \{h_{\lft}, \dots, k^- \}$.  By definition of $k^-$, we have $(a_{\lft}  - i +1) \Phi(i) \le (a_{\lft}  - k^- +1) \Phi(k^-)$, i.e.
\begin{equation}
\label{po1}
\hat \mu_{\beta}(i) \le \frac{(a_{\lft}  - k^- +1) \Phi(k^-)}{\kappa_i (a_{\lft}  - i +1)} = \frac{(a_{\lft}  - k^- +1) \hat \mu_{\beta}(k^-)}{\kappa_i (a_{\lft}  - i +1)}
\end{equation}

By Proposition~\ref{hb1prop},  $(a_{\lft}  - k^- +1) \hat \mu_{\beta}(k^-) \ge \phi \sqrt{\lambda}$. In particular, the RHS of Eq.~(\ref{po1}) is at least $\bar p = \frac{\phi}{h_{\rgt} - h_{\lft} + 1}$. Since Line 1 well-estimates $i$ and since $\hat \mu_{\beta}(k^-) \leq \sqrt{\lambda} \mu_{\beta}(k^-)$,  this implies
\begin{equation}
\label{po2}
\mu_{\beta}(i) \le \frac{ \lambda (a_{\lft}  - k^- +1) \mu_{\beta}(k^-)}{\kappa_i  (a_{\lft}  - i +1)}.
\end{equation}

For $i \in\{ h_{\lft} + 1, \dots, k^-\}$, we have $\kappa_i = 1$, and so Eq. (\ref{po2}) is exactly Eq.~(\ref{ttx1}). For $i = h_{\lft}$, we have $\kappa_i = \lambda$ and so Eq.~(\ref{po2}) shows 
\begin{equation}
\label{ttx2}
( a_{\lft} - h_{\lft} +1) \mu_{\beta}(h_{\lft}) \le (a_{\lft}  -  k^- +1) \mu_{\beta}(k^-)
\end{equation}
which again establishes Eq. (\ref{ttx1}) (with additional slack).  Finally, when  $i < h_{\lft}$, we show Eq.~(\ref{ttx1}) by combining Lemma~\ref{lemma:extremal}(a) with Eq.~(\ref{ttx2}):
\[
( a_{\lft} - i + 1 ) \mu_{\beta}(i) \le \lambda (a_{\lft} - h_{\lft} + 1) \mu_{\beta} (h_{\lft})  \leq \lambda (a_{\lft}  -  k^- +1) \mu_{\beta}(k^-) \qedhere
\]
\end{proof}

Thus, the preschedule $\mathcal J$ at the end maintains property ({\tt I3}), and hence it is $\phi$-proper.  This concludes the proof of Lemma~\ref{prop:sch1}.

\subsection{Converting the preschedule into a covering schedule}
\label{sec:convert}
 There are two steps to convert the preschedule into a covering schedule. First, we throw away redundant intervals; second, we ``uncross'' the adjacent intervals. While we are doing this, we also check if the resulting schedule is proper;  if not, we will discard it and generate a new preschedule from scratch. 
 
\begin{algorithm}[H]
\DontPrintSemicolon
\While{$\mathcal J$ contains any adjacent segments $\theta_i = (\beta_i, \sigma_i), \theta_{i+1} = (\beta_{i+1}, \sigma_{i+1})$ with $\beta_i = \beta_{i+1}$} {
merge the segments, i.e. replace $\theta_i, \theta_{i+1}$ with a single segment $\theta' = (\beta_i, [\sigma_i^-, \sigma_{i+1}^+])$. 
}

\DontPrintSemicolon
\While{$\mathcal J$ contains any index $i \in \{1, \dots, t-1 \}$ with $\sigma_{i+1}^- \leq \sigma_{i-1}^+$}{
remove segment $\theta_i$ from $\mathcal J$
}

let $\mathcal J' = ((\beta_0, \sigma_0), \dots, (\beta_t, \sigma_t))$ be the preschedule after these modification steps.

{\bf for} $i = 0, \dots, t$ {\bf do}  let $\hat \mu_{\beta_i} \leftarrow {\tt Sample}(\beta_i; \tfrac{1}{2} \log \lambda,    \frac{w_i}{\lambda}, \frac{\gamma}{4 (t+1)} )$ where $w_i = \phi/|\sigma_i|$ \\
\For{$i = 1, \dots, t$} {
\If{ $\exists\; k\in\{\sigma_{i-1}^+,\sigma_{i}^-\}$ with \  
$\hat\mu_{\beta_{i-1}}(k)\ge w_{i-1} / \sqrt{\lambda}$
and
$\hat\mu_{\beta_{i}}(k)\ge w_i / \sqrt{\lambda}$
}
{
set $k_{i} = k$ for arbitrary such $k$
}
\textbf{else return} ERROR.
}
\Return{covering schedule $\mathcal I = ( \beta_0, w_0 / \lambda, k_1, \beta_1,  w_1 / \lambda, k_2, \dots, k_t, \beta_t, w_t / \lambda)$}
\caption{${\tt TryToUncross}(\cal J, \gamma)$ for preschedule $\mathcal J$.
\vspace{-0.15in} \label{alg:UncrossSchedule}}
\end{algorithm}

\begin{proposition}
\label{min-prop1}
At Line 5, the preschedule $\mathcal J'$ satisfies $\sum_{i=0}^t \frac{1}{w_i} \leq \frac{2(n+1)}{\phi}$.
Furthermore, if the original preschedule $\mathcal J$ is $\phi$-proper, then so is $\mathcal J'$.
\end{proposition}
\begin{proof}
When segments are merged at Line 2, the new segment has a larger span, so it preserves properness. The discarding step at Line 4 clearly does not make a schedule improper. 

We claim that for any $k \in \mathcal H$,  there are at most two segments $\theta_i = (\beta_i, \sigma_i) \in \mathcal J'$ with $k \in \sigma_i$. For, suppose $k \in \sigma_{i_1} \cap \sigma_{i_2} \cap \sigma_{i_3}$ with $i_1 < i_2 < i_3$. Then by ({\tt I2}), we have $\sigma_{i_2 + 1}^- \leq \sigma_{i_3}^- \leq k$ and $\sigma_{i_2 - 1}^+ \geq \sigma_{i_1}^+ \geq k$. So $\sigma_{i_2 +1}^- \leq \sigma_{i_2-1}^+$ and we could have discarded segment $i_2$ in the loop at Line 3.

Consequently, we calculate $\sum_{i=0}^t \frac{1}{w_i} = \frac{1}{\phi} \sum_{i=0}^t |\sigma_i| = \frac{1}{\phi} \sum_{k \in \mathcal H}  | \{i : k \in \sigma_i \}  | \leq \frac{2(n+1)}{\phi}.$
\end{proof}

\begin{theorem}\label{th:subroutines2}
The algorithm ${\tt TryToUncross}(\mathcal J, \gamma)$ satisfies the following properties: \\
(a) The output is either ERROR or a covering schedule $\mathcal I$ with ${\tt InvWt}(\mathcal I) \leq  \frac{2 \lambda (n+1)}{\phi}$. \\
(b) Irrespective of $\mathcal J$, it outputs an improper covering schedule with probability at most $\gamma$. \\
(c) If $\mathcal J$ is $\phi$-proper,  it outputs ERROR with probability at most $\gamma$.  \\
(d) The cost is $O( n \rho \log\frac n{\gamma})$.
\end{theorem}
\begin{proof}
From Proposition~\ref{min-prop1}, we have ${\tt InvWt}(\mathcal I) = \sum_{i=0}^t \frac{1}{w_i / \lambda} \leq \frac{2 \lambda (n+1)}{\phi}$.  Similarly, the algorithm cost from Line 6 is $\sum_{i=0}^t O( w_i \log \tfrac{t}{\gamma} ) \leq  O(  \frac{2(n+1)}{\phi} \cdot  \log \tfrac{n}{\gamma})  \leq O( n \rho \log \tfrac{n}{\gamma})$, thus showing part (d). 

To show $\mathcal I$ is a covering schedule, we need to show $k_1 < \dots < k_t$ and $\beta_0 = \betamin < \dots < \beta_t = \betamax$. The condition $\beta_0 = \betamin, \beta_t = \betamax$ follows from ({\tt I1}) while the bound $\beta_i < \beta_{i+1}$ follows from the loop at Line 1. To show $k_i < k_{i+1}$ for $i = 1, \dots, t-1$, observe that from ({\tt I0}) we have $\sigma_{i-1}^+ \geq \sigma_i^-$ and $\sigma_i^+ \geq \sigma_{i+1}^-$. Due to the loop at Line 3, we must have $k_{i+1} \geq \sigma_{i+1}^- > \sigma_{i-1}^+ \geq k_i$ (else segment $\theta_i$ would have been discarded).

With probability at least $1 - \gamma$, Line 6 well-estimates every value $\sigma_i^+, \sigma_i^-, \sigma_{i+1}^-, \sigma_{i+1}^+$. We claim that, in this case,  the algorithm outputs either a proper covering schedule or ERROR, and that the latter case only holds if $\mathcal J'$, and hence $\mathcal J$, is improper.

First, if the algorithm reaches Line 11, then  $\hat \mu_{\beta_i}(k_i) \geq w_i / \sqrt{\lambda}$ and $\hat \mu_{\beta_{i-1}}(k_i) \geq  w_{i-1} / \sqrt{\lambda}$ for each $i$. Since Line 6 well-estimates $k_i$, this implies $ \mu_{\beta_i}(k_i) \geq w_i / \lambda$ and $\mu_{\beta_{i-1}}(k_{i}) \geq w_{i-1} / \lambda$. So $\mathcal I$ is proper.

Next, suppose $\mathcal J'$ is $\phi$-proper but the algorithm outputs ERROR at an iteration $i$. By definition of $\phi$-properness, we have $\mu_{\beta_{i-1}}(\sigma_{i-1}^+) \geq w_{i-1}$ and $\mu_{\beta_i}(\sigma_i^-) \geq w_i$. Since Line 6 well-estimates $\sigma_{i-1}^+$ and $\sigma_i^+$, this implies $\hat \mu_{\beta_{i-1}}(\sigma_{i-1}^+) \geq w_{i-1} / \sqrt{\lambda}$ and $\hat \mu_{\beta_i}(\sigma_i^-) \geq w_i / \sqrt{\lambda}$. Neither value $k \in \{ \sigma_{i-1}^+, \sigma_i^- \}$ satisfied the check at Line 8, so $\hat \mu_{\beta_i}(\sigma_{i-1}^+) <  w_i / \sqrt{\lambda}$ and $\hat \mu_{\beta_{i-1}}(\sigma_i^-) < w_{i-1} / \sqrt{\lambda}$. In turn, since Line 6 well-estimates these values, we have $\mu_{\beta_i}(\sigma_{i-1}^+) < w_i$ and $\mu_{\beta_{i-1}}(\sigma_i^-) < w_{i-1}$.  But now $\mu_{\beta_i}(\sigma_{i-1}^+) \mu_{\beta_{i-1}}(\sigma_i^-) < w_{i-1} w_i \leq \mu_{\beta_i}(\sigma_i^-) \mu_{\beta_{i-1}}(\sigma_{i-1}^+)$. This contradicts the Uncrossing Inequality since $\sigma_i^- \leq \sigma_{i-1}^+$ by ({\tt I0}).
\end{proof}

We can finish by combining all the preschedule processing algorithms, as follows:

\begin{algorithm}[H]
\DontPrintSemicolon
\While{{\tt true}}
{
	call Algorithm~\ref{alg:schedule} with appropriate constants $ \lambda, \tau$ to compute preschedule $\calJ$ \\	
	call $\calI\leftarrow{\tt TryToUncross}(\calJ, \gamma/4)$ \\
	\textbf{if} $\mathcal I \neq $ ERROR \textbf{then return} $\mathcal I$  \\
}
\caption{Algorithm ${\tt FindCoveringSchedule}(\gamma)$ \label{alg:schedule:final}}
\end{algorithm}

By Lemma~\ref{prop:sch1} and Theorem~\ref{th:subroutines2},  each iteration of Algorithm~\ref{alg:schedule:final} terminates with probability at least $\tfrac{1}{2} (1-\gamma/4) \geq 3/8$, so there are $O(1)$ expected iterations. Each call to {\tt TryToUncross} has cost $O(n \rho \log \frac{n}{\gamma})$. By Proposition~\ref{prop:schedule:post-analysis}, each call to Algorithm~\ref{alg:schedule} has cost $O(n \log q + n \rho \log^2 n)$. 

By Theorem~\ref{th:subroutines2}(a), we have ${\tt InvWt} (\mathcal I) \leq 2 \rho  (n+1) \cdot \lambda^4 / \tau$. The term $\lambda^4 / \tau$ gets arbitrarily close to $2$ for constants $ \lambda, \tau$ sufficiently close to $1, \frac 12$ respectively. 

By Proposition~\ref{th:subroutines2}, each iteration of Algorithm~\ref{alg:schedule:final} returns a non-proper covering schedule with probability at most $\gamma/4$ (irrespective of the choice of $\mathcal J$). So the total probability of returning a non-proper covering schedule over all iterations is at most $\sum_{i=0}^{\infty} (3/8)^i \gamma/4 = 2 \gamma/5$.

This shows Theorem~\ref{th:schedule-existence}.

\section{Combinatorial applications}
\label{app-sec}
Consider a configuration with $c_i$ objects of weights $i = 0, \dots, n$,  where we can sample from a Gibbs distribution at certain values $\beta \in [\betamin, \betamax]$. To determine these counts, we will solve $\Pcoef{\delta, \eps}$ for $\delta \leq \Delta_{\min} \eqdef \min_{x \in \mathcal H} \Delta(x)$; the resulting estimated counts $\hat c_i \propto \hat \pi(i)$ will then be accurate up to scaling to within $e^{\pm O(\eps)}$ relative error. (To determine the scaling, we usually assume that at least one count is known outright, e.g. $c_0 = 1$.)

We will examine three problems from graph theory: connected subgraphs, matchings, and independent sets.  There is essentially only one piece of information we need about the underlying Gibbs distribution, namely, the value $\Delta_{\min}$; most other problem-specific details can be ignored. Furthermore, log-concave distributions have a simple ``automatic'' bound on $\Delta_{\min}$; combined with our algorithm for log-concave distributions, this gives a particularly clean algorithmic result. 

\begin{theorem}\label{th:main:two2}
Suppose the counts are non-zero and log-concave.  If $\betamin \leq \log \frac{c_0}{c_1}$ and $\betamax \geq \log\frac{c_{n-1}}{c_n}$, then $\Delta_{\min} \geq \frac{1}{n+1}$ and $\log Q(\betamax) \leq q = 3 n \Gamma$ for $\Gamma = \max\{ \betamax, \log \frac{c_1}{c_0}, 1 \}$. 

In particular, for $\delta = \frac{1}{n+1}$, we can solve $\Pcoef{\delta, \eps}$ with cost $O \bigl( \min \bigl \{ \frac{ n \Gamma \log n \log \frac{1}{\gamma}}{\eps^2}, \frac{ n^2 \log\frac{1}{\gamma}}{\eps^2} + n \log \Gamma \bigr \} \bigr)$.

\end{theorem}
\begin{proof}
Define $b_i = \log(c_{i-1}/c_i)$ for $i = 1, \dots, n$; the sequence $b_1, \dots, b_n$ is non-decreasing since $c_i$ is log-concave. We claim that for each $i, k \in \mathcal H$, there holds
\begin{equation}
\label{ddd1}
c_i e^{i b_i} \geq c_k e^{k b_i}
\end{equation} 

To show this for $k > i$, we note that $\frac{c_i e^{i b_i}}{c_k e^{k b_i}} = e^{(i-k) b_i} \prod_{j=i}^{k-1} \frac{c_j}{c_{j+1}} = e^{\sum_{j=i}^{k-1}  b_{j+1} - b_i} \geq 1$. A similar calculation holds for $k <i$. Since $\mu_{\beta}(k) \propto c_k e^{\beta k}$, Eq.~(\ref{ddd1}) implies 
$$
\mu_{b_i}(i) = \frac{\mu_{b_i}(i)}{\sum_{k=0}^n \mu_{b_i}(k)} \geq \frac{\mu_{b_i}(i)}{\sum_{k=0}^n \mu_{b_i}(i)} = \frac{1}{n+1}.
$$

Similarly, we have $\mu_{b_1}(0) \geq \frac{1}{n+1}$. So $b_i \in [b_1, b_n] \subseteq [\betamin, \betamax]$ and $\Delta_{\min} \geq \tfrac{1}{n+1}$ as claimed. We next turn to the bound on $Q(\betamax)$. We have the lower bound $Z(\beta_{\min}) \geq c_0$.  To upper-bound $Z(\betamax)$, observe from Eq.~(\ref{ddd1}) and from our hypothesis $\betamax \geq b_n$ that for every $k \leq n$, that
$$
\frac{c_n e^{n \betamax}}{c_k e^{k \betamax}} = \frac{c_n e^{n b_n}}{c_k e^{k b_n}} \cdot e^{(\betamax - b_n) (n - k)} \geq 1.
$$
So $Z(\beta_{\max}) = \sum_i c_i e^{i \beta_{\max}} \leq (n+1) c_n e^{n \beta_{\max}}$ and hence $Q(\betamax) \leq \frac{e^{n \betamax} (n+1) c_n}{c_0}$. By telescoping products, $\frac{c_{n}}{c_0} = \prod_{i=1}^{n} \frac{c_i}{c_{i-1}} \leq (\frac{c_1}{c_0})^n$, giving $Q(\betamax) \leq e^{n \betamax} \cdot (n+1)  \cdot (\frac{c_1}{c_0})^n \leq e^{n \Gamma} \cdot (n+1) \cdot e^{n \Gamma} \leq e^q$.
\end{proof}

\subsection{Counting connected subgraphs}
Let $G = (V,E)$ be a connected graph.  The problem of \emph{network reliability} is to sample a connected subgraph $G' = (V,E')$ with probability proportional to $\prod_{f \in E'} \frac{1 - p(f)}{p(f)}$, for any weighting function $p: E \rightarrow [0,1]$. This can be interpreted as each edge $f$ ``failing'' independently with probability $p(f)$, and conditioning on the resulting subgraph remaining connected. Equivalently, if we set $p(f) = \frac{e^{\beta}}{1 + e^{\beta}}$ for all edges $f$, it can be interpreted as a Gibbs distribution where $c_i$ is the number of connected subgraphs of $G$ with $|E| - i$ edges and with $n = |E| - |V| + 1$.

Note that $c_i$ is the number of independent sets in the co-graphic matroid. By the result of \cite{Adiprasito:18}, this sequence $c_i$ is log-concave.  Anari et al. \cite{anari2018log} provides an FPRAS for counting independent sets in arbitrary matroids, which would include connected subgraphs, but do not provide concrete complexity estimates for their algorithm.  There have also been a series of exact and approximate sampling algorithms for network reliability via Gibbs distributions \cite{GuoJerrum:ICALP18, GuoHe:18, chen}. The most recent result is due to \cite{chen}, which applies to general matroids satisfying certain computational properties. For graph reliability, it provides the following algorithmic result:
\begin{theorem}[\cite{chen}, Corollary 3]
\label{guohe:alg}
There is an algorithm to approximately sample from the Gibbs distribution at $\beta$, up to total variation distance $\rho$, with expected runtime $\tilde O( |E| ( e^{\beta} + 1) \log \tfrac{1}{\rho} )$.
\end{theorem}

Based on Theorem~\ref{guohe:alg}, we can immediately get our counting algorithm.

\begin{proof}[Proof of Theorem~\ref{th:count:subgraph}] Here $c_0 = 1$ and $c_1/c_0 \leq |E|$ and $c_{n-1}/c_n \leq |E|$.  So we can apply Theorem~\ref{th:main:two2},  setting $\betamax = \log |E|, \betamin = -\infty$ and $\Gamma = \log |E|$.  The algorithm uses $O ( \frac{ |E| \log^2 |E| }{ \eps^2 })$ samples in expectation to achieve failure probability $\gamma = O(1)$.  Accordingly, we need to run the approximate sampler of Theorem~\ref{guohe:alg} with $\rho = \poly(1/n, \eps)$. The total algorithm runtime is then $\tilde O( |E|^3 / \eps^2 )$.
\end{proof}

\subsection{Counting matchings in high-degree graphs}
Consider a graph $G = (V,E)$ with a perfect matching. For $i = 0, \dots, v = |V|/2$, let $c_i$ denote the number of $i$-edge matchings. As originally shown in \cite{hl72}, the sequence $c_i$ is log-concave.  In \cite{jerrum1989approximating, jerrum1996markov}, Jerrum \& Sinclair described an MCMC algorithm to approximately sample from the Gibbs distribution on matchings. We quote their most optimized result as follows:
\begin{theorem}[\cite{jerrum1996markov}]
\label{js:alg}
There is an algorithm to approximately sample from the Gibbs distribution at  $\beta$  with up to total variation distance $\rho$, with expected runtime $\tilde O( |E| |V|^2 (1 + e^{\beta}) \log \frac{1}{\rho} )$.
\end{theorem}

\begin{proof}[Proof of Theorem~\ref{th:count:matchings}]
Here $c_0 = 1$ and $n = v$ and  $c_1/c_0 \leq |E|$, and by assumption  $c_{n-1}/c_n \leq f$. We apply Theorem~\ref{th:main:two2} setting $\betamin = -\log |E|, \betamax = \log f$, and $\Gamma \leq \max \{ \log |E|, \log f \}$. The algorithm uses $O( n \log(|E| f) \log^2 n  / \eps^2)$ samples in expectation for failure probability $\gamma = O(1)$. Accordingly, we need to take the approximate sampler of Theorem~\ref{js:alg} with $\rho = \poly(1/n, \eps,1/ f)$. The total algorithm runtime is then $\tilde O( \frac{|E| |V|^3 f}{ \eps^2} )$.  As shown in  \cite{jerrum1989approximating}, if $G$ has minimum degree at least $|V|/2$ then $c_{v} > 0$ and $c_{v-1}/c_{v} \leq O( |V|^2)$ and clearly $|E| \leq O(|V|^2)$.
\end{proof}

\subsection{Counting independent sets in bounded-degree graphs}
\label{subsec:countindep}
Let $G = (V,E)$ be a graph of of maximum degree $D$. A key problem in statistical physics is to efficiently sample the independent sets of $G$. We consider the Gibbs distribution on independent sets weighted by their size, i.e. each independent set $I \subseteq V$ is sampled with probability $e^{\beta |I|}$, and $c_k$ is the number of cardinality-$k$ independent sets. Note that this distribution is not necessarily log-concave.  

As shown in \cite{chen2021optimal}, there is a critical hardness threshold defined by $\beta_c = \log \bigl( \frac{(D-1)^{D-1}}{(D - 2)^D} \bigr)$: it is intractable to sample from the Gibbs distribution beyond $\beta > \beta_c$ and there is a polynomial-time sampler for the Gibbs distribution for $\beta < \beta_c$. We quote the following result:

\begin{theorem}[\cite{chen2021optimal}] 
\label{haal1} Let $D \geq 3$ and $\xi > 0$ be any fixed constants. There is an algorithm to approximately sample from the Gibbs distribution at  $\beta \in [-\infty,\beta_c - \xi]$, up to total variation distance $\rho$, with runtime $O( |V| \log |V| \log \tfrac{1}{\rho} )$.\footnote{Concretely, \cite{chen2021optimal} describes an MCMC known as \emph{Glauber dynamics} whose stationary distribution is the Gibbs distribution. They show this Markov chain has mixing time $O( |V| \log |V|)$. For a constant-degree graph, each step of the Glauber dynamics can be executed in $O(1)$ time.}
\end{theorem}
The related problem of estimating counts was considered in \cite{davies_et_al,jain2022approximate}. They identified a related threshold value $\alpha_c = \frac{e^{\beta_c}}{1 + (D+1) e^{\beta_c}}$ such that for $k > \alpha_c |V|$, it is intractable to estimate $c_k$ or sample approximately uniformly from size-$k$ independent sets, while on the other hand, for any constants $D \geq 3, \xi > 0$ and any $k < (\alpha_c - \xi) |V|$, there is a polynomial-time algorithm for both tasks. A key technical result shown in \cite{jain2022approximate} was that, in the latter regime, $\mu_{\beta}$ could be approximated by a normal distribution, i.e. it obeyed a type of Central Limit Theorem. 

For our purposes, using Theorem 3.1 of \cite{jain2022approximate}, we can get a cruder estimate:
\begin{lemma}[\cite{jain2022approximate}]
\label{haal}
Let $D \geq 3$ and $\xi > 0$ be any fixed constants. There is a constant $\xi' > 0$ such that, for any $k \leq (\alpha_c - \xi) |V|$, there is a value $\beta < \beta_c - \xi' $ with $\mu_{\beta}(k) \geq \Omega(1/\sqrt{|V|})$.
\end{lemma}

This bound is all we need for our counting algorithm, which we summarize as follows:
\begin{proof}[Proof of Theorem~\ref{th:count:ind}]
Here $c_0 = 1$ and $ n = |V|$. We set $\betamin = -\infty$ and $\betamax = \beta_c - \xi'$. So then $\frac{Z(\betamax)}{Z(\betamin)} \leq \frac{2^n e^{\betamax n}}{c_0}$;  in particular, since $\betamax = O(1)$ (for fixed $D$), we can take $q = \Theta(n)$.  It suffices to solve $\Pcoef{\delta,  \eps}$ for $\delta = \Omega(1/\sqrt{n})$.  We use the continuous-setting algorithm of Theorem~\ref{pcoefalg1thm}; for $\gamma = O(1)$, it has expected sample complexity $O \bigl( \frac{ q \log n + (\sqrt{q \log n} / \delta) \log \tfrac{q}{\delta}}{\eps^2} \bigr) = O\bigl( \frac{  n \log^{3/2} n }{\eps^2} \bigr)$. The sampler of Theorem~\ref{haal1} then should set $\rho = \poly(1/n,\eps)$, and the overall runtime is $O \bigl( \frac{  n^2 \log^{5/2} n  \log(n/\eps)}{\eps^2} \bigr) = \tilde O(n^2/\eps^2).$ \qedhere
\end{proof}

\section{Lower bounds on sample complexity}
\label{lb-sec}
In this section, we show lower bounds for the sample complexity of $\Pratio{\eps}$ and $\Pcoef{\delta, \eps}$ in the black-box model.  At this point, we note that the specification of problem $\Pcoef{\delta, \eps}$, while useful from an algorithmic point of view, has some  restrictive choices for its normalization and error terms. To get lower bounds which do not depend on these details, we consider a relaxed count-estimation problem called $\Ptcoef{\delta, \eps}$,  namely, to compute a vector $\hat c \in \mathbb R_{\geq 0}^{\mathcal F}$ satisfying the following property:
 $$
\Bigl | \log \frac{\hat c_x}{\hat c_y} - \log \frac{c_x}{c_y} \Bigr | \leq \eps  \qquad \qquad \text{ for all pairs $x,y$ with $\Delta (x), \Delta (y) \geq \delta$}
 $$

One can easily check that a solution $\hat \pi$ to $\Pcoef{\delta, \eps/5}$ is also a solution to $\Ptcoef{\delta, \eps}$. Thus, up to constant factors in parameters, the problem $\Ptcoef{}$ is easier to solve compared to $\Pcoef{}$.

Our strategy for the lower bounds is to construct a target instance $c = c^{(0)}$ surrounded by an envelope of alternate instances which can be distinguished by solving $\Pratio{\eps}$ or $\Ptcoef{\delta, \eps}$. Specifically, for vectors $\lambda^{(1)}, \dots, \lambda^{(d)} \in \mathbb R^{\mathcal F}$, we will define $2 d$ alternate instances $c^{(+1)}, c^{(-1)}, \dots, c^{(+d)}, c^{(-d)}$ by:
$$
c^{(+r)}_x = e^{\lambda^{(r)}_x} c_x, \qquad  c^{(-r)}_x = e^{-\lambda^{(r)}_x} c_x  \qquad \qquad \text{for $r = 1, \dots, d$ and $x \in \mathcal F$}
$$

Thus, these $2 d$ instances are arranged in $d$ pairs which ``straddle'' the base instance $c^{(0)}$.

For each instance $c^{(r)}$, we define $\mu_{\beta}^{(r)}$ to be the Gibbs distribution, and $Z^{(r)}$ to be its partition function, and $z^{(r)} = \log Z^{(r)}$, etc. When the superscript is omitted, we always refer to the base instance $c = c^{(0)}$, e.g. we write $z(\beta) = z^{(0)}(\beta), \Delta(x) = \max_{\beta \in [\betamin, \betamax]} \mu_{\beta}^{(0)}(x)$.

We also define parameters
$$
\kappa = \max_{\substack{x \in \mathcal F, r \in \{1, \dots, d \} }} | \lambda^{(r)}_x | , \qquad  \Psi =     \max_{\beta \in [\betamin, \betamax]}  \sum_{r=1}^d \log \frac{Z^{(+r)}(\beta) Z^{(-r)}(\beta) }{Z(\beta)^2}.
$$

In this setting,  we have a key lower bound on the sample complexity of any procedure to distinguish the distributions. 
\begin{lemma}[Indistinguishability lemma]
\label{gtt5}
Let $\A$ be an algorithm which generates queries $\beta_1, \dots, \beta_T \in [\betamin, \betamax]$ and receives values $x_1, \dots, x_T$, where each $x_i$ is drawn from $\mu_{\beta_i}$. At some point the procedure stops and  outputs YES or NO. The queries $\beta_i$ and the stopping time $T$ may be adaptive and may be randomized.

If $\A$ outputs YES on instance $c^{(0)}$ with probability at least $1 - \gamma$ and outputs NO on each instance $c^{(-d)}, \dots c^{(-1)}, c^{(1)}, \dots, c^{(d)}$ with probability at least $1 - \gamma$, then $\A$ has cost $\Omega( \frac{ d \log(1/\gamma)}{\Psi})$ when run on instance $c^{(0)}$.
\end{lemma}
\begin{proof}
Any finite run of $\mathfrak A$ can be described by a random vector $X = ((\beta_1,x_1),\ldots,(\beta_T,x_T),\sigma)$, where $\sigma\in\{\text{YES},\text{NO} \}$. The probability measures  $\P^{(r)}(\cdot)$ over such runs on $c^{(r)}$ can be decomposed as:
\begin{equation}
\d\P^{(r)}(X)= \psi^{(r)}(X) \: \d \mu^\A X
\label{eq:GALSKGHAKSFJASG}
\end{equation}
where the measure $\mu^\A$  depends only on the algorithm $\A$, and the function $\psi^{(r)}$ is defined via
$$
\psi^{(r)} ((\beta_1,x_1),\ldots,(\beta_T,x_T),\sigma)=\prod_{i=1}^T \mu_{\beta_i}^{(r)}(x_i)
$$

Define $R = \{-d, \dots, -1, +1, \dots, +d \}$ (the alternate problem instances). For any finite run $X = ((\beta_1,x_1),\ldots,(\beta_T,x_T),\sigma)$ we have
\begin{align}
\prod_{r \in R} \psi^{(r)}(X) &= \prod_{r \in R} \prod_{t=1}^T \mu^{(r)}_{\beta_t}(x_t) = \prod_{t=1}^T (\mu_{\beta_t}(x_t))^{2 d} \prod_{r \in R} \frac{c^{(r)}_{x_t} e^{\beta_t x_t} Z(\beta_t)}{c_{x_t} e^{\beta_t x_t} Z^{(r)}(\beta_t)} \notag \\
&= (\psi^{(0)}(X))^{2 d} \prod_{t=1}^T \prod_{r=1}^d \frac{ Z(\beta_t)^2 }{Z^{(+r)}(\beta_t) Z^{(-r)}(\beta_t)}  \geq  (\psi^{(0)}(X))^{2 d} e^{-T \Psi}
\label{gfss10}
\end{align}

Let $\mathcal X$ denote the set of runs of length at most $\tau = \frac{2 d \log (1/(4\gamma))}{\Psi}$. For any  run $X \in \mathcal X$, Eq.~(\ref{gfss10}) gives:
\begin{equation}
\label{gfss1}
\frac{1}{2d}\sum_{r \in R} \psi^{(r)}(X) 
\ge \Bigl (\prod_{r \in R} \psi^{(r)}(X) \Bigr)^{1/(2 d)}  \geq \psi^{(0)}(X) e^{-\tau \Psi/2 d}= 4 \gamma \psi^{(0)}(X)
\end{equation}
where the first bound comes from the inequality between arithmetic and geometric means, and the second bound comes from (\ref{gfss10}).

Partition $\mathcal X$ into sets of runs  $\mathcal X^{\text{Y}}, \mathcal X^{\text{N}}$  which output YES and NO respectively.  By hypothesis of the lemma, we have $\P^{(r)}(\mathcal X^\text{Y}) < \gamma$ for all $r \in R$. Using Eq.~(\ref{gfss1}), this implies
$$
\P^{(0)}(\mathcal X^{\text{Y}}) = \negthinspace \int_{\calX^{\text{Y}}} \negthinspace \psi^{(0)}(X) \: \d\mu^\A X \leq \frac{1}{8 d \gamma}\sum_{r \in R} \int_{\calX^{\text{Y}}} \negthinspace \psi^{(r)}(X) \: \d\mu^\A X = \frac{1}{8 d \gamma} \sum_{r \in R} \P^{(r)} (\mathcal X^{\text{Y}}) \leq \frac{1}{8 d \gamma} \sum_{r \in R} \gamma = \frac{1}{4}.
$$

By hypothesis, $\P^{(0)}(\mathcal X^{\text{N}}) < \gamma$. So  $\A$ runs for more than $\tau$ timesteps on $c^{(0)}$ with probability at least $1-\P^{(0)}(\mathcal X^{\text{Y}}) -\P^{(0)} (\mathcal X^{\text{N}}) \geq 1 - \gamma - \frac{1}{4}$; by our assumption that $\gamma \in (0,\tfrac{1}{2})$, this is at least $\tfrac{1}{4}$. The expected cost is thus at least $\frac 14\cdot \tau=\Omega( \frac{ d \log(1/\gamma)}{\Psi})$.
\end{proof}

\begin{corollary}
\label{gtt5h}
If $z(\betamin, \betamax) \leq q - 2 \kappa$, then  $z^{(r)} (\betamin, \betamax) \leq  q$ for all $r = \pm 1, \dots, \pm d$. Moreover, in this case:

\smallskip

\noindent(a) If $|z(\betamin, \betamax)-z^{(r)}(\betamin,\betamax)| > 2 \eps$ for all $r = \pm 1, \dots, \pm d$, then algorithm for $\Pratio{\eps}$ has cost $\Omega( \frac{ d \log(1/\gamma)}{\Psi})$ on $c^{(0)}$.

\noindent(b) If each $r \in \{1, \dots,d \}$ has values $x, y \in \mathcal F$ with $\Delta(x), \Delta (y) \geq \delta e^{2 \kappa}$ and $| \lambda^{(r)}_x - \lambda^{(r)}_y | > 2 \eps$,  then any algorithm for $\Ptcoef{\delta, \eps}$ has cost $\Omega( \frac{ d \log(1/\gamma)}{\Psi})$ on  $c^{(0)}$.   (We call $x,y$ the \emph{witnesses} for instances $c^{(\pm r)}$).
\end{corollary}
\begin{proof}
As $| \log c^{(r)}_x - \log c_x | \leq \kappa$ for all $x$, we have $\frac{Z^{(r)}(\beta)}{ Z(\beta)} \in [e^{-\kappa}, e^{\kappa}]$ and $\mu^{(r)}_{\beta}(x) = \frac{c_x^{(r)} e^{\beta x}}{Z^{(r)}(\beta)} \geq  \mu_{\beta}(x) e^{-2\kappa}$. So $z^{(r)}(\betamin, \betamax) \leq q$ for all $r = \pm 1, \dots, \pm d$.  In particular, the upper-bound value $q$ is valid for   all the instances $c^{(r)}$.

In case (a), an algorithm for $\Ptcoef{\delta,\eps}$ can  distinguish $c^{(0)}$ from $c^{(\pm 1)}, \dots, c^{(\pm d)}$: namely,  given a solution $\hat Q(\betamax)$ for $\Pratio{\eps}$, we output YES if and only if $|  \log \hat Q(\betamax) - \log Q^{(0)}(\betamax) | \leq \eps$.

In case (b), an algorithm for $\Pratio{\eps}$ can distinguish $c^{(0)}$ from $c^{(\pm 1)}, \dots, c^{(\pm d)}$:  namely, given a solution $\hat c$ to $\Ptcoef{\delta, \eps}$, we output YES if and only if  for all $r = 1, \dots, d$ the witnesses $x,y$ satisfy $| \log  \frac{\hat c_x}{\hat c_y}-  \log \frac{c_x}{c_y}| \leq \eps$. Note that $\Delta^{(\pm r)} (x), \Delta^{(\pm r)} (y) \geq \delta$.
\end{proof}

By applying Corollary~\ref{gtt5h} to carefully constructed  instances, we will show the following results which, in turn, immediately imply Theorems~\ref{main-lb-thmx} and~\ref{main-lb-thmx2}: 
\begin{theorem}
\label{main-lb-thm}
Let $n >  n_0, q > q_0, \eps < \eps_{0}, \delta < \delta_{0}$ for certain absolute constants $n_0, q_0, \eps_{0}, \delta_{0}$. There are problem instances $\mu$ which satisfy the given bounds $n$ and $q$  such that:

\noindent (a)  $ \Ptcoef{\delta, \eps}$  requires cost 
$\Omega(  \frac{ \min \{q + \sqrt{q}/\delta, n^2 + n/\delta \} }{\eps^2}  \log \frac{1}{\gamma})$, and $\mu$ is integer-valued.

\noindent (b)  $ \Ptcoef{\delta, \eps}$  requires cost $\Omega( \frac{ 1/\delta +  \min \{q, n^2 \} }{\eps^2} \log \frac{1}{\gamma})$, and $\mu$ is log-concave.

\noindent (c)  $\Pratio{\eps}$ requires cost  $\Omega( \frac{ \min \{q, n^2 \}}{\eps^2}  \log \frac{1}{\gamma})$, and $\mu$ is  log-concave.

\noindent (d)  $\Ptcoef{\delta, \eps}$  requires cost $\Omega( \frac{ (q + \sqrt{q}/\delta)}{\eps^2}  \log \frac{1}{\gamma})$.

\noindent (e) $\Pratio{\eps}$ requires cost  $\Omega( \frac{ q }{\eps^2}  \log \frac{1}{\gamma})$. 
\end{theorem}

\subsection{Bounds in the log-concave setting}
\label{sec92}
To begin,  observe that the lower bound of $\Omega(\frac{1}{\delta \eps^2} \log \tfrac{1}{\gamma})$ for $\Ptcoef{\delta,\eps}$ is rather trivial, and does not really depend on properties of Gibbs distributions. We summarize it in the following result:
\begin{proposition} 
\label{pc1}
For $\delta < 1/10$ and $\eps < 1/20$, there is a log-concave instance with arbitrary $n \geq 1$ and $\betamin = \betamax = 0$, for which solving $\Ptcoef{\delta, \eps}$ requires cost $\Omega( \frac{ \log(1/\gamma)} {\delta \eps^2})$.
\end{proposition}
\begin{proof}
Take $d = 1$ with $c_0 = 2 \delta, c_1 = 1, c_2 = \dots = c_n = 0$, and $\lambda^{(1)}_0 = 3 \eps, \lambda^{(1)}_1 = \dots = \lambda^{(1)}_n = 0$.  We can apply Corollary~\ref{gtt5h}(b) with witnesses $x = 0, y = 1$.  It is straightforward to calculate $\Delta(0), \Delta(1) \geq 5/3 \cdot \delta \geq \delta e^{2 \kappa}$ and $\Psi =  \log \bigl( \frac{(2 \delta e^{-3 \eps} + 1)(2 \delta e^{3 \eps} + 1)}{ (2 \delta + 1)^2} \bigr) \leq O( \delta \eps^2 )$.
\end{proof}

So it suffices to consider $\delta = \Omega(1)$, and we turn to the bounds in terms of $n$ and $q$. For this construction, take $d = 1, \betamin = 0, \betamax = n$,  set $c_0, \dots, c_n$  to be the coefficients of the polynomial  $g(x) = \prod_{k=0}^{n-1} (e^{k}+ x)$, and set $\lambda^{(1)}_k = 10 k \eps / n$ for $k = 0, \dots, n$.

Since $g(x)$ is a real-rooted polynomial, and counts $c^{(\pm 1)}$ are derived by an exponential shift, the counts $c^{(0)}, c^{(+1)}, c^{(-1)}$ are log-concave~\cite{Braenden:2015}.  Due to the polynomial representation of the counts as $\prod_{k=0}^{n-1} (e^{k}+ x) = c_0 + c_1 x + \dots + c_{n} x^{n}$, we have an elegant formula for the partition ratio function: 
\begin{align*}
Z(\beta) &= \sum_x c_x e^{\beta x} = g(e^{\beta x}) = \prod_{k=0}^{n-1} (e^{k}+ e^{\beta}) \\
Z^{(\pm 1)} (\beta) &=\sum_x c_x e^{\pm \nu} e^{\beta x} = Z(\beta \pm \nu) \qquad \text{ for $\nu = 10 \eps/n$}. 
\end{align*}

For intuition, consider independent random variables $X_0, \dots, X_{n-1}$, wherein each $X_i$ is Bernoulli-$p_i$ for $p_i = \frac{e^{\beta}}{e^i + e^{\beta}}$. Then $\mu_{\beta}$ is the probability distribution on the sum $X = X_0 + \dots + X_{n-1}$, and distributions $\mu^{(\pm 1)}$ are derived by slightly shifting the probabilities $p_i$.

\begin{observation}
\label{zxeqn}
For any values $\beta, x$ we have $|z(\beta+x) - z(\beta) - z'(\beta) x| \leq 2 x^2$.
\end{observation}
\begin{proof}
We bound the second derivative of function $z$ as 
$$
z''(\beta) = \sum_{k=0}^{n-1} \frac{e^k e^{\beta}}{(e^k + e^{\beta})^2} \leq \sum_{k \leq \beta} \frac{e^k e^\beta}{e^{2 \beta}} +\sum_{k \geq \beta} \frac{e^k e^{\beta}}{e^{2 k}}  \leq  \frac{e}{e-1} + \frac{e}{e-1} \leq 4.
$$
 The result then follows from Taylor's theorem.
\end{proof}

\begin{lemma}
\label{lb-lemma}
For $n \geq 7$ and $\eps < 1/100$, the following bounds hold: \\
\noindent (a) $z(\betamin, \betamax) \leq n^2$. \\
(b) $|z^{(r)}(\betamin,\betamax) - z(\betamin, \betamax) | > 2 \eps$ for $r = \pm 1$. \\
(c) $\Psi \leq O( \eps^2 / n^2 ).$  \\
(d) $\mu_{\betamin} (0 )  \geq 1/5$ and $\mu_{\betamax}(n) \geq 1/5$.
\end{lemma}
\begin{proof}
For (a), we calculate: 
$$
z (0,n) = \sum_{k=0}^{n-1} \log( e^k + e^{n}) - \sum_{k=0}^{n-1} \log(  e^k + 1) = \sum_{k=0}^{n-1} \log \bigl( \frac{ e^k + e^{n}}{e^k + 1} \bigr) \leq \sum_{k=0}^{n-1} \log \bigl( \frac{e^{n}}{1} \bigr) = n^2.
$$

For (b), we can use Observation~\ref{zxeqn} and the formula $z^{(\pm 1)}(\beta) = z(\beta \pm \nu)$ to get:
\begin{align*}
&|z^{(\pm 1)} (\betamin, \betamax) - z (\betamin, \betamax)| = \bigl| \bigl( z(n \pm \nu) - z(n) \bigr) - \bigl( z(\pm \nu) - z(0) \bigr) \bigr| \\
& \qquad \qquad  \geq (z'(n) \nu - 2 \nu^2) - (z'(0) \nu + 2 \nu^2)  = (z'(n) - z'(0)) \nu - 4 \nu^2
\end{align*}

We have $z'(n)= \sum_{k=0}^{n-1} \frac{e^{n}}{e^k + e^{n}} \geq \sum_{k=0}^{n-1} \frac{e^n}{e^n + e^n} = n/2$ and $z'(0) = \sum_{k=0}^{n-1} \frac{1}{1 + e^k} \leq \sum_{k=0}^{\infty} e^{-k} = \frac{e}{e-1}$. So $|z^{(+\pm)} (\betamin, \betamax) - z (\betamin, \betamax)| \geq (n/2 - \frac{e}{e-1}) \nu - 4 \nu^2$, which is larger than $2 \eps$ for $n \geq 7, \eps < 1/100$. 
 
For (c), we have $\log \frac{Z^{(+1)}(\beta) Z^{(-1)}(\beta)}{Z(\beta)^2} = (z(\beta + \nu) - z(\beta)) + (z(\beta - \nu) - z(\beta))$. By Observation~\ref{zxeqn}, this is at most $(z'(\beta) \nu + 2 \nu^2) + (-z'(\beta) \nu + 2 \nu^2) = 4 \nu^2 = O( \eps^2 / n^2)$.

For (d), we calculate: 
\begin{align*}
\mu_{\betamin}(0) &= \frac{c_0}{Z(0)} =  \prod_{k=0}^{n-1} \frac{e^k}{e^k + 1}  \geq \prod_{k=0}^{n-1} \exp( -e^{-k}) \geq \exp( -\sum_{k=0}^{\infty} e^{-k} ) = e^{-e/(e-1)} \geq 1/5. \\
\mu_{\betamax}(n)  &= \frac{c_{n}}{Z(n)} = \prod_{k=0}^{n-1} \frac{ e^{n} }{  e^k + e^{n}} \geq \prod_{k=0}^{n-1} \exp( -e^{-(n - k)}) \geq e^{-e/(e-1)} \geq 1/5. \qedhere
\end{align*}
\end{proof}

\begin{proposition}
For $n > n_0, q > q_0, \eps < 1/100$,  where $n_0, q_0$ are absolute constants, the instance can be constructed so that both $\Pratio{\eps}$ or $\Ptcoef{1/10, \eps}$ require cost $\Omega( \frac{  \min \{q, n^2 \} \log \frac{1}{\gamma}}{\eps^2})$.
\end{proposition}
\begin{proof}
Here  $\kappa =n \nu = 10 \eps$. By reducing $n$ as needed, and taking $n_0, q_0$ sufficiently large,  we may suppose that $100 < n < \sqrt{q}/10$. In this case, we use Corollary~\ref{gtt5h} for both problems;  with $z(\betamin, \betamax) + 2 \kappa \leq n^2 + 1/5 \leq q$ for $q_0$ sufficiently large.

For $\Pratio{\eps}$ we have $|z(\betamin,\betamax) - z^{(\pm 1)}(\betamin,\betamax)| > 2 \eps$ as required. For  $\Ptcoef{1/10, \eps}$ take witnesses $x = 0, y = n$;  we have $\Delta(x), \Delta(y) \geq 1/5 \geq 1/10 \cdot e^{2 \kappa}$ and $| \lambda^{(1)}_x - \lambda^{(1)}_y| = n \nu = 10 \eps$ as required. So both problems require cost $\Omega( \frac{ \log(1/\gamma) }{\Psi} ) = \Omega( \frac{ \log(1/\gamma) n^2}{\eps^2})$.
\end{proof}

Combined with Proposition~\ref{pc1}, this shows Theorem~\ref{main-lb-thm}(b,c).
\subsection{Bounds for $\Ptcoef{}$  in the general integer setting}
\label{sec93}
Let $m = \lfloor n/2 \rfloor$. In this construction, we take $d = m$ and $\betamin = 0, \betamax = m \log 2$. We define $c$ by:
\begin{align*}
c_{2 i} &= 2^{-i^2} && \text{ for $i = 0, \dots, m$} \\
c_{2 i - 1} &= 2^{i -  i^2} \cdot 8 \delta && \text{ for $i = 1, \dots, m$} 
\end{align*}
 For each $i = 1, \dots, m$,  we form $\lambda^{(i)}$ as
 $$
 \lambda^{(i)}_{2 i - 1} = 3 \eps, \qquad \qquad \lambda^{(i)}_x = 0  \text{ for all other values $x \neq 2 i - 1$}
 $$

Intuitively, for each $i= 1, \dots, m$, there is a small count $c_{2 i - 1}$ hidden between two larger counts $c_{2i - 2}, c_{2 i}$. 
If we sample from $\mu_{\beta}$ for $\beta \gg i \log 2$, then $\mu_{\beta}(2i)$ will drown out $\mu_{\beta}(2i-1)$, and likewise if $\beta \ll i \log 2$ then $\mu_{\beta}(2i-2)$ will drown out $\mu_{\beta}(2i-1)$. So the only way to precisely determine each $c_{2 i - 1}$ is to take roughly $\frac{1}{\delta \eps^2}$ samples of $\mu_{\beta}$ with $\beta \approx i \log 2$. Overall, solving $\Pcoef{\delta, \eps}$ requires roughly $\Omega(\frac{m}{\delta \eps^2})$ samples in total. 

\begin{observation}
\label{zbeqn5}
Suppose $\delta < 1/100$. Then for any $k \in \mathbb Z$, we have $Z( k \log 2) \leq 2^{k^2 + 2}$.
\end{observation}
\begin{proof}
Here we calculate
$$
Z (k \log 2) = \sum_{i=0}^{m} 2^{-i^2} e^{2 i (k \log 2)}  + 8 \delta \sum_{i=1}^{m} 2^{i - i^2} e^{(2 i-1) (k \log 2)} \leq \sum_{i=-\infty}^{\infty} 2^{-i^2 + 2 i k}  + 8 \delta \sum_{i=-\infty}^{\infty} 2^{i - i^2 + (2 i-1) k}.
$$

Terms in the first sum decay at rate at least $1/2$ away from the peak value $i = k$, while the terms in the second sum decay at rate at least $1/4$ from their peak values at $i = k, k+1$. So $Z ( \beta ) \leq ( 3 \cdot 2^{k^2}  + 8 \delta \cdot \frac{8}{3} 2^{k^2})$, which is less than $2^{k^2  + 2}$ for $\delta < 1/100$.
\end{proof}

\begin{proposition}
\label{vv1}
Suppose $n \geq 7, \eps < 1/100, \delta < 1/100$ for large enough constant $n_0$.  Then we have: \\
\noindent (a) $\Psi \leq O( \delta \eps^2 )$ \\
\noindent (b) $z (\betamin,\betamax) \leq n^2/4$ \\
\noindent (c) $\Delta(0) \geq 1/4$ \\
\noindent (d) $\Delta (2k - 1) \geq 2 \delta$ for all $k = 1, \dots, m$.
\end{proposition}
\begin{proof}
For any $\beta \geq 0$ and $i = 1, \dots, m$, we have:
\begin{align*}
&\frac{Z^{(+i)}(\beta) Z^{(-i)}(\beta) }{Z(\beta)^2} =  \bigl( 1 + \frac{ (e^{3 \eps} - 1) 2^{i - i^2} \cdot 8 \delta  e^{(2 i - 1) \beta} }{Z(\beta)} \bigr) \bigl( 1 + \frac{ (e^{-3 \eps} - 1) 2^{i - i^2} \cdot 8 \delta  e^{(2 i - 1) \beta} }{Z(\beta)} \bigr) \\
&\qquad \leq \exp \bigl(  \frac{ (e^{3 \eps} - 1) 2^{i - i^2} \cdot 8 \delta  e^{(2 i - 1) \beta} }{Z(\beta)} \bigr)  \cdot \exp \bigl( \frac{ (e^{-3 \eps} - 1) 2^{i - i^2} \cdot 8 \delta  e^{(2 i - 1) \beta} }{Z(\beta)} \bigr) \qquad \text{(since $1+x \leq e^x$)} \\
&\qquad = \exp \Bigl ( \frac{ 8 \delta (e^{3 \eps} + e^{-3 \eps} - 2)  ( 2^{i-i^2} e^{(2 i - 1) \beta})}{ Z(\beta)} \Bigr)
\end{align*}
Since $\eps < 1/100$ we have $e^{3 \eps} + e^{-3 \eps} - 2 \leq 10 \eps^2$. Also, $Z(\beta) \geq 2^{-i^2} e^{2 i \beta} + 2^{-(i-1)^2} e^{2(i-1) \beta}$ from counts $c_{ 2i}$ and $c_{ 2 i - 2}$. So, we can upper-bound the term $\sum_{r=1}^d \log \frac{Z^{(+r)}(\beta) Z^{(-r)}(\beta) }{Z(\beta)^2}$ in the definition of $\Psi$ by:
\begin{align*}
 80 \delta \eps^2  \sum_{i=1}^{m} \frac{ 2^{i-i^2} e^{(2 i - 1) \beta}}{2^{-i^2} e^{ 2i \beta} + 2^{-(i-1)^2} e^{2(i-1) \beta}}=80 \delta \eps^2 \sum_{i=0}^{m-1} \frac{ 2^{i} e^{-\beta}}{1 + (2^{i} e^{-\beta})^2/2}\leq O(\delta \eps^2).
\end{align*}

\smallskip

For (b), Observation~\ref{zbeqn5} gives $Z( \betamax) = Z(m \log 2) \leq 2^{m^2 + 4}$, and trivially we have $Z(\betamin) \geq c_0 = 1$. So $z(\betamin, \betamax) \leq (m^2 + 4) \log 2$, which is smaller than $n^2/4$ since $m \leq n/2$ and $n \geq 7$.

For (c), Observation~\ref{zbeqn5} gives $Z(0) \leq 4$. So $\Delta(0) \geq \mu_{0}(0) \geq 1/4$.

For  (d), we take $\beta = k \log 2$; Observation~\ref{zbeqn5} gives $\mu_{\beta}(2k - 1) = \frac{c_{2k-1} e^{(2 k - 1) \beta}}{Z(\beta)}  \geq \frac{ 8 \delta \cdot 2^{k^2}}{2^{k^2  + 2}}  = 2  \delta.$
\end{proof}

\begin{proposition}
For $n > n_0, q > q_0, \delta < 1/100, \eps < 1/100$,  where $n_0, q_0$ are absolute constants, the instance can be constructed so that $\Ptcoef{\delta, \eps}$ requires cost $\Omega( \frac{\min\{ \sqrt{q}, n \} \log (1/\gamma) }{\delta \eps^2})$.
\end{proposition}
\begin{proof}
Here $\kappa = 3 \eps$.  By reducing $n$ as needed, and taking $n_0, q_0$ sufficiently large, we may suppose that $n < \sqrt{q}/10$. In this case, we apply Corollary~\ref{gtt5h}(b), where $z(\betamin, \betamax) + 2 \kappa \leq n^2/4 + 6 \eps \leq q$ for $q_0$ sufficiently large.  For each instance $c^{(\pm i)}$, we use witnesses $x = 0, y = 2 i - 1$, where  $\Delta(x), \Delta(y) \geq 2 \delta \geq \delta e^{2 \kappa}$ for $\delta < 1/100$, and $| \lambda_x^{(i)} - \lambda_y^{(i)} | = 3 \eps$. So $\Ptcoef{\delta, \eps}$ has cost $\Omega ( \frac{ d \log (1/\gamma)}{\Psi} )\ = \Omega \bigl( \frac{n \log(1/\gamma)}{\delta \eps^2} \bigr)$.
\end{proof}

Combined with Theorem~\ref{main-lb-thm}(b) already shown, this shows Theorem~\ref{main-lb-thm}(a).

\subsection{Bounds for the continuous setting}
\label{sec94}
In the constructions of Sections~\ref{sec92} and \ref{sec93}, the complexity bounds have a term $\min\{q, n^2 \}$. We can modify these constructions to yield continuous-setting instances where the lower bound depends only on $q$ irrespective of $n$. Namely, given a family $c^{(i)}:  i =0, \pm 1, \dots, \pm d$ of integer-setting distributions on $[\betamin, \betamax]$, we construct a new family of continous-setting instances $\tilde c$ via
$$
\tilde c^{(i)}_x = c^{(i)}_{n(x-1)} \qquad \qquad \text{ for each $i, x$.}
$$
\begin{theorem}
\label{lblemma44}
Define $\tilde \beta_{\min} = \beta_{\min} n, \tilde \beta_{\max} = \betamax n$. The family $\tilde c^{(i)}: i = 0, \pm 1, \dots, \pm d$ of continuous-setting distributions on $[\tilde \beta_{\min}, \tilde \beta_{\max}]$ has the following properties: \\
\noindent (a) Each distribution $\tilde c^{(i)}$ is supported on $x \in [1,2]$,  i.e. it has $\tilde n = 2$. \\
\noindent (b) Any algorithm to solve $\Pratio{\eps}$ or $\Ptcoef{\delta, \eps}$ for the instances $\tilde c$ yields a corresponding algorithm to solve the same problem for $c$, with the same complexity. \\
(c) The instances $\tilde c$ have $\tilde z^{(i)}(\tilde \beta_{\min}, \tilde \beta_{\max}) \leq z^{(i)}(\betamin, \betamax) + n(\betamax - \betamin)$ for each $i$.
\end{theorem}
\begin{proof}
Note that $\tilde \mu_{\beta n}^{(i)}(1 + x/n) = \mu_{\beta}^{(i)}(x)$ for all $\beta$, and $\tilde Z^{(i)}(\beta)  = \sum_x c^{(i)}_{n(x-1)} e^{\beta x} = \sum_y c_y^{(i)} e^{\beta (1 + y/n)} = e^{\beta} Z^{(i)} (\beta/n).$ So $\tilde Q^{(i)} (\tilde \beta_{\max}) = e^{n(\betamax - \betamin)} Q(\betamax)$. We can simulate oracle calls to $\tilde c$ at input $\tilde \beta$ with corresponding oracle calls to $c$ at $\beta = \tilde \beta / n$. Thus, solving $\Pratio{\eps}$ or $\Ptcoef{\delta, \eps}$ on  $\tilde c$ is equivalent to solving them on the original instances $c$. 
\end{proof}

To get the continuous-setting lower bound $\Omega(\frac{q} {\eps^2} \log \tfrac{1}{\gamma})$ for $\Pratio{\eps}$ and $\Ptcoef{\delta, \eps}$, for sufficiently large $q$, we construct the problem instances $c$ from Section~\ref{sec92} for value $n' = \lfloor \sqrt{q}/2 \rfloor, q' =  q/2 $. By Lemma~\ref{lb-lemma} and Theorem~\ref{lblemma44}, we have $\tilde z^{(i)}(\tilde \beta_{\min}, \tilde \beta_{\max}) \leq  n' (\betamax - \betamin) + z^{(i)} (\betamin, \betamax) \leq q$. 

Similarly, for the continuous-setting lower bound $\Omega(\frac{\sqrt{q}/\delta }{ \eps^2} \log \tfrac{1}{\gamma})$ for $\Ptcoef{\delta, \eps}$,  for sufficiently large $q$, we construct the problem instances $c$ of Section~\ref{sec93} for value $n' = \lfloor \sqrt{q}/10 \rfloor, q'= q/2$. For $m' = \lfloor n'/2 \rfloor$, we have $n (\betamax - \betamin) + z^{(i)}(\betamin, \betamax) \leq n' (m' \log 2 - 0) + q/2 \leq q$.

This concludes the proof of Theorems~\ref{main-lb-thm}(d,e).

%%%%%%%%%%%%%%%%%%%%%%%%%%%%%%%%%%%%%%%%%%%%%%%%%%%%%%%%%%%%%%%%%%%%%%%%%%%%%%%%%%%%%%%%%%%%%%%%%%%%%%%%%%%%%%%%%%%%%%%%%%%%%%%%%%%%%
%%%%%%%%%%%%%%%%%%%%%%%%%%%%%%%%%%%%%%%%%%%%%%%%%%%%%%%%%%%%%%%%%%%%%%%%%%%%%%%%%%%%%%%%%%%%%%%%%%%%%%%%%%%%%%%%%%%%%%%%%%%%%%%%%%%%%
%%%%%%%%%%%%%%%%%%%%%%%%%%%%%%%%%%%%%%%%%%%%%%%%%%%%%%%%%%%%%%%%%%%%%%%%%%%%%%%%%%%%%%%%%%%%%%%%%%%%%%%%%%%%%%%%%%%%%%%%%%%%%%%%%%%%%
%%%%%%%%%%%%%%%%%%%%%%%%%%%%%%%%%%%%%%%%%%%%%%%%%%%%%%%%%%%%%%%%%%%%%%%%%%%%%%%%%%%%%%%%%%%%%%%%%%%%%%%%%%%%%%%%%%%%%%%%%%%%%%%%%%%%%
%%%%%%%%%%%%%%%%%%%%%%%%%%%%%%%%%%%%%%%%%%%%%%%%%%%%%%%%%%%%%%%%%%%%%%%%%%%%%%%%%%%%%%%%%%%%%%%%%%%%%%%%%%%%%%%%%%%%%%%%%%%%%%%%%%%%%
%%%%%%%%%%%%%%%%%%%%%%%%%%%%%%%%%%%%%%%%%%%%%%%%%%%%%%%%%%%%%%%%%%%%%%%%%%%%%%%%%%%%%%%%%%%%%%%%%%%%%%%%%%%%%%%%%%%%%%%%%%%%%%%%%%%%%
%%%%%%%%%%%%%%%%%%%%%%%%%%%%%%%%%%%%%%%%%%%%%%%%%%%%%%%%%%%%%%%%%%%%%%%%%%%%%%%%%%%%%%%%%%%%%%%%%%%%%%%%%%%%%%%%%%%%%%%%%%%%%%%%%%%%%
%%%%%%%%%%%%%%%%%%%%%%%%%%%%%%%%%%%%%%%%%%%%%%%%%%%%%%%%%%%%%%%%%%%%%%%%%%%%%%%%%%%%%%%%%%%%%%%%%%%%%%%%%%%%%%%%%%%%%%%%%%%%%%%%%%%%%

\appendix

\section{Proofs for error bounds of $\Pcoef{\delta, \eps}$}
\label{app:approx-mualpha}

\noindent \textbf{Proof of Theorem~\ref{approx-mualpha}.}
Note that $\mu_{\alpha}(x) = \pi(x) e^{\alpha x} / Q(\alpha)$. If $\pi(x) = 0$, then the guarantee of $\Pcoef{\delta,\eps}$ directly gives $\hat \mu_{\alpha}(x) = 0 = \mu_{\alpha}(x)$. Otherwise, we can use the guarantees of $\Pcoef{\delta,\eps}$ and $\PratioAll{\eps}$ to show the upper bound on $\hat \mu_{\alpha}(x)$:
\begin{align*}
\hat \mu_{\alpha}(x) &\leq  \frac{ e^{\eps} \pi(x) (1 + \eps (1 + \delta/\Delta(x))) e^{\alpha x}}{Q(\alpha)} = e^{\eps} \mu_{\alpha}(x) (1 + \eps (1 + \delta/\Delta(x))) \\
&\leq e^{\eps} \mu_{\alpha}(x) (1 + \eps (1 + \delta/\mu_{\alpha}(x)))  \qquad \qquad  \qquad \qquad \qquad (\text{$\mu_{\alpha}(x) \leq \Delta(x)$ by definition}) \\
&= e^{\eps} ( 1 + \eps) \mu_{\alpha}(x) + e^{\eps} \delta \leq \mu_{\alpha}(x) + 3 \eps (\mu_{\alpha}(x) + \delta)  \qquad \   (\text{straightforward calculus for $\eps < 1/2$})
\end{align*}

The lower bound is completely analogous.

\bigskip

\noindent \textbf{Proof of Theorem~\ref{th:main:one-restate2}.} Let us show $\hat Q(\alpha \mid \mathcal D) \leq e^{\eps} Q(\alpha)$ for all $\alpha$; the lower bound $\hat Q(\alpha \mid \mathcal D) \geq  e^{-\eps} Q(\alpha)$ is completely analogous. For each value $i$ with $c_i \neq 0$, the guarantee of problem $\Pcoef{1/n, \eps/3}$ gives $\hat \pi(i) \leq \pi(i) + \eps/3 \cdot \pi(i) \bigl( 1 + \frac{1}{n \Delta(i)} \bigr)$. Since also $\hat \pi(i) = \pi(i)$ when $c_i = 0$, we can calculate
  \begin{equation}
  \label{qeqn555}
    \hat Q(\alpha \mid \mathcal D) \leq \sum_i \pi(i) (1 + \eps/3) e^{\alpha i}  + \eps/3 \sum_{i: c_i \neq 0} \frac{\pi(i) e^{\alpha  i}}{n \Delta(i)}
   \end{equation}

Note that $Q(\alpha) = \sum_i \pi(i) e^{\alpha i}$.  Thus, the first term in the RHS of (\ref{qeqn555}) is $(1 + \eps/3) Q(\alpha)$. We also observe that $\Delta(i) \geq \mu_{\alpha}(i) = \frac {\pi(i) e^{\alpha i}}{Q(\alpha)}$, so $\sum_{i: c_i \neq 0} \frac{\pi(i) e^{\alpha  i}}{n \Delta(i)} \leq     \sum_{i: c_i \neq 0} \frac{Q(\alpha)}{n} \leq Q(\alpha)$.  Overall we have $\hat Q(\alpha \mid \mathcal D) \leq (1 + \eps/3) Q(\alpha) + \eps/3 \cdot Q(\alpha) \leq e^{\eps} Q(\alpha)$ as desired. 
  
\section{Correctness with approximate oracles}\label{sec:approx}

Here, we define the total variation distance $||\cdot||_{TV}$ by $||\tilde\mu_\beta-\mu_\beta||_{TV}=\tfrac{1}{2} \sum_{x}|\tilde\mu_\beta(x)-\mu_\beta(x)|$.

\begin{theorem}\label{th:approx}
Suppose that sampling procedure $\A$ has cost $T$ and that  $||\tilde\mu_\beta-\mu_\beta||_{TV}\le \rho$ for a known parameter $\rho$. Let $\tilde\A$ be the algorithm obtained from $\A$ as follows: (i) replace calls $x \sim\mu_\beta$
with calls $x\sim\tilde\mu_\beta$; (ii) terminate the algorithm after $1/\rho$ steps and return arbitrary answer.

  Then $\tilde\A$ has cost $O(T)$, and the total variation distance of the outputs  of $\mathfrak A$ and $\mathfrak A'$ is at most $2 \rho T$. 
\end{theorem}
\begin{proof}
The distributions $\mu_\beta$ and $\tilde\mu_\beta$ can be maximally coupled so that 
samples $x \sim \mu_\beta$ and $x\sim\tilde\mu_\beta$ are identical with probability at least $1-\rho$.
Assume the $k^{\text{th}}$ call to $\mu_\beta$ in $\A$ is coupled with the $k^{\text{th}}$ call to $\tilde\mu_{\tilde\beta}$ in $\tilde{\A}$
when both calls are defined and $\beta=\tilde\beta$. % (for $k\le\tfrac 1\delta$).
We say the $k^{\text{th}}$ call is good if either (i) both calls
are defined and the produced samples are identical, or (ii)  $\A$ has terminated earlier.
Note that $\P[\mbox{$k^{\text{th}}$ call is good}\:|\:\mbox{all previous calls were good}]\ge 1-\rho$, since the conditioning 
event implies $\beta=\tilde\beta$ (assuming the calls are defined).

Let $A$ and $\tilde A$ be the number of calls to the sampling oracle by algorithms $\A$ and $\tilde \A$, respectively. By assumption, $\E[A] = T$. We say that the execution is {\em good} if three events hold:
\begin{enumerate}
\item[$\mathcal E_1$:] All calls are good. The union bound gives 
$\Pr{\mathcal E_1 \mid A=k}\ge 1-\rho k$, and therefore
$$
\Pr{\mathcal E_1}= \sum_{k=0}^{\infty}\Pr{A=k}\cdot \Pr{\mathcal E_1 \mid A=k} \geq \sum_{k=0}^{\infty}\Pr{A=k}\cdot \left(1-\rho  k\right) = 1- \rho \E[A] \geq 1 - \rho T
$$
\item[$\mathcal E_2$:] The number of oracle calls by $\A$ is at most $1/\rho$. By Markov's inequality, this has probability at least $1 - \frac{\E[A]}{1/\rho} \geq 1- \rho T$.
\end{enumerate}

If these two events occur, then $\tilde\A$ returns the same output as $\A$; by the union bound, this has probability at least $1-2 \rho T$. Thus, $\tilde \A$ and $\A$ have total variation distance at most $2 \rho T$. 

To bound $\E[\tilde A]$, observe that $\tilde A \leq 1/\rho$, while if $\mathcal E_1$ occurs we have $A = \tilde A$. So, taking expectations, we have $\mathbb E[ \tilde A ] \leq \mathbb E[ A ] + \frac{1 - \Pr{\mathcal E_1}}{\rho} \leq T + \frac{\rho T}{\rho} = 2 T$.
\end{proof}

\section{Proofs for statistical sampling}
\label{ppe-app}
%%%%%%%%%%%%%%%%%%%%%%%%%%%%%%%%%%%%%%%%%%%%%%%%%%%%%%%%%%%%%%%%%%%%%%%%%%%%%%%%%%%%%%%%%%%%%%%%%%%%%%%%%%%%%

\noindent \textbf{Proof of Lemma~\ref{binom:succ:lem}.}
When $p \geq \bar p$ we have $F(N p, N \eps(p + \bar p)) \leq F( N p, N \eps p) \leq 2e^{-N p \eps^2/3}$; for $N$ larger than the stated bound and $p \geq \bar p$ this is at most $\gamma/2$. When $p < \bar p$, we have $F (N p, N \eps (p + \bar p)) \leq F(N p, N \eps \bar p)$; by monotonicity this is at most $F(N \bar p, N \eps \bar p )\leq \gamma/2$.

For the second bound, we use standard estimates $F_{+}(N p, N x) \leq \exp( \frac{- N x^2}{2(p+x)} )$ with $x = p (e^{\eps} - 1)$ and $F_{-}(Np, N x)  \leq \exp(\frac{-N x^2}{2 p})$ with $x = p (1 - e^{-\eps})$. These show that, for $N$ larger than the stated bound and $p \geq e^{-\eps} \bar p$, the upper and lower deviations have probability at most $\gamma/4$. When $p \leq e^{-\eps} \bar p$, we use the monotonicity property $F_{+}(N p, N ( \bar p - p) ) \leq F_{+}(N e^{\eps} \bar p, N ( \bar p - e^{-\eps} \bar p))$. 

\bigskip

\noindent \textbf{Proof of Lemma~\ref{general-pcoef-lemma}.} For brevity, let $a = \mu_{\alpha}(x), \hat a = \hat \mu_{\alpha}(x), \eta = Q(\alpha) e^{- \alpha x}, \hat \eta = \hat Q(\alpha) e^{ - \alpha x}$ and $y = \delta/\Delta(x)$. If $a = 0$, then $\hat a = 0$ since it is an empirical estimate and hence $\hat \pi(x) = 0$ as desired. So suppose $a > 0$.  For the upper bound, we calculate 
  \begin{align*}
 \frac{ \hat \pi(x) - \pi(x) }{\pi(x)}  &= \frac{ \hat \eta  \hat a }{ \eta  a} - 1 \leq \frac{ \eta e^{\eps/3} \cdot (a + \eps a/3 \cdot  (1 + y))}{\eta a}  - 1  & \text{(by (A1) and (A2))} \\    
 &= \bigl( e^{ \eps/3} \cdot (1 + \eps/3)  - 1 \bigr) + \bigl( 1/3 \cdot e^{\eps/3} \bigr) \eps y \\
 &\leq \eps +  \eps y  = \eps (1 + \delta/\Delta(x)) & \text{(straightforward calculus)}
 \end{align*}
 
   The complementary inequality is completely analogous.
  
  \bigskip

\noindent \textbf{Proof of Theorem~\ref{ppe-est-thm}.} As a starting point, consider the following procedure: for each $i = 1, \dots, N$  draw $r   =  \lceil 100 \tau \rceil$ copies of  $X_i$ and compute the sample average $\overline {X_{i}}$; then, for any pair $i' \leq i$, set $Y_{i',i} = \prod_{\ell=i'}^i \overline {X_{\ell}}$.  We will show that, with  high probability, all values $Y_{1,i}$ are close $\mathbb E[ Y_{1,i} ]$.

We can first observe that, for any values $i, i'$, we have $\Pr{ Y_{i',i} / \E[Y_{i', i}] \in [e^{-\eps/2},e^{\eps/2}]} \geq 0.93$. For, let $Y = Y_{i',i}$ for brevity.  Since multiplying the variables $X_j$ by constants does not affect the claim, we may assume that $\E[X_j] = 1$. Each variable $\overline{X_{j}}$ is the mean of $r$ independent copies of $X_j$, so $\E[ \overline{X_{j}} ] = \E[X_j] = 1$ and $\text{Variance}[ \overline{X_{j}} ] = \text{Variance}[ X_j ]/r = \Vrel(X_j)/r$. So $\E[Y] = 1$, and we have:
$$
\E[Y^2]  =  \prod_{j = i'}^i (1 + \frac{\text{Variance}[X_{j}]}{r} )  \leq e^{\sum_{j=i'}^{i} \Vrel(X_{j})/ r} \leq e^{ (\tau \eps^2) / r}  \leq e^{\eps^2/100}
$$
Now by Chebyshev's inequality, we have $\Pr{ Y >  \notin [e^{-\eps/2}, e^{\eps/2} ]} \leq \Pr{|Y-1|> (1 - e^{-\eps/2})}\le \frac{\text{Variance}[Y]}{ (1 - e^{-\eps/2})^2} \leq \frac{ e^{\eps^{-2}/100} - 1}{  (1 - e^{-\eps/2})^2 }$. Straightforward calculations show that this is at most $0.07$ for $\eps \in (0,1)$.

Next, we claim that, with probability at least $0.92$ there holds $Y_{1,i} / \E[Y_{1,i}] \in [e^{- \eps},e^{ \eps}]$ for all $i$. For, let $\mathcal E$ be the bad event that $Y_{1,i}  \notin [e^{-\eps}, e^{\eps}]$ for some index $i$. Suppose we reveal random variables $\overline X_1, \overline X_2, \ldots$ in sequence; if $\mathcal E$ occurs, let $i$ be the first index with $Y_{1,i}  \notin [e^{-\eps}, e^{\eps}]$. For concreteness, suppose $Y_{1,i} > e^{\eps}$; the case $Y_{1,i} < e^{-\eps}$ is completely analogous. Random variables $\overline X_{i+1}, \dots, \overline X_N$ have not been revealed at this stage; thus,  $Y_{i+1, N}$ still has its original, unconditioned, probability distribution. As shown above, we have $Y_{i+1, N} \geq e^{-\eps/2}$ with probability at least $0.93$, in which case $Y_{1,i} Y_{i+1, N} = Y_{1,N} > e^{\eps} \cdot e^{-\eps/2} = e^{\eps/2}$. So $\Pr{ Y_{1,N} \notin [e^{-\eps/2}, e^{\eps/2}] \mid \mathcal E } \geq 0.93$.   On the other hand,  we have shown that $\Pr{ Y_{1,N} \notin [e^{-\eps/2}, e^{\eps/2}]} \leq 0.07$. These two bounds imply  $\Pr{ \mathcal E} \leq 0.07/0.93 \leq 0.08$.

To obtain the algorithm {\tt EstimateProducts}, we can run the above procedure for $k = O( \log \tfrac{1}{\gamma} )$ independent trials and return $$
\hat X^{\tt prod}_i = \text{median}( Y_{1,i}^{(1)}, \dots, Y_{1,i}^{(k)} )   \qquad \text{for $i = 1, \dots, N$} 
$$ where $Y_{1,i}^{(j)}$ is the value of statistic $Y_{1,i}$ in the $j^{\text{th}}$ trial.  With probability at least $1 - \gamma$ the bad-event $\mathcal E^{(j)}$ holds in fewer than $k/2$ trials $j$. In this case $\hat X^{\tt prod}$ satisfies the required condition for all $i$.

\section{The ${\tt Balance}$ subroutines}\label{sec:Balance}

 The starting point for the algorithm is a sampling procedure of Karp \& Kleinberg \cite{kk} for noisy binary search, which we summarize as follows:\footnote{The original algorithm in \cite{kk} only gives success probability $3/4$; it is well-known that it can be amplified to $1 -\gamma$ by checking the returned solution and restarting as needed.}
\begin{theorem}[\cite{kk}]
\label{kk-thm}
Suppose we can sample from Bernoulli random variables $X_0, \dots, X_N$, wherein each $X_i$ has some unknown mean $x_i$, and $0 \leq x_0 \leq x_1 \leq x_2 \leq \dots \leq x_N \leq 1$. Let us also write $x_{-1} = 0, x_{N+1} = 1$.  For any $\alpha, \nu, \gamma \in (0,1)$, there is a procedure which uses $O( \frac{\log( N/\gamma)}{\nu^2} )$ samples from the variables $X_i$ in expectation. With probability at least $1-\gamma$, it returns an index $V \in \{-1, 0, 1, \dots, N  \}$ such that $x_V - \nu \leq \alpha \leq x_{V+1} + \nu$.
\end{theorem}

Consider the non-decreasing function $p(\beta)=\mu_\beta [\chi, n] = \frac{ \sum_{x \geq \chi} c_x e^{\beta x}}{Z(\beta)}$. Our basic plan is to quantize the search space $[\betamin, \betamax]$ into  discrete points $u_0, \dots, u_N$, for some value $N = \poly(n,q)$, and then apply Theorem~\ref{kk-thm} for probabilities $x_i = p(u_i)$. Note that we can simulate a Bernoulli random variable of rate $x_i$ by drawing $y \sim \mu_{u_i}$ and checking if $y \geq \chi$. The full details are as follows:

~\vspace{-5pt} \\ \noindent \begin{minipage}{\textwidth}

{
\vspace{1pt}
\noindent \hspace{8.5pt}{\bf Procedure ${\tt Balance}(\chi, \tau, \gamma)$.} 
\vspace{-17pt}

\setlength{\interspacetitleruled}{7pt}
\begin{algorithm}[H]
\DontPrintSemicolon
\textbf{if $\chi \leq 0$ then return} $\beta = \betamin$  \\
 Set $\betamin' = \max\{ \betamin, \betamax - (q+1) \}$  and $N = \Bigl \lceil \frac{  n (q+1) }{2 \log \bigl( \frac{4 - 4 \tau}{3 - 2 \tau} \bigr)}  \Bigr \rceil$ \\
\textbf{for $j = 0, \dots, N$ set $u_j = (1 - j/N) \betamin' + (j/N) \betamax$} \\
Apply Theorem~\ref{kk-thm} for values $x_j = p(u_j)$ with $\alpha = \tfrac{1}{2}, \nu = (\tfrac{1}{2} - \tau)/2$; let $V$ be the return value. \\
\textbf{if $-1 <  V < N$ then return} $\beta = \frac{u_V + u_{V+1}}{2}$  \\
\textbf{else if $V = -1$ then return} $\beta = \betamin'$ \\
\textbf{else if $V = N$ \ then return} $\beta = \betamax$
\end{algorithm}
}
\end{minipage}

This algorithm has complexity $O( \log \frac{N}{\gamma} ) = O( \log \frac{nq}{\gamma} )$ for fixed $\tau$. It is clearly correct if $\chi \leq 0$ since then $\mu_{\betamin}[\chi, n] = 1$. For $\chi > 0$, we claim that if    $V$ satisfies  $x_V - \nu \leq \alpha \leq x_{V+1} + \nu$, which occurs with probability at least $1-\gamma$, then $\beta \in {\tt BalancingVals}(\chi, \tau)$.  There are a number of cases.
\begin{itemize}
\item Suppose that $-1 < V < N$. It suffices to show that $\tau \leq p( \beta) \leq 1 - \tau$.  We will show only the inequality $p(\beta) \leq 1 - \tau$; the complementary inequality is completely analogous.

The condition $x_V - \nu \leq \alpha$ implies that $p(u_V) \leq 1/2 + \nu$.  We have  $\beta - u_V = \frac{ \betamax - \betamin'  }{2 N}  \leq \frac{ \log \bigl( \frac{4 - 4 \tau}{3 - 2 \tau} \bigr) }{n}$, and so we get
$$
p(\beta) =  p(u_V) \cdot \frac{Z(u_V)}{Z(\beta)} \cdot \frac{\sum_{x \in [\chi, n]} c_x e^{\beta x}}{ \sum_{x \in [\chi, n]} c_x e^{u_V x} } \leq p(u_V) \cdot 1 \cdot e^{(\beta - u_V) n}  \leq (1/2 + \nu) \cdot \frac{4 - 4 \tau}{3  - 2 \tau} = 1-  \tau.
$$

\item Suppose that $V = -1$. Here, the condition $x_{V+1} + \nu \geq \alpha$ implies that $p(u_0) \geq 1/2 - \nu$. So for the lower bound on $p(\beta)$, we have $p(\beta) = p(u_0) \geq 1/2 - \nu \geq \tau$. 

We have either $\beta = \betamin$ or $\beta = \betamax - (q+1)$.  In the former case, we do not need to show the upper bound on $p(\beta)$.  In the latter case,  we have
$$
p(\beta) = p(\betamax) \cdot \frac{Z(\betamax)}{Z(\beta)} \cdot \frac{ \sum_{x \in [\chi, n]} c_x e^{\beta x}}{ \sum_{x \in [\chi, n]} c_x e^{\betamax x}}  \leq 1 \cdot \frac{ Z(\betamax)}{Z(\betamin)} \cdot e^{-(\betamax - \beta)}
$$
where the last inequality holds since $c_x = 0$ for $x \in (0,1)$ and $\chi > 0$.  Since $\beta = \betamax - (q+1)$, we conclude that $p(\beta) \leq \frac{ Z(\betamax)}{Z(\betamin)} \cdot e^{-(q+1)} \leq e^{-1} \leq 1/2$.

\item Suppose that $V = N$. Since $\beta = \betamax$, we only need to show the upper bound on $p(\beta)$. Here, the condition $x_V - \nu \leq \alpha$ implies $p(\beta) = p(u_V) \leq 1/2 + \nu \leq 1 - \tau$.
\end{itemize}

This concludes the proof of Theorem~\ref{th:subroutines2x}. We next describe the accelerated algorithm given access to a data structure for $\PratioAll{\eps}$ where $\eps = \frac{1/2 - \tau}{30}$. Here,  the search space is reduced to just $N = O(q)$.

~\vspace{-5pt} \\ \noindent \begin{minipage}{\textwidth}
{
\vspace{1pt}
\noindent \hspace{8.5pt}{\bf Procedure ${\tt BalancePreprocessed}(\chi, \tau, \mathcal D,\gamma)$ } 
\vspace{-17pt}

\setlength{\interspacetitleruled}{7pt}
\begin{algorithm}[H]
\DontPrintSemicolon
\textbf{if $\log \hat Q(\betamax \mid \mathcal D) > q + \eps$ then terminate} with ERROR. \\
set parameters $N = \Bigl  \lfloor \frac{\log \hat Q(\betamax \mid \mathcal D)}{10 \eps} \Bigr \rfloor$, and set $u_0 = \betamin, u_N = \betamax$ \\
\textbf{for $j = 1, \dots, N-1$} let $u_j \in [\betamin, \betamax]$ be an arbitrary value with $|\log \hat Q(u_j \mid \mathcal D) - 10 \eps j| \leq \eps$ \ \ \ \ \ \ \ \ \ \ \ \ \ \  \ \ \ \ \ \ \ \ \ \ \ \ \ \ \ \ \ \ \  \ \ \ \ \ \ \  \ \ \ \ \  \ (if no such value exists then \textbf{terminate} with ERROR).  \\
Apply Theorem~\ref{kk-thm} for values $x_j = p(u_j)$ with $\alpha = \tfrac{1}{2}, \nu = (\tfrac{1}{2} - \tau)/2$; let $V$ be the return value. \\
\textbf{if $-1 <  V < N$ then return} $\beta = u_{V+1}$  \\
\textbf{else if $V = -1$ then return} $\beta = \betamin$ \\
\textbf{else if $V = N$ \ then return} $\beta = \betamax$
\end{algorithm}
}
\end{minipage}

Due to the check at Line 1, the algorithm has complexity $O( \log \frac{q}{\gamma} )$ as desired. For correctness, we suppose that $\mathcal D$ is $\eps$-close. So the check at Line 1 does not terminate with error. Also, for each $j = 1, \dots, N - 1$ the value $u$ with $Q(u) = e^{10 \eps j}$ would work at Line 3; hence this check does not fail either. Furthermore, for every $j = 0, \dots, N - 1$ we have $| \log Q(u_j) - 10 \eps j | \leq 2 \eps$ and similarly $\log Q(u_N) - \log Q(u_{N-1})  \in [8 \eps, 22 \eps]$. So, for all $j = 0, \dots, N$, we have
\begin{equation}
\label{zuveqn}
6 \eps \leq z(u_j, u_{j+1})  \leq 22 \eps
\end{equation}
So $0 \leq x_0 \leq \dots \leq x_N \leq 1$ as required. To finish the proof, we claim that if  $V$ satisfies  $x_V - \nu \leq \alpha \leq x_{V+1} + \nu$, which occurs with probability at least $1-\gamma$, then $\beta \in {\tt BalancingVals}(\chi, \tau)$. 

\begin{itemize}
\item Suppose that $-1 < V < N$. The inequality $x_{V+1} + \nu \geq \alpha$ implies that $p(\beta) = p(u_{V+1}) \geq \alpha - \nu \geq \tau$. The inequality $x_V - \nu \leq \alpha$ implies that $p(u_V) \leq \alpha + \nu$. From Eq.~(\ref{zuveqn}), we can calculate
$$
1-p(\beta) = (1-p(u_V)) \cdot \frac{Z(u_V)}{Z(u_{V+1})} \cdot \frac{ \sum_{x < \chi} c_x e^{u_{V+1} x}}{\sum_{x < \chi} c_x e^{u_{V} x}} \geq (1 - \alpha - \nu) \cdot e^{-22 \eps} \cdot 1
$$
and straightforward calculations show that this is at least $\tau$.

\item Suppose that $V = -1$. The condition $x_{V+1} + \nu \geq \alpha$ implies that $p(\beta) = p(u_0) \geq 1/2 - \nu \geq \tau$. Since $\beta = \betamin$, we do not need to show the upper bound on $p(\beta)$.  

\item Suppose that $V = N$. The condition $x_V - \nu \leq \alpha$ implies that $p(\beta) = p(u_N) \leq 1/2 + \nu \leq 1 - \tau$. Since $\beta = \betamax$, we do not need to show the lower bound on $p(\beta)$.
\end{itemize}

This concludes the proof of Theorem~\ref{th:subroutines2x}.
\section{Some useful inequalities}
\label{sec:FindSegment}
\begin{lemma}\label{lemma:incrseq}
Let $b_1,\ldots,b_{m+1}$ be a strictly increasing positive sequence. Then we have $\sum_{i=1}^m \frac{b_i}{b_{i+1} - b_i} \geq \frac{m^2}{\log (b_{m+1}/b_1)} - m/2$.
\end{lemma}
\begin{proof}
Define $v = \log(b_{m+1}/b_1)$.  For $i = 1, \dots, m$ let $a_i = \log(b_{i+1}/b_i)$. Then $\sum_{i=1}^m  \frac{b_i}{b_{i+1} - b_i} =\sum_{i=1}^m  \frac{1}{e^{a_i} - 1}$, and by telescoping sums we have $\sum_i a_i = \log(b_2/b_1) + \log(b_3/b_2) + \dots + \log(b_{m+1}/b_m) = v$.

Jensen's inequality applied to the concave-up function $x \mapsto \frac{1}{e^x - 1}$ gives: $$
\sum_{i=1}^m \frac{1}{e^{a_i} - 1} \geq m \cdot \frac{1}{e^{\tfrac{1}{m} \sum_{i=1}^m a_i} - 1} = \frac{m}{e^{v/m} - 1}.
$$

By straightforward calculus, it can be verified that $\frac{1}{e^y - 1} \geq \frac{1}{y} - \frac{1}{2}$ for all $y > 0$.  Applied to $y = v/m$, we conclude that $\sum_{i=1}^m \frac{b_i}{b_{i+1} - b_i} \geq m \bigl( m/v - 1/2 \bigr)  = m^2/v - m/2$ as claimed.
\end{proof}

\begin{lemma}\label{lemma:logconcave:harmonic}
Let $b_1,\ldots,b_m$ be a log-concave sequence with $b_k\le \frac 1k$ for each $k$.
Then $b_1+\ldots+b_m <  e$.
\end{lemma}
\begin{proof}
Let $k \in \{1, \dots, m \}$ be chosen to maximize the value $k b_k$ (breaking ties arbitrarily). If $b_k = 0$, then due to maximality of $k$ we have $b_1 = \dots = b_m = 0$ and the result holds. Suppose that $1 < k < m$; the cases with $k = 1$ and $k = m$ are very similar. Define the sequence $y_1, \dots, y_m$ by:
$$
y_i=\begin{cases}
\frac{1}{k} (\frac{k-1}{k})^{i-k} & \mbox{if }i< k \\
\frac{1}{k} (\frac{k}{k+1})^{i-k} & \mbox{if }i\ge k
\end{cases}
$$
By maximality of $k$, we have $\frac{y_{k-1}}{y_k}=\frac {k}{k-1} \ge \frac{b_{k-1}}{b_k}$
and $\frac{y_{k+1}}{y_k}=\frac k{k+1}\ge \frac{b_{k+1}}{b_k}$ and $y_k = \frac{1}{k} \geq b_k$. Since $\log y_i$ is linear for $i \leq k$ and for $i \geq k$, log-concavity of sequence $b$ gives $b_i\le y_i$ for $i = 1, \dots, m$. So
\begin{align*}
\sum_{i=1}^m b_i
&\leq  \sum_{i=1}^\infty y_i = 
 \sum_{i=1}^{k-1} \frac{1}{k} \Bigl( \frac{k-1}{k} \Bigr)^{i-k} + \sum_{i=k}^\infty \frac{1}{k} \Bigl( \frac{k}{k+1} \Bigr)^{i-k}  =  \Bigl( 1 - \frac{1}{k} \Bigr)^{1 - k}+ \frac{1}{k}
\end{align*}

Now consider the function $g(x) = (1 - x)^{1 - 1/x} + x$. We have shown $\sum_i b_i \leq g(1/k)$. To finish the proof, it suffices to show $g(x) < e$ for all $x \in (0,\tfrac12)$. This follows from some standard analysis showing that $\lim_{x \rightarrow 0} g(x) = e$ and $\lim_{x \rightarrow 0} g'(x) = 1 - e/2 < 0$ and $g''(x) < 0$ in the interval $(0, \tfrac{1}{2})$. 
\end{proof}

\section*{Acknowledgments}
We thank Heng Guo for helpful explanations of algorithms for sampling connected
subgraphs and matchings, Maksym Serbyn for bringing to our attention the Wang-Landau algorithm and its use in physics.
The author Vladimir Kolmogorov was supported by the European Research Council under the European Unions Seventh Framework Programme (FP7/2007-2013)/ERC grant agreement no 616160.

\small
\bibliographystyle{plain}
\bibliography{sampling}

\begin{thebibliography}{10}

\bibitem{Adiprasito:18}
K.~Adiprasito, J.~Huh, and E.~Katz.
\newblock Hodge theory for combinatorial geometries.
\newblock {\em Ann. Math.}, 188(2):381--452, 2018.

\bibitem{anari2018log}
Nima Anari, Kuikui Liu, Shayan~Oveis Gharan, and Cynthia Vinzant.
\newblock Log-concave polynomials {II}: High-dimensional walks and an {FPRAS}
  for counting bases of a matroid.
\newblock {\em Ann. Math.}, 199(1):259--299, 2024.

\bibitem{BelardinelliPereyra:JCP127:2007}
R.~E. Belardinelli and V.~D. Pereyra.
\newblock {W}ang-{L}andau algorithm: A theoretical analysis of the saturation
  of the error.
\newblock {\em J. Chem. Phys.}, 127(18):184105, 2007.

\bibitem{Bezakova08}
I.~Bez\'akov\'a, D.~\v{S}tefankovi\v{c}, V.~V. Vazirani, and E.~Vigoda.
\newblock Accelerating simulated annealing for the permanent and combinatorial
  counting problems.
\newblock {\em SIAM J. Comput.}, 37:1429--1454, 2008.

\bibitem{Braenden:2015}
P.~Br{\"a}nd{\'e}n.
\newblock Unimodality, log-concavity, real–rootedness and beyond.
\newblock In {\em Handbook of Enumerative Combinatorics}, chapter~7, pages
  438--483. CRC Press, 2015.

\bibitem{carlson2022computational}
Charlie Carlson, Ewan Davies, Alexandra Kolla, and Will Perkins.
\newblock Computational thresholds for the fixed-magnetization {I}sing model.
\newblock In {\em Proc. 54th ACM Symposium on Theory of Computing (STOC)},
  pages 1459--1472, 2022.

\bibitem{chen}
Xiaoyu Chen, Heng Guo, Xinyuan Zhang, and Zongrui Zou.
\newblock Near-linear time samplers for matroid independent sets with
  applications.
\newblock {\em arXiv preprint arXiv:2308.09683}, 2023.

\bibitem{chen2021optimal}
Zongchen Chen, Kuikui Liu, and Eric Vigoda.
\newblock Optimal mixing of {G}lauber dynamics: Entropy factorization via
  high-dimensional expansion.
\newblock {\em SIAM J. Comput.}, pages STOC21--104, 2023.

\bibitem{davies_et_al}
Ewan Davies and Will Perkins.
\newblock Approximately counting independent sets of a given size in
  bounded-degree graphs.
\newblock {\em SIAM J. Comput.}, 52(2):618--640, 2023.

\bibitem{Fort:15}
G.~Fort, B.~Jourdain, E.~Kuhn, T.~Leli\'evre, and G.~Stoltz.
\newblock Convergence of the {W}ang-{L}andau algorithm.
\newblock {\em Mathematics of Computation}, 84(295):2297--2327, 2015.

\bibitem{GuoHe:18}
H.~Guo and K.~He.
\newblock Tight bounds for popping algorithms.
\newblock {\em Random Struct. Algor.}, 57(2):371--392, 2020.

\bibitem{GuoJerrum:ICALP18}
H.~Guo and M.~Jerrum.
\newblock A polynomial-time approximation algorithm for all-terminal network
  reliability.
\newblock {\em SIAM J. Comput.}, 48(3):964--978, 2019.

\bibitem{hl72}
O.~J. Heilmann and E.~H. Lieb.
\newblock Theory of monomer-dimer systems.
\newblock In {\em Statistical Mechanics}, pages 45--87. Springer, 1972.

\bibitem{Huber:Gibbs}
M.~Huber.
\newblock Approximation algorithms for the normalizing constant of {G}ibbs
  distributions.
\newblock {\em Ann. Appl. Probab.}, 25(2):974--985, 2015.

\bibitem{TPA}
M.~Huber and S.~Schott.
\newblock Using {TPA} for {B}ayesian inference.
\newblock {\em Bayesian Statistics 9}, pages 257--282, 2010.

\bibitem{TPA:journal}
M.~Huber and S.~Schott.
\newblock Random construction of interpolating sets for high-dimensional
  integration.
\newblock {\em J. Appl. Prob.}, 51:92--105, 2014.

\bibitem{jain2022approximate}
Vishesh Jain, Will Perkins, Ashwin Sah, and Mehtaab Sawhney.
\newblock Approximate counting and sampling via local central limit theorems.
\newblock In {\em Proc. 54th ACM Symposium on Theory of Computing (STOC)},
  pages 1473--1486, 2022.

\bibitem{jerrum1989approximating}
M.~Jerrum and A.~Sinclair.
\newblock Approximating the permanent.
\newblock {\em SIAM J. Comput.}, 18(6):1149--1178, 1989.

\bibitem{jerrum1996markov}
M.~Jerrum and A.~Sinclair.
\newblock The {M}arkov {C}hain {M}onte {C}arlo method: an approach to
  approximate counting and integration.
\newblock {\em Approximation algorithms for NP-hard problems}, pages 482--520,
  1996.

\bibitem{kk}
R.~M. Karp and R.~Kleinberg.
\newblock Noisy binary search and its applications.
\newblock In {\em Proc.~18th ACM-SIAM Symposium on Discrete Algorithms (SODA)},
  pages 881--890, 2007.

\bibitem{Kingman:PPP}
J.~F.~C. Kingman.
\newblock {\em Poisson Processes}.
\newblock Clarendon Press, 1992.

\bibitem{Kolmogorov:COLT18}
V.~Kolmogorov.
\newblock A faster approximation algorithm for the {G}ibbs partition function.
\newblock {\em Proceedings of Machine Learning Research}, 75:228--249, 2018.

\bibitem{ojeda}
Pedro Ojeda, Martin~E Garcia, Aurora Londo{\~n}o, and Nan-Yow Chen.
\newblock {M}onte {C}arlo simulations of proteins in cages: influence of
  confinement on the stability of intermediate states.
\newblock {\em Biophys. J.}, 96(3):1076--1082, 2009.

\bibitem{Schur:JP1252:2019}
L.~N. Shchur.
\newblock On properties of the {W}ang-{L}andau algorithm.
\newblock {\em J. Phys.: Conference Series}, 1252, 2019.

\bibitem{Stefankovic:JACM09}
D.~\v{S}tefankovi\v{c}, S.~Vempala, and E.~Vigoda.
\newblock Adaptive simulated annealing: A near-optimal connection between
  sampling and counting.
\newblock {\em J. of the ACM}, 56(3) Article \#18, 2009.

\bibitem{WangLandau:PRL86:2001}
F.~Wang and D.~P. Landau.
\newblock Efficient, multiple-range random walk algorithm to calculate the
  density of states.
\newblock {\em Phys. Rev. Lett.}, 86(10):2050--2053, 2001.

\end{thebibliography}

\end{document}